\documentclass[10pt,journal,compsoc]{IEEEtran}

\ifCLASSOPTIONcompsoc
\usepackage[nocompress]{cite}
\else
\usepackage{cite}
\fi

\usepackage{graphicx,subfigure}

\usepackage{subfigure}
\usepackage{amsmath}
\usepackage{graphicx} % Required for including pictures
\usepackage{caption}
\usepackage{graphicx,subfigure}
\usepackage{multicol} % Required for splitting text into multiple columns
\usepackage{graphicx} % Required for including pictures
\usepackage{amssymb}
\usepackage{float} % Allows putting an [H] in \begin{figure} to specify the exact location of the figure
\usepackage{wrapfig} % Allows in-line images such as the example fish picture
\usepackage{booktabs}
\usepackage{lipsum} % Used for inserting dummy 'Lorem ipsum' text into the template
\usepackage{indentfirst}
\usepackage{bm}
\usepackage{color}

\usepackage{amsmath}
\usepackage{enumerate}
\usepackage{extarrows}
\usepackage{url}
\usepackage{graphicx}
\usepackage{epstopdf}
\usepackage{multirow}
\usepackage{tablefootnote}
\usepackage[flushleft]{threeparttable}
\usepackage{slashbox}
\usepackage{array}
\usepackage{multirow}
\usepackage{amsmath}
\usepackage{balance}

\usepackage{enumitem}

\usepackage{amsmath}
\usepackage{amsthm}

\usepackage{hyperref}
\hypersetup{hidelinks}

\usepackage[ruled,linesnumbered,vlined]{algorithm2e}

\usepackage{url}

\usepackage{cite}

\def\eq{\triangleq}
\def\ic{ \stackrel{ {\tiny \text{IC}}  }{\Longleftrightarrow}}
\def\ir{\succcurlyeq 0}

\def\Stheta{\Theta}
\def\Sbeta{\mathcal{B}}

\def\plan{\mathcal{T}}
\def\dcap{Q}
\def\mechanism{\kappa}
\def\pcap{\Pi}
\def\adfee{\pi}
\def\dmean{\bar{d}}
\def\dmax{D}

\def\val{\theta}
\def\cut{\beta}

\def\contract{\phi}
\def\Contrset{\Phi}

\def\profit{W}
\def\revenue{R}
\def\cost{C}

\def\payoff{\bar{S}}
\def\slope{\sigma}

\def\type{\Lambda}

\newtheorem{claim}{\textbf{Claim}}

\newtheorem{problem}{\textbf{Problem}}
\newtheorem{question}{\textbf{Question}}

\newtheorem{lemma}{\textbf{Lemma}}

\newtheorem{proposition}{\textbf{Proposition}}
\newtheorem{definition}{\textbf{Definition}}
\newtheorem{theorem}{\textbf{Theorem}}

\begin{document}

\title{Multi-Cap Optimization for Wireless Data Plans with Time Flexibility}

\author{Zhiyuan~Wang,
        Lin~Gao,~\IEEEmembership{{Senior Member,~IEEE,}}
        and~Jianwei~Huang,~\IEEEmembership{Fellow,~IEEE}
	\IEEEcompsocitemizethanks{
		\IEEEcompsocthanksitem Part of the results appeared in ACM MobiHoc 2018 \cite{Zhiyuan2018MobiHoc} and ACM NetEcon 2017 \cite{Zhiyuan2017netecon}.
		\IEEEcompsocthanksitem Zhiyuan Wang is with Department of Information Engineering, The Chinese University of Hong Kong, Shatin, N.T., Hong Kong, China. 
		E-mail:	wz016@ie.cuhk.edu.hk
		\IEEEcompsocthanksitem Lin Gao is with the School of Electronic and Information Engineering, Harbin Institute of Technology, Shenzhen, China. 
		E-mail: gaol@hit.edu.cn
		\IEEEcompsocthanksitem Jianwei Huang is with the School of Science and Engineering, The Chinese University of Hong Kong, Shenzhen, China, and Department of Information Engineering, The Chinese University of Hong Kong, Hong Kong, China.
		E-mail:	jwhuang@ie.cuhk.edu.hk
%		\IEEEcompsocthanksitem This work is supported by the General Research Fund CUHK 14219016 from Hong Kong UGC, and the Presidential Fund from the Chinese University of Hong Kong, Shenzhen, and by the National Natural Science Foundation of China (Grant No. 61771162).
	}
	
%\thanks{Manuscript received April 19, 2005; revised August 26, 2015.}
}

%% The paper headers
%\markboth{Journal of \LaTeX\ Class Files,~Vol.~14, No.~8, August~2015}%
%{Shell \MakeLowercase{\textit{et al.}}: Bare Demo of IEEEtran.cls for IEEE Journals}

\IEEEtitleabstractindextext{
\begin{abstract}
An effective way for a Mobile network operator (MNO) to improve its revenue is \textit{price discrimination}, i.e., providing different combinations of  data caps and subscription fees.
Rollover data plan (allowing the unused data in the current month to be used in the next month) is an innovative data mechanism with \textit{time flexibility}.
In this paper, we study the MNO's optimal multi-cap data plans with time flexibility in a realistic asymmetric information scenario.
Specifically, users are associated with multi-dimensional private information, and the MNO designs a contract (with different data caps and subscription fees) to induce users to truthfully reveal their private information.
This problem is quite challenging due to the multi-dimensional private information.
We address the challenge in two aspects. 
First, we find that a feasible contract (satisfying incentive compatibility and individual rationality) should allocate the data caps according to users' willingness-to-pay (captured by the slopes of users' indifference curves).
Second, for the non-convex data cap allocation problem, we propose a Dynamic Quota Allocation Algorithm, which has a low complexity and guarantees the global optimality.
Numerical results show that the time-flexible data mechanisms increase both the MNO's profit (25\% on average) and users' payoffs (8.2\% on average) under price discrimination.
\end{abstract}

% Note that keywords are not normally used for peerreview papers.
}

\maketitle

\IEEEdisplaynontitleabstractindextext

\IEEEpeerreviewmaketitle

% For peer review papers, you can put extra information on the cover
% page as needed:
% \ifCLASSOPTIONpeerreview
% \begin{center} \bfseries EDICS Category: 3-BBND \end{center}
% \fi
%
% For peerreview papers, this IEEEtran command inserts a page break and
% creates the second title. It will be ignored for other modes.
\IEEEpeerreviewmaketitle

\section{Introduction}  
\subsection{Background and Motivation}  
\IEEEPARstart{M}{obile} Network Operators (MNOs) profit from the wireless data services through carefully designing their wireless data plans.
The pricing strategy involved in the wireless data plans has evolved  from the flat-rate scheme to the usage-based scheme in the past years \cite{sen2013survey}.
Now the most widely used data plan consists of a monthly data cap, a monthly one-time subscription fee, and a linear price for any unit of additional data consumption beyond the data cap.
Based on this pricing strategy, MNOs usually offer multiple data caps together with different monthly subscription fees for users to choose from.
For example, in the US market, AT\&T charges \$20 for 300MB, \$45 for 1GB, \$55 for 2GB, and \$70 for 4GB; and the linear price for exceeding the data cap is \$15/GB \cite{ATT}.

The purpose of MNO's multi-cap offering is to capture more user surplus by differentiating users based on their preferences, also called  \textbf{price discrimination} in economics \cite{huang2013wireless}.
To make such a price discrimination scheme  work, the MNO must be able to identify the market segments by users' preferences that are usually  users' private information, and the MNO needs to enforce the scheme through some incentive mechanism.
For example, the MNO may want to offer a larger monthly data cap with a larger monthly subscription fee to  businessmen, who have a stronger ability to pay and a relatively inelastic data demand comparing with other consumers (such as students).
%; and a lower subscription fee for students who have a weaker ability to pay and a relatively elastic demand.
%Intuitively, the MNO can enforce such a price discrimination by offering the businessmen a larger data cap with a higher subscription fee.
However, it is a very challenging problem to induce users to truthfully reveal their private preferences in practice, especially when users have multi-dimensional private preferences.  
This motivates us to ask the first key question in this paper.
\begin{question}
	How should the MNO optimize the multi-cap data plan offering?
\end{question}

Recently the growing market competition forces the MNOs to explore various innovations on their mobile data plans.
For example, the rollover data plan enables users to enjoy the \textbf{time flexibility} over their data consumptions, by allowing the unused data from the previous month to be used in the current month.  
Such a rollover mechanism is attractive to users, as a user's data demand is often stochastic and  the rollover mechanism helps users balance the possible \textit{data waste} within the data cap and the possible \textit{overage usage} when consuming beyond the data cap. 

Although based on the same rollover principle, different  rollover data plans are different in terms of the consumption priority between the rollover data and the monthly data cap.
For example, the rollover data plan offered by AT\&T requires that the rollover data from the previous month should be consumed after the current monthly data cap \cite{ATTrollover}, while China Mobile requires the other way around \cite{CMrollover}.
In our previous work \cite{Zhiyuan2018TMC,ZhiyuanCompetitionTMC}, we analyzed the MNO's optimal data plan with time flexibility under the single-cap scheme (without price discrimination) and found that the time flexibility can increase both the MNO's profit and users' payoff, hence improve the social welfare.
This motivates us to ask the second key question in this paper.
\begin{question}
	What is the impact of time flexibility under the multi-cap scheme?
\end{question}

In this paper, we will study the MNO's price discrimination through the multi-cap data plans, taking into account the time-flexible data mechanisms.
%We hope the results in this paper can provide the MNO with valuable engineering insights on the optimal multi-cap design, and the public with the impact of time flexibility under the more practical multi-cap scheme.
%Meanwhile, this paper also serves as a highly nontrivial extension to our previous work  focusing on the single-cap scheme.

\subsection{Solutions and Contributions} 
We study how the MNO optimizes its multi-cap data plans under different data mechanisms with time flexibility.
Specifically, we consider an asymmetric information scenario, where the users' preferences for the wireless data plans are private and multi-dimensional.
We formulate this problem as a multi-dimensional contract design.
More specifically, the MNO needs to design a contract (with different combinations of data caps and the corresponding subscription fees) for users of different types,  so that each user will truthfully reveal his type (i.e., private preferences) by selecting a contract item intended for his type.

The key results and contributions of this paper are summarized  as follows: 
\begin{itemize}  
	%====================================================================================
	\item \textit{Systematic Study on MNO's Price Discrimination:}
	To the best of our knowledge, this is the first work studying the MNO's price discrimination through optimizing the multi-cap wireless data plans. 
	We take into account both the time flexibility (of the rollover data mechanisms) and the realistic asymmetric information.
	
	\item \textit{Exploring Time Flexibility in Price Discrimination:}
	We investigate three different data mechanisms (i.e., one traditional data mechanism and two rollover data mechanisms) and analyze the MNO's multi-cap data plan optimization under the three data mechanisms in a common design framework.
	%====================================================================================
	\item \textit{Solving the Optimal Contract:} 
	The MNO's contract problem involves user's multi-dimensional private information, hence is challenging to solve.
	We exploit the separable structure (between users' types and quota allocation) of our problem  and develop a tractable approach to solve the MNO's contract problem.
	First, we find that the slope of a user's indifference curve on the contract plane corresponds to his willingness-to-pay, and a feasible contract (satisfying the incentive compatibility and individual rationality conditions) should allocate the data caps according to users' willingness-to-pay.
	This enables us to obtain the optimal prices for a particular data cap allocation in closed-form.
	Second, for the non-convex data cap allocation problem, we propose a Dynamic Quota Allocation Algorithm, which guarantees the global optimality with a low computational complexity.
	%====================================================================================
	\item \textit{Performance Evaluation based on Empirical Data:} 
	We evaluate the optimal contract under different data mechanisms based on the empirical data.
	The numerical results show that the time-flexible data mechanisms increase both the MNO's profit (25\% on average) and users' payoffs (8.2\% on average) under the multi-cap price discrimination, hence improves the social welfare.
\end{itemize}

The remainder of this paper is organized as follows. 
In Section \ref{Section: Literature Review}, we review the related works. 
Section \ref{Section: System Model} introduces the system model. 
Section \ref{Section: Contract Feasibility} analyzes the contract feasibility and Section \ref{Section: Contract Optimality} studies the contract optimality.
In Section \ref{Section: Numerical Results}, we present the numerical results. 
Finally, we conclude this paper in Section \ref{Section: Conclusion}.

\section{Literature Review\label{Section: Literature Review}} 
There have been many excellent studies on the wireless data plan optimizations (e.g., \cite{wang2017role,zheng2018optimizing,xiong2017economic,wang2016user}). 
However, they did not take into account the recently introduced rollover mechanism or the ubiquitous multi-cap scheme.

The rollover mechanisms have been studied in \cite{zheng2016understanding,wei2018novel,Zhiyuan2018TMC,ZhiyuanCompetitionTMC,Zhiyuan2019Infocom}.
Zheng \textit{et al.} in \cite{zheng2016understanding} found that moderately price-sensitive users can benefit from subscribing to the rollover data plan compared with the traditional data plan. 
Wei \textit{et al.} in \cite{wei2018novel} studied the rollover period length from a profit-maximizing MNO's perspective. 
In our previous works, we studied the optimization of the time-flexible data plans in \cite{Zhiyuan2018TMC} and investigated the impact of the market competition in \cite{ZhiyuanCompetitionTMC} and the trading market in \cite{Zhiyuan2019Infocom}.
However, all of these studies were based on the single-cap scheme without considering the ubiquitous multi-cap adoption.

The MNO's multi-cap offering was seldom studied in previous literature.
Dai  \textit{et al.} in \cite{dai2015effect} considered a case where the MNO offers two different data caps, i.e., a cap of basic rate and a cap of premium rate.
However, the analysis was difficult to be generalized to more than two data caps.
Therefore, there is no existing systematic study on the MNO's optimal multi-cap design, let alone under the time-flexible data mechanisms.
A key challenge for this problem is that different users make their data cap choices  based on their individual preferences, which are often private information and can be  multi-dimensional.
Hence the MNO needs to properly design multiple data caps to differentiate users without knowing their exact  private information and maximize the MNO's profit. 
Such a problem naturally leads to a contract design problem \cite{bolton2005contract}.
%This motivates us to study the contract-theoretic design for multi-cap mobile data plan.
\begin{table}
		\setlength{\abovecaptionskip}{3pt}
		\setlength{\belowcaptionskip}{0pt}
	\renewcommand{\arraystretch}{1.05}		
	\caption{Comparing Related Literature.} 
	\label{table: Literature}
	\centering
	\begin{tabular}{c c c c c c}
		\toprule
		\textbf{Literature} 			& \textbf{Rollover  Considered?	}			&  \textbf{Multi-Cap Considered?}	\\
		\midrule
		{[10]-[13]}						&  No			& No										\\
		%		\midrule
		{[8][9][14]-[16]}						&  Yes			& No										\\
		%		\midrule
		{\cite{dai2015effect}}						& No			& Yes (but limited)										\\
		%		\midrule
		{This Paper}					& Yes			& Yes										\\
		\bottomrule
	\end{tabular}  
\end{table}

Users' multi-dimensional private information leads to a multi-dimensional contract design problem.
Such a problem is often very challenging, since the multi-dimensional private information makes it difficult to achieve the global incentive compatibility \cite{deneckere2011multi}. 
To address this problem, McAfee and McMillan in \cite{mcafee1988multidimensional} proposed the generalized single-crossing condition  to ensure the globally incentive compatibility for a contract problem with multi-dimensional private information, but such a strong condition is not satisfied in many models (including ours).
Rochet and Chone in \cite{rochet1998ironing} developed a \textit{sweeping} procedure which adjusts the solution to ensure the global incentive compatibility. 
Such an approach requires that the dimension of the type space and allocation space coincide (which is not applicable to the MNO's multi-cap data plans optimization), and cannot be solved analytically except in very special cases.
In this paper, we introduce users' willingness-to-pay by investigating their indifference curves, based on which we can develop a tractable approach for the MNO to provide the global incentive to all user types and solve its optimal contract under multi-dimensional private information.

\section{System Model\label{Section: System Model}} 
We formulate the MNO's multi-cap data plan design as a three-step process  as shown in Fig. \ref{fig: SystemModel}.
In Step I, the MNO collects data from the user market to estimate the \textit{statistical information} of users' individual preferences (i.e., a user's type), which are often \textit{private} and \textit{multi-dimensional} information (hence difficult to predict on a per user basis).
In Step II, the MNO chooses a data mechanism to provide the subscribers with time flexibility.
Then in Step III, the MNO proceeds with the multi-cap contract design to induce users truthfully revealing  their types and hence maximize the MNO's profit.
Generally speaking, the MNO should periodically (e.g., every year) repeat the three steps to capture users' varying requirements (due to, for example, technology changes).

Furthermore, the MNO should extract as many dimensions of the user type as possible to characterize users' private information precisely, which leads to a contract problem with multi-dimensional private information.
As mentioned as Section I, a multi-dimensional contract is challenging to solve.
In this paper, we exploit the separable structure (between the user's types and the data cap allocation), and propose to characterize each type of users' willingness-to-pay by investigating their indifference curves.
To provide a clear demonstration, we use a two-dimensional user type to illustrate our approach.\footnote{In reality, the MNO can further introduce more dimensions and solve the multi-dimensional contract using our method if the users' types and the data cap allocation exhibit a similar structure.}

Next we describe three data mechanisms in Section \ref{SubSection: Data Mechanisms}.
Then we introduce users' two-dimensional characteristics and derive users' payoffs under different data mechanism in Section \ref{SubSection: User Model}.
Finally, we formulate the MNO's optimal contract  problem in Section \ref{SubSection: MNO's Contract Formulation}.

\begin{figure}
	\setlength{\abovecaptionskip}{5pt}
	\setlength{\belowcaptionskip}{0pt}
	\centering
	\includegraphics[width=0.8\linewidth]{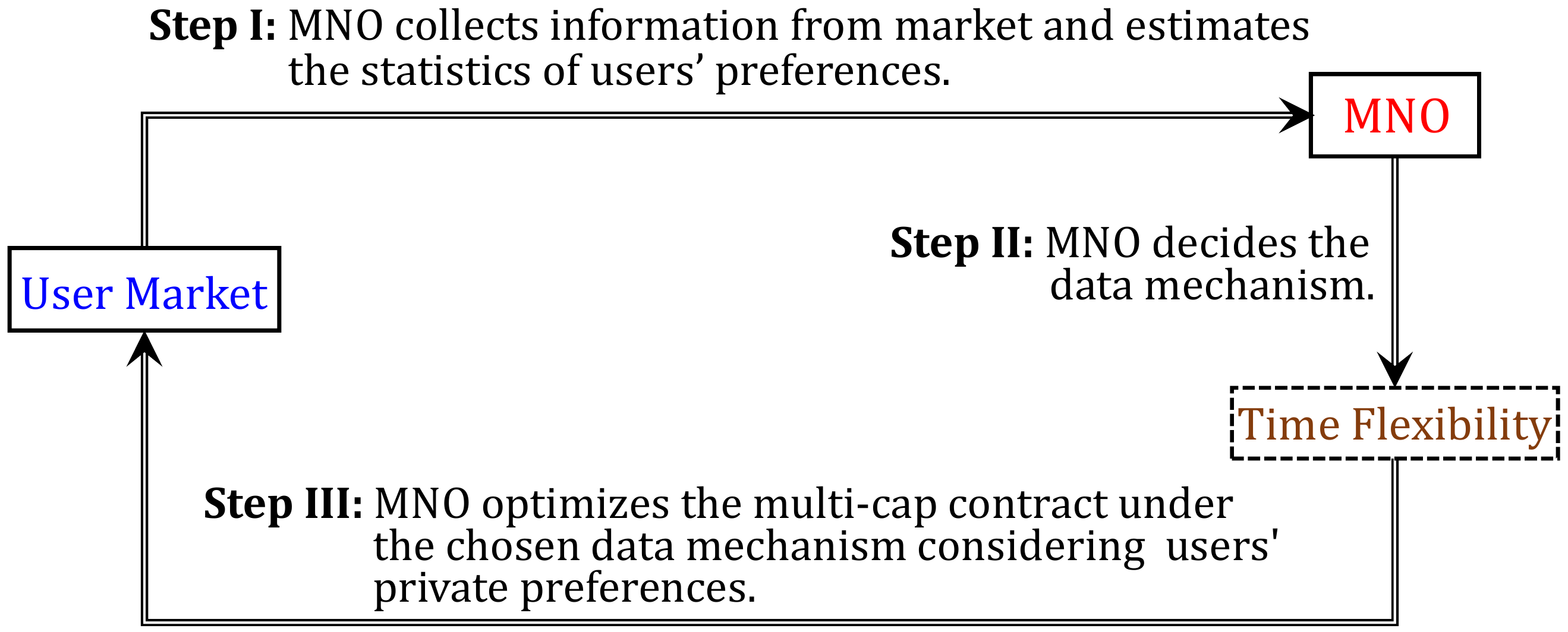}
	\caption{System model for the MNO's multi-cap design.\vspace{-5pt}}
	\label{fig: SystemModel}
\end{figure}

\subsection{Data Mechanisms\label{SubSection: Data Mechanisms}} 
A mobile data plan  can be characterized by the tuple $\plan=\{\dcap,\pcap,\adfee,\mechanism\}$, where a subscriber pays a  lump-sum subscription fee $\pcap$ for a data usage up to the monthly data  cap $\dcap$, beyond which the MNO will charge an additional fee $\adfee$ for each unit of data consumption.\footnote{We assume that all data plans have the same additional unit usage fee $\adfee$. This is often true in practice. For example, for AT\&T, $\adfee=\$15$/GB.}
Here $\mechanism\in\{0,1,2\}$ represents different data mechanisms that offer subscribers different time flexibilities on their data consumption over time.
%In our previous work, we a

The key differences among the three data mechanisms are the \textit{rollover data} and \textit{consumption priority}, both of which will affect the subscriber's expected overage data consumption \cite{Zhiyuan2018TMC}. 
First, the rollover data from the previous month can enlarge a user's \textit{effective data cap} of the current month,  within which no additional fee involved. 
Second, the consumption priorities  of the rollover data and the monthly data cap further affect how much the effective cap is enlarged.
In Table \ref{table: Various Plans}, we use $\tau$ to denote a user's rollover data from the previous month. 
More specifically, 
\begin{itemize} 
	\item The case of $\mechanism=0$ denotes the traditional data plan. 
	The subscriber has no rollover data, and the effective cap of each month is $\dcap_0^e(\tau)=\dcap$; 
	
	\item The case of $\mechanism=1$ denotes the rollover data plan offered by AT\&T. 
	The rollover data $\tau\in[0,\dcap]$ from the previous month is consumed \textit{after} the current monthly data cap $\dcap$. 
	Thus the effective cap of the current month is $\dcap_1^e(\tau)=\dcap+\tau$;
	
	\item The case of $\mechanism=2$ denotes the rollover data plan offered by China Mobile. 
	The rollover data $\tau\in[0,\dcap]$ from the previous month is consumed \textit{prior} to the current monthly data cap $\dcap$. 
	Thus the effective cap of the current month is $\dcap_2^e(\tau)=\dcap+\tau$; 
%	\item The case of $\mechanism=3$ denotes the credit data plan. 
%	The special data of the previous month is the credit data deficit $\tau\in[-\dcap,0]$ borrowed from the current month, which is consumed \textit{after} the previous monthly data cap. Thus the effective cap of the current month is $\dcap_3^e(\tau)=2\dcap+\tau$ (including the data that can be borrowed from the next month).  
\end{itemize}

\begin{table}
	\setlength{\abovecaptionskip}{1pt}
	\setlength{\belowcaptionskip}{0pt}
	\renewcommand{\arraystretch}{1.1}		
	\caption{ $\plan\eq\{\dcap,\pcap,\adfee,\mechanism\},\ \mechanism\in\{ 0, 1, 2\}$.} 
	\label{table: Various Plans}
	\centering
	\begin{tabular}{c c c c c c}
		\toprule
		Plan 				& Rollover Data $\tau$		& Consumption Priority						& $\dcap_\mechanism^e(\tau)$\\
		\midrule
		$\mechanism=0$		& 0							& Cap							& $\dcap$		\\
		%		\hline
		$\mechanism=1$		& $\tau\in[0,\dcap]$		& Cap$\Rightarrow$Rollover		& $\dcap+\tau$ 	\\
		%		\hline
		$\mechanism=2$		& $\tau\in[0,\dcap]$		& Rollover$\Rightarrow$Cap		& $\dcap+\tau$ 	\\
		\bottomrule
	\end{tabular}\vspace{-5pt}
\end{table}

%The rollover data plan and the credit data plan correspond to exploiting the time flexibility in the backward fashion and forward fashion, respectively. 
%The traditional data plan is a special case with no time flexibility \cite{Zhiyuan2018TMC}. 

As we mentioned above, the time flexibility can enlarge the subscriber's effective data cap.
According to Table \ref{table: Various Plans}, the effective data cap of the traditional data mechanism $\mechanism=0$ is always $\dcap$.
However, for $\mechanism\in\{1,2\}$, the effective data cap is $\dcap+\tau$, which is no smaller than $\dcap$ in the traditional data mechanism.
Although $\mechanism=1$ and $\mechanism=2$ lead to the same expression $\dcap+\tau$, the stationary distribution of $\tau$ is different for $\mechanism\in\{1,2\}$.\footnote{We refer interested readers to Section 4 of \cite{Zhiyuan2018TMC} for more details. Moreover, when we consider the $K$-month rollover period, the rollover data has an even larger range, i.e., $\tau\in\{0,1,2,...,K\dcap\}$.}
Intuitively, the larger the effective data cap is, the less additional payment is incurred, which will further change users' subscription choices.

\subsection{User Model\label{SubSection: User Model}}  

\subsubsection{User Characteristics}
Next we introduce users' stochastic data demand $d$ and the two-dimensional preferences: $\val$ for the valuation of unit data and $\cut$ for the network substitutability.

To capture the stochastic nature of a user's data demand over time, we model a user's data demand as a discrete random variable with a probability mass function $f(d)$, a mean value of $\dmean$, and a finite integer support $\{0,1,2,...,\dmax\}$.\footnote{In practice, the MNO can estimate users' demand distributions based on their historical data usage, and incorporate such a difference among users into the user type modeling. In this paper, we focus on the user differences in data evaluation and network substitutability, and assume homogeneous demand distribution \cite{lambrecht2007does,nevo2016usage}. Notice that users' demand realizations can still be different. }
Here the data demand $d$ is measured in the minimum data unit (e.g, 1KB or 1MB according to the MNO's billing practice).
Accordingly, we denote $\val$ as a user's utility from one unit of data consumption, i.e., his valuation for unit data \cite{wang2017role,ma2016usage}.

%We further assume that all users have the same data demand distribution $f(d)$.

%We consider a linear utility function as in \cite{Zhiyuan2017rollover}\cite{lambrecht2007does}\cite{ma2016usage}, i.e., a user's satisfaction of consuming $d$ units data is
%\begin{equation}
%u(d)= \val d ,
%\end{equation}
%where $d$ is the user's realized data usage.
%Such a linear utility enables us to obtain a unified expression of the user's expected payoff under different data mechanisms, and allows us to proceed the contract analysis to obtain clear insights.\footnote{Later on we will show that it can significantly increase the formulation complexity to compute the expectation (over $d$ and $\tau$).}

Furthermore, a user's data consumption behavior might change after exceeding the \textit{effective cap}, since it incurs additional payment. 
Intuitively, the user will  still continue to consume data in this case, but may reduce his data consumption by utilizing alternative networks (e.g., Wi-Fi) instead.
Therefore, we follow \cite{sen2012economics} by incorporating users' network substitutability $\cut$ as one of the user's characteristics. 
Mathematically speaking, $\cut\in[0,1]$ denotes the fraction of overage usage shrink. 
A larger $\cut$ value represents more overage usage cut (thus, a better substitutability).
A user's mobility pattern can significantly influence the availability of alternative  networks, which will further change a user's data plan choice.
For example, a businessman who is always on the road may have a poor network substitutability (hence a small value of $\cut$), hence prefers to a large data cap; while a student can take advantage of the school Wi-Fi network (hence a large value of $\cut$), hence will be fine with a small data cap.

Different from our previous works in \cite{Zhiyuan2018TMC,ZhiyuanCompetitionTMC}, in this paper, we consider a more realistic \textit{asymmetric information}  scenario, i.e., the parameters $\val$ and $\cut$ are each user's private information that the MNO does not know precisely.
As a result, we propose to use a contract-theoretic approach to cope with users' multi-dimensional private information and optimize the MNO's multi-cap data plans.

\subsubsection{User Payoff}
A user's payoff is defined as the difference between his utility and payment.
Specifically, for a type-$(\cut,\val)$ user with $d$ units data demand and an effective cap $\dcap_\mechanism^e(\tau)$, his realized data consumption is $d-\cut[d-\dcap_\mechanism^e(\tau)]^+$ where $[x]^+=\max\{0,x\}$.
Hence a type-$(\cut,\val)$ user's utility is $\val(d-\cut[d-\dcap_\mechanism^e(\tau)]^+)$.
In addition, the user's total payment consists of the monthly subscription fee $\pcap$ and the overage charge $\adfee(1-\cut)[d-\dcap_\mechanism^e(\tau)]^+$.
Therefore, the (monthly) payoff of the type-$(\cut,\val)$ user with a data demand $d$ and an effective cap $\dcap_{\mechanism}^e(\tau)$ is
\begin{equation}
\begin{aligned} 
S(\plan,\cut,\val,d,\tau)= & \val\left(d-\cut\left[d-\dcap_\mechanism^e(\tau)\right]^+\right)\\
&\qquad          -\adfee(1-\cut)\left[d-\dcap_\mechanism^e(\tau)\right]^+-\pcap.
\end{aligned}	
\end{equation}
%where $[x]^+=\max\{0,x\}$, $\val\left(d-\cut[d-\dcap_\mechanism^e(\tau)]^+\right)$ is the user's utility due to data consumption, $\adfee(1-\cut)[d-\dcap_\mechanism^e(\tau)]^+$ is the user's possible additional  payment due to his data usage above the effective cap, and $\pcap$ is the monthly lump-sum subscription fee.

Here both $d$ and $\tau$ are random variables, and we take the expectation over them  to obtain a user's expected payoff as
\begin{equation}\label{Equ: expected payoff}
\begin{aligned}
\payoff(\plan,\cut,\val) 	
&=\mathbb{E}_{d,\tau}\big\{ S(\plan,\cut,\val,d,\tau) \big\} \\
&=\val\left[\dmean-\cut A_\mechanism(\dcap)  \right] -\adfee(1-\cut) A_\mechanism(\dcap)-\pcap,
\end{aligned}
\end{equation}
where $A_\mechanism(\dcap)$ is the type-$(0,\val)$ subscriber's expected overage data consumption, as follows:
\begin{equation}\label{Equ: system model A_i(Q_i)}
\begin{aligned}
A_\mechanism(\dcap)=&\mathbb{E}_{d,\tau} \big\{[ d-\dcap_\mechanism^e(\tau) ]^+ \big\} \\
=& \textstyle \sum\limits_{d=0}^{\dmax}\sum\limits_{\tau=0}^{\dcap}	\left[d-\dcap_\mechanism^e(\tau)\right]^+f(d) p_\mechanism(\tau).
\end{aligned}
\end{equation}

Note that the differences among the three data mechanisms are entirely captured by $A_\mechanism(\dcap)$ in (\ref{Equ: system model A_i(Q_i)}).
Specifically,  $p_\mechanism(\tau)$ in (\ref{Equ: system model A_i(Q_i)}) represents the distribution of the subscriber's rollover data under data mechanism $\kappa$, which is the key difference among the three data mechanisms.
In our previous work, we have introduced how to compute $p_\mechanism(\tau)$ and $A_{\mechanism}(\dcap)$ in details (see Section 4 of \cite{Zhiyuan2018TMC}).
In this paper, we directly summarize the key conclusion from \cite{Zhiyuan2018TMC} in Proposition \ref{Proposition: flexibility}.
\begin{proposition}\label{Proposition: flexibility}
	For an arbitrary data demand distribution $f(d)$,  $A_0(\dcap)> A_1(\dcap)> A_2(\dcap)$ for any $\dcap\in(0,\dmax)$. 
\end{proposition}

Proposition \ref{Proposition: flexibility} indicates that a user incurs less overage data consumption under the rollover mechanism $\mechanism\in\{1,2\}$ than the traditional one $\mechanism=0$.
Moreover, among the two rollover mechanisms  $\mechanism\in\{1,2\}$, $\mechanism=2$ is more time-flexible than $\mechanism=1$, since $A_1(\dcap)> A_2(\dcap)$.
This is why we say that {the rollover mechanism $\mechanism=2$ offers the best time flexibility, while $\mechanism=0$ offers the worst}.
%In this paper, we will directly use this conclusion, and refer interested readers to Section 4 in \cite{Zhiyuan2018TMC} for more details.

The above discussion indicates  that a user's expected payoffs under different mechanisms have a similar expression. 
The difference is only in terms of the expected overage usage $A_\mechanism(\dcap_\mechanism)$. 
Thus, for notation simplicity, we will focus on a generic data mechanism and express the expected payoff of a type-$(\cut,\val)$ user as
\begin{equation}\label{Equ: Generalization Payoff}
\begin{aligned}
\payoff(\dcap,\pcap,\cut,\val)=V(\dcap,\cut,\val)-P(\dcap,\cut)-\pcap, \\
\end{aligned}
\end{equation} 
where $V(\dcap,\cut,\val)\eq\val[\dmean-A_\mechanism(\dcap) \cut ]$ is the  subscriber's  utility, and $P(\dcap,\cut)\eq\adfee(1-\cut) A_\mechanism(\dcap)$ is the  overage payment.
In economics, the subscription fee is a user's \textit{sunk cost} (incurred in advance and often independent of the user's actual consumption), while the overage payment is the \textit{prospective cost} (depending on the user's actual consumption).
Therefore, we call the user's payoff without the sunk cost as the ``virtual payoff'', defined as
\begin{equation}\label{Equ: Virtual Payoff}
L(\dcap,\cut,\val) \triangleq V(\dcap,\cut,\val)-P(\dcap,\cut), 
\end{equation} 
which will be used in Section \ref{Section: Contract Feasibility} and Section \ref{Section: Contract Optimality}.

So far we have generalized the users' expected payoffs under different data mechanisms into a unified expression.
Our later analysis for the MNO's optimal contract problem is based on this general framework.

\subsection{MNO's Contract Formulation\label{SubSection: MNO's Contract Formulation}}
Next we formulate the MNO's optimal contract problem.

\subsubsection{Feasible Contract}
The MNO offers a contract (with different combinations of data caps and corresponding subscription fees) to a group of users who are distinguished by two-dimensional private information: the data valuation $\val$ and the network substitutability $\cut$. 
Recall that in Step I (of Fig. \ref{fig: SystemModel}), the MNO collects the statistical information from the user market.
For example, we consider a set $\Stheta=\{\val_k:1\le k\le{K}\}$ of $K$ data valuation types and a set  $\Sbeta=\{\cut_m:1\le m\le{M}\}$ of $M$ network substitutability types. 
Hence there are a total of ${KM}$ types of users in the market, characterized by a joint probability mass function $q_{m,k}$ for each type-$(\cut_m,\val_k)$ user.\footnote{The MNO can flexibly divide users' into several categories through some data mining techniques such as $k$-means \cite{ma2016economic,kodratoff2014introduction}. The choices of parameters $K$ and $M$ determine the trade-off between contract complexity and profit.}
Without loss of generality, we assume that users' types are indexed in the ascending sort order in both dimensions, i.e., $\val_1<\val_2<...<\val_{K}$ and $\cut_1<\cut_2<...<\cut_{M}$.
%The case of MNO's single-cap design is a special case with $K=M=1$ \cite{Zhiyuan2017rollover}. 

According to the revelation principle \cite{fudenberg1991game}, it is enough for the MNO to consider a class of contracts that enables users to truthfully reveal their types. 
In other words, it is enough for the MNO to design a contract, denoted by $\Phi(\Sbeta,\Stheta)=\{\phi_{m,k},1 \leq m \leq M,1\leq k \leq K\}$ that consists of $KM$ contract items $\phi_{m,k}=\{\dcap_{m,k},\pcap_{m,k}\}$, one for each user type.
Formally, a contract is \textit{feasible} if and only if it ensures that each user selects the contract item intended for this type.
It is obvious that a contract is \textit{feasible} if and only if it satisfies the Individual Rationality (IR) and Incentive Compatibility (IC) conditions, defined as follows:

%===============================================
\begin{definition}[Individual Rationality]\label{Definiation: Individual Rationality}
	A contract is individually rational if for all $1\le m\le M$ and $1\le k\le K$, the type-$(\cut_m,\val_k)$ user achieves a non-negative payoff by choosing the contract item $\phi_{m,k}$ intended for this user type, denoted by $\contract_{m,k}\ir$, i.e., %$\payoff(\phi_{m,k},\cut_m,\val_k)\ge0$.
	\begin{equation}\label{Equ: IR}
	\payoff(\phi_{m,k},\cut_m,\val_k)\ge0.
	\end{equation}
\end{definition}

\begin{definition}[Incentive Compatibility]\label{Definiation: Incentive Compatibility}
	A contract is incentive compatible if for all $1\le m\le M$ and $1\le k\le K$, the type-$(\cut_m,\val_k)$ user maximizes its payoff by choosing the contract item $\phi_{m,k}$ intended for this user type, i.e., 
	\begin{equation}\label{Equ: IC}
	\payoff(\phi_{m,k},\cut_m,\val_k)\ge \payoff(\contract_{n,l},\cut_m,\val_k),\ \forall\ (n,l)\ne(m,k).
	\end{equation}
\end{definition}

Our later analysis for the contract feasibility in Section \ref{Section: Contract Feasibility} involves the concept of Pairwise Incentive Compatibility (PIC) in Definition \ref{Defination: PIC}.
Basically, PIC consists of the all IC conditions in the two-user scenario.
That is, the $KM(KM-1)$ IC conditions are equivalent to the $KM(KM-1)/2$ PIC conditions for all the two-user pairs.
\begin{definition}[Pairwise Incentive Compatibility]\label{Defination: PIC}
	The contract items $\phi_{m,k}$ and $\phi_{n,l}$ are pairwise incentive compatible, denoted by $\contract_{m,k}\ic\contract_{n,l}$, if and only if
	\begin{equation}
	\left\{
	\begin{aligned}
	& \payoff(\phi_{m,k},\cut_m,\val_k) \ge \payoff(\phi_{n,l},\cut_m,\val_k), \\
	& \payoff(\phi_{n,l},\cut_n,\val_l) \ge \payoff(\phi_{m,k},\cut_n,\val_l).
	\end{aligned}
	\right.
	\end{equation}
\end{definition}

\subsubsection{MNO's Profit}

Next we derive the MNO's revenue, cost, and profit under a \textit{feasible} contract $\Phi$.

The MNO's revenue from a subscriber consists of the subscription fee and the overage fee.
Based on the above discussion of the feasible contract, the MNO's expected revenue $\revenue(\Phi)$ under a \textit{feasible} contract $\Phi$ is
\begin{equation}\label{Equ: Revenue}
\begin{aligned}
&\revenue(\Phi)=  \sum\limits_{k=1}^{K}\sum\limits_{m=1}^{M} q_{m,k} \big[ \underbrace{\pcap_{m,k}}_\text{subscription} + \underbrace{P(\dcap_{m,k},\cut_m)}_\text{overage} \big].
\end{aligned}
\end{equation}

Furthermore, we consider two kinds of costs experienced by the MNO, i.e.,  the capacity cost and operational cost.
%To be more specific, in economics and business decision-making, the sunk cost is a cost that has already been incurred and cannot be recovered, while the prospective cost is a cost that may be incurred or changed if an action is taken \cite{frank1991microeconomics}.

The MNO's capital expenditure is mainly due to its investment on the network capacity \cite{sen2013survey}.
Imposing the data cap would help manage the network congestion and arrange the scarce network capacity \cite{dai2015effect}.
%Therefore, once the MNO decides a data cap to be offered to a user, it should make sure a corresponding network capacity to be prepared in advance.
Motivated by this phenomenon, we model the MNO's capacity cost caused by a type-$(\cut_m,\val_k)$ subscriber as an increasing function $J(\dcap)$ in his data cap $\dcap$ \cite{ma2016usage}.
Intuitively, a larger data cap corresponds to a severer network congestion on average that requires the MNO's more investment on the network in advance.

The MNO's operational cost is mainly due to the system management \cite{wang2017optimal}.
After the MNO decides which data plan to implement, the subscribers' total data consumption will influence the MNO's operational expense.
Therefore, the MNO's operational cost caused by a type-$(\cut_m,\val_k)$ subscriber with data cap $\dcap$ can be formulated as $c\cdot U(\dcap,\cut_m)$, where $c$ is the MNO's marginal cost for the system management \cite{dai2015effect}, and $U(\dcap,\cut_m)=\dmean-\cut_m A(\dcap)$ is the type-$(\cut_m,\val_k)$ subscriber's expected data consumption.\footnote{Such a linear-form cost has been widely used to model an operator's operational cost, e.g., \cite{luo2016integrated,duan2012duopoly}.}
%Therefore, we model the MNO's prospective cost as $\load_i(\dcap_i,\pcap_i,\adfee_i)\cdot c$, where $\load_i(\dcap_i,\pcap_i,\adfee_i)$ is the users' total data consumption under $\plan_i$, and $c$ represents the MNO's marginal prospective cost.

Therefore, the MNO's expected cost $\cost(\Phi)$ under a \textit{feasible} contract $\Phi$ can be calculated as
\begin{equation}\label{Equ: Cost}
\begin{aligned}
&\cost(\Phi)=  \sum\limits_{k=1}^{K}\sum\limits_{m=1}^{M} q_{m,k} \big[ \underbrace{c\cdot U(\dcap_{m,k},\cut_m)}_\text{Operational cost}   + \underbrace{J(\dcap_{m,k})}_\text{Capacity cost} \big].
\end{aligned}
\end{equation}

The MNO's expected profit under a \textit{feasible} contract $\Phi$ is the difference between its revenue and cost, given by
\begin{equation}\label{Equ: Profit}
\begin{aligned}
&\profit(\Phi)= \revenue(\Phi)-\cost(\Phi).
\end{aligned}
\end{equation}

\subsubsection{MNO's Multi-dimensional Contract Problem}
Based on the above discussion, we formulate the MNO's contract problem as follows:
\begin{problem}[Optimal Contract Design]\label{Problem: Optimal Contract}
	\begin{equation} 
	\begin{aligned}
	&\max\limits_{\Phi}\ \profit(\Phi) \\
	&{\ \rm\ s.t.}\ (\ref{Equ: IR}),(\ref{Equ: IC}).
	\end{aligned}
	\end{equation}
\end{problem}

The key idea of the contact design problem is to ensure the individual rationality and the incentive compatibility of all user types, so that each user is willing to participate and truthfully reveals his type by selecting the contract item intended for this type of users.
Problem \ref{Problem: Optimal Contract} makes it clear, where the MNO needs to address a total of $KM$ IR constraints (condition  (\ref{Equ: IR})) and a total of $(KM-1)KM$ IC constraints (condition (\ref{Equ: IC})).

The main difficulty of Problem \ref{Problem: Optimal Contract} is twofold:
\begin{enumerate}
	\item \textit{The non-monotonicity of the allocation rule}. 
			A monotonic allocation rule usually requires the satisfaction of the  single-crossing property, under which  two  indifference curves of any two different user types cross only once \cite{bolton2005contract}.
			That is, the user's marginal utility should be monotone increasing (or monotone decreasing) in the user type.
			When this condition holds, an allocation rule is incentive compatible only if the rule is monotonic in the user type \cite{athey2001single}.
			In Problem \ref{Problem: Optimal Contract}, we have
			\begin{equation}
			\frac{\partial^2 \payoff(\dcap,\pcap,\cut,\val)}{\partial \dcap\partial\val}=-\frac{\partial A(\dcap)}{\partial \dcap}\cdot \cut\ge0,\ \forall \cut\in[0,1],
			\end{equation}
			which indicates that the marginal utility increases in the data valuation $\val$ for any $\cut\in[0,1]$.
			Therefore, the higher valuation user deserves a larger allocation for any $\cut\in[0,1]$.
			However, for the network substitutability $\cut$, we have 
			\begin{equation}
			\frac{\partial^2 \payoff(\dcap,\pcap,\cut,\val)}{\partial \dcap\partial\cut}=-\frac{\partial A(\dcap)}{\partial \dcap}\cdot\left( \val-\adfee \right),
			\end{equation}
			which can be positive or negative, depending on the relationship between the data valuation $\val$ and the per-unit fee $\adfee$.
			Therefore, the allocation rule in terms of the network substitutability $\cut$ is not monotonic and hence is challenging to analyze.
	\item \textit{Two-dimensional user types}.
			A contract design involving multi-dimensional user types is also very challenging in general. 
			For contract problems involving only one-dimensional user types, the satisfaction of single-crossing condition guarantees a monotone allocation rule.
			Therefore, the approach used in \cite{gao2011spectrum,duan2014cooperative,zhang2015contract,zhang2017non} can significantly reduce the unbinding IC and IR constraints so that the contract problem is more tractable.
			However, the approach in \cite{gao2011spectrum,duan2014cooperative,zhang2015contract,zhang2017non} cannot be easily generalized to the two-dimensional user type case, even if the allocation rule is consistent (and we have shown that it is not in our problem). 
\end{enumerate}

%In general, our multi-dimensional contract Problem \ref{Problem: Optimal Contract} is very challenging to solve.
Next we will exploit  the special structure in Problem \ref{Problem: Optimal Contract} and propose a new approach of solving the problem. 
This is a key contribution of this paper.
Specifically, we will investigate the contract feasibility and optimality in Section \ref{Section: Contract Feasibility} and Section \ref{Section: Contract Optimality}, respectively.
Table \ref{Table: Key notation} summarizes the key notation in this paper.

\begin{table}
	\setlength{\abovecaptionskip}{1pt}
	\setlength{\belowcaptionskip}{0pt}
	\renewcommand{\arraystretch}{1.1}		
	\caption{Key Notation} 
	\label{Table: Key notation}
	\centering
	\begin{tabular}{c|l}
		\toprule
		\textbf{Symbol} 		& $\qquad\qquad\qquad\qquad$\textbf{Physical Meaning}					\\
		\midrule		
		$\dcap$					& The monthly data cap.	\\
		$\pcap$					& The fixed monthly subscription fee.	\\
		$\adfee$				& The overage usage fee when exceeding the data cap.	\\
		$\mechanism$			& The data mechanism $\mechanism\in\{0,1,2\}$.	\\
		\midrule
		$\val$					& The user's  data valuation.	\\
		$\cut$					& The user's  network substitutability.	\\
		$\Stheta$				& A total of $K$ different $\val$, i.e., $\Stheta=\{\val_k,1\le k\le K\}$.\\
		$\Sbeta$				& A total of $M$ different $\cut$, i.e., $\Sbeta=\{\cut_m,1\le m\le M\}$.\\
		$ \type_i$				& The $i$-th ($1\le i\le KM$) user type after sorting as (\ref{Equ: order}).	\\
		$\type_{\epsilon}$		& The smallest-payoff user type defined in (\ref{Equ: smallest-payoff user}).\\
		$\slope$				& The user's willingness-to-pay, defined in (\ref{Equ: slope}). \\
		$\payoff$				& The user's monthly expected payoff, defined in (\ref{Equ: Generalization Payoff}).\\
		$L$ 					& The user's virtual payoff, defined in (\ref{Equ: Virtual Payoff}).\\
		$\eta^+,\eta^-$		& The user's virtual payoff increment, defined in (\ref{Equ: Virtual payoff increament}). \\
		$\rho^+,\rho^-$		& The user's virtual payoff differences, defined in (\ref{Equ: virtual payoff differences}).\\
		\midrule
		$\revenue$				& MNO's expected revenue, defined in (\ref{Equ: Revenue}). \\
		$C$						& MNO's expected cost, defined in (\ref{Equ: Cost}).\\
		$\profit$				& MNO's expected profit, defined in (\ref{Equ: Profit}).\\
		$\Phi$					& MNO's contract $\Phi=\{\phi_{m,k},1\le m\le M, 1\le k\le K\}$. \\
		$\contract_{m,k}$		& Contract item $\{\dcap_{m,k},\pcap_{m,k}\}$ for type-$(\cut_m,\val_k)$ user.	\\
		$\contract_i$			& Contract item $\{\dcap_i,\pcap_i\}$ for type-$\type_i$ user after sorting. \\	
		\bottomrule
	\end{tabular}  
\end{table}

%%%%%%%%%%%%%%%%%%%%%%%%%%%%%%%%%%%%%%%%%%%%%%%%%%%%%%%%%%%%%%%%%%%%%%%%%%%%%%%%%%%%%%%%%%%%%%%%%%%%%%%%%%%%%%%%%%%%%%%%%% %%%%%%%%%%%%%%%%%%%%%%%%%%%%%%%%%%%%%%%%%%%%%%%%%%%%%%%%%%%%%%%%%%%%%%%%%%%%%%%%%%%%%%%%%%%%%%%%%%%%%%%%%%%%%%%%%%%%%%%%%% 
\section{Contract Feasibility\label{Section: Contract Feasibility}}
To study the feasibility of the two-dimensional contract, we will first introduce a user's marginal rate of substitution (which also represents the user's willingness-to-pay) and the new user ordering in Section \ref{Subsection: Marginal Rate of Substitution} and Section \ref{Subsection: User Ordering}, respectively.
Then we investigate the necessary and sufficient conditions for a feasible contract in Section \ref{SubSection: Necessary Condition} and Section \ref{SubSection: Sufficient Condition}, respectively.

\subsection{Marginal Rate of Substitution (Willingness-to-Pay)\label{Subsection: Marginal Rate of Substitution}}
In economics, a consumer's indifference curve connects those good bundles that achieve the same consumer satisfaction (payoff). 
In our problem, we can plot a user's indifference curve over the contract plane (i.e., the data cap $\dcap$ and the subscription fee $\pcap$) as in Fig. \ref{fig:illustration_MRS}.
On the $(\dcap,\pcap)$ plane, a type-$(\cut,\val)$ user's indifference curve with a fixed payoff $\payoff$ satisfies
\begin{equation}
\payoff= \val[\dmean-\cut A(\dcap)]-\adfee(1-\cut) A(\dcap) -\pcap.
\end{equation}

Fig. \ref{fig:illustration_MRS} shows that the indifference curve is increasing and concave\footnote{Showing the increasing and concave property for the indifference curve is equivalent to showing that $A(\dcap)$ is decreasing and convex in $\dcap$, which has been proved in our previous  work (see Section 5.2 of \cite{Zhiyuan2018TMC}).} in the data cap $\dcap$, which indicates that the subscription fee would increase (with a diminishing marginal increment) as the data cap increases to maintain the same payoff.
Moreover, as a user's indifference curve shifts downward, his payoff increases because of the decreasing subscription fee.

\begin{figure}
	\setlength{\abovecaptionskip}{3pt}
	\setlength{\belowcaptionskip}{0pt}
	\centering
	\includegraphics[width=0.6\linewidth]{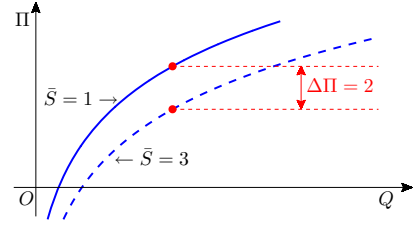}
	\caption{Two indifference curves of the same user type  with two different expected payoffs, i.e., $\payoff=1$ and $\payoff=3$.}
	\label{fig:illustration_MRS}
\end{figure}

\begin{figure*}
		\setlength{\abovecaptionskip}{2pt}
		\setlength{\belowcaptionskip}{0pt}
	\centering
	\subfigure[$\val_1<\val_K<\adfee$]{\label{fig: ThetaBeta_low} \includegraphics[width=0.28\linewidth]{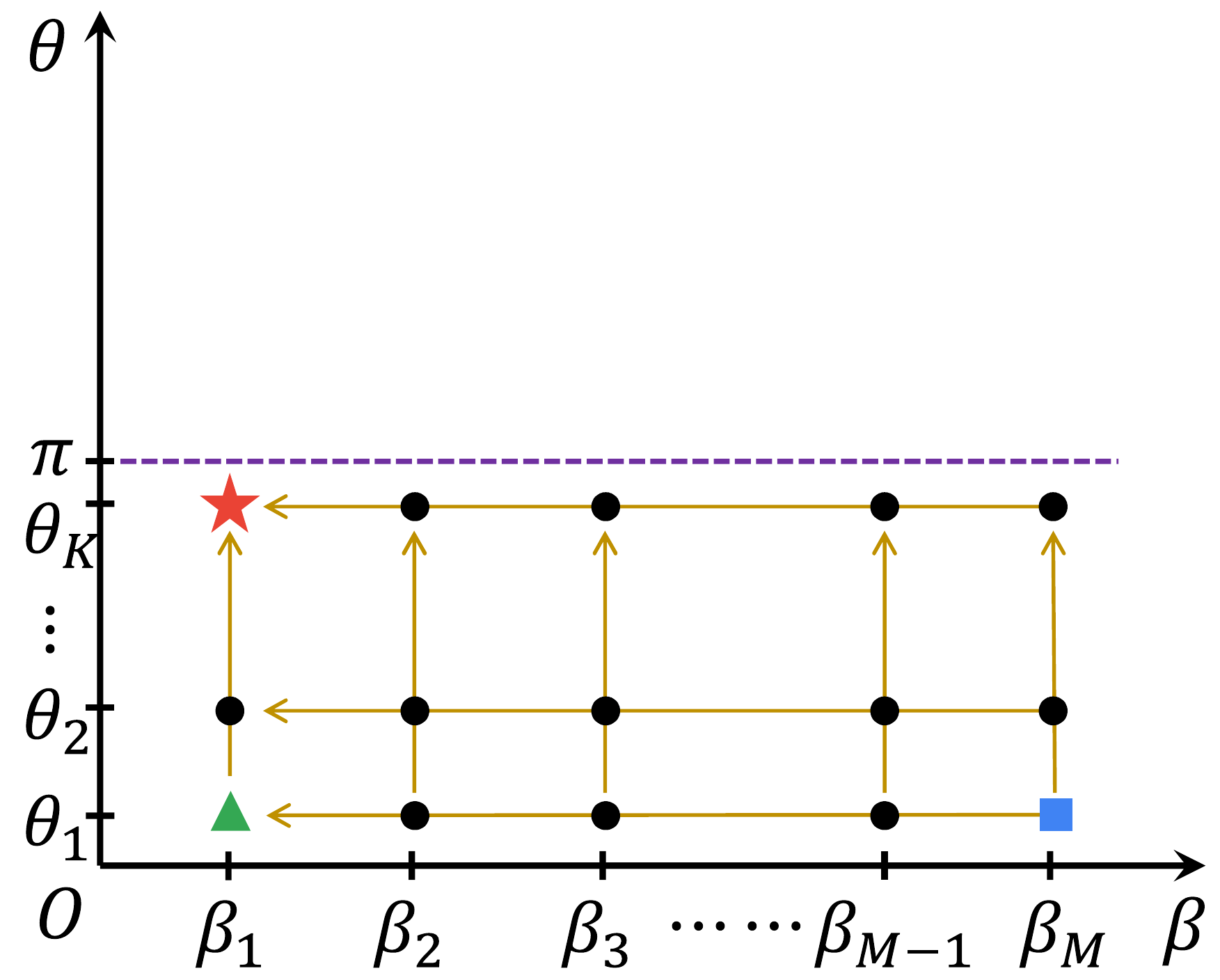}} \quad
	\subfigure[$\val_1<\adfee<\val_K$]{\label{fig: ThetaBeta_cross} \includegraphics[width=0.28\linewidth]{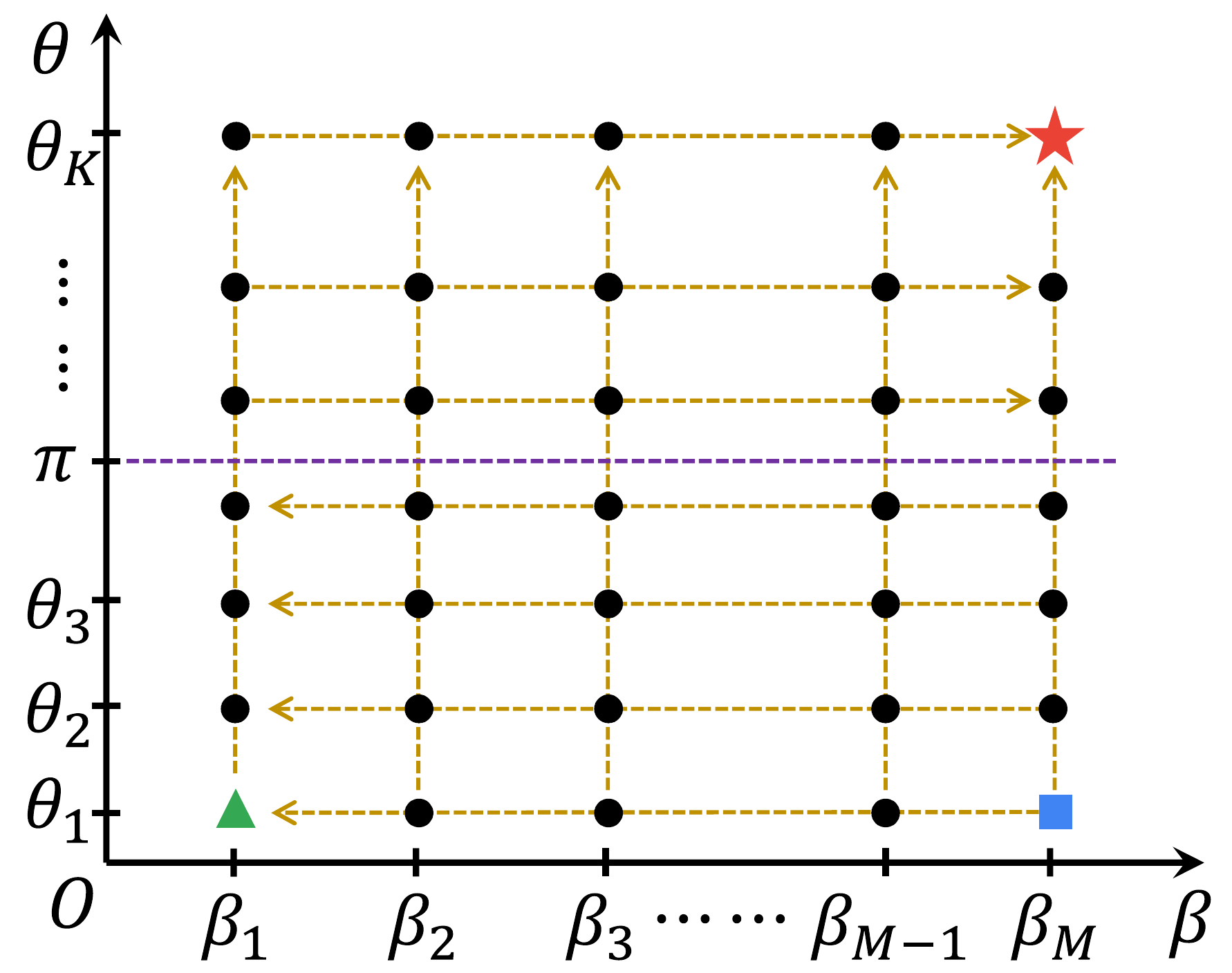}} \quad
	\subfigure[$\adfee<\val_1<\val_K$]{\label{fig: ThetaBeta_high} \includegraphics[width=0.28\linewidth]{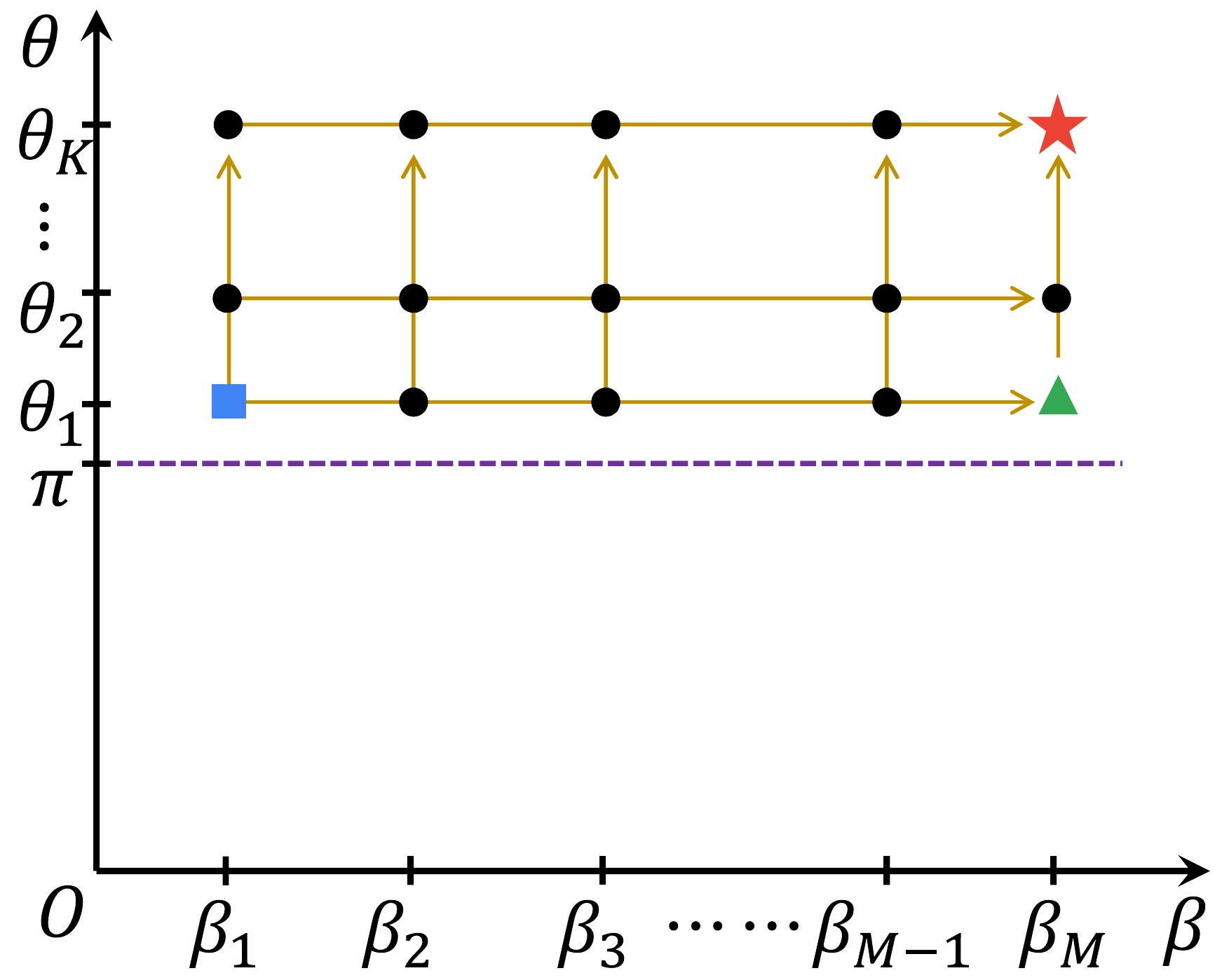}} 
	\caption{Three market modes. }
	\label{fig: market}	
\end{figure*}

The slope of an indifference curve is called the marginal rate of substitution (MRS), which  is the rate at which a consumer is ready to give up one good in exchange for another good, while maintaining the same level of satisfaction.
%Note that another interpretation of MRS is the slope of the indifference curve.
In our problem, we denote the MRS of a type-$(\cut,\val)$ user on a data cap $\dcap$ as
\begin{equation}\label{Equ: slope}
 \slope(\dcap,\cut,\val)\eq \frac{\partial \pcap}{\partial \dcap}=-\left[\val\cut+\adfee(1-\cut)\right]\frac{\partial A(\dcap)}{\partial \dcap},
\end{equation}
which depends on the user's private information $(\cut,\val)$ and the data cap $\dcap$.
The MRS $\slope(\dcap,\cut,\val)$ indicates  a type-$(\cut,\val)$ user's \textit{willingness-to-pay} for an additional unit of data on a  data cap $\dcap$. 
In the rest of the paper, we will use the three phrases ``marginal rate of substitution'', ``slope of the indifference curve'', and ``willingness-to-pay'' interchangeably.
%Moreover, most proofs are quite lengthy and are given in our technical report \cite{Technicalreport}.

\subsection{User Ordering Based on Willingness-to-Pay\label{Subsection: User Ordering}}
Without loss of generality, now we sort and index the ${KM}$ user types $(\cut_m,\val_k)$ based on the corresponding willingness-to-pay $\slope(\dcap,\cut_m,\val_k)$ in an ascending order as follows:
\begin{equation}\label{Equ: order}
\type_1(\dcap),\ \type_2(\dcap),\ ...,\ \type_{KM}(\dcap),
\end{equation}
where $\type_i(\dcap)\eq \{\cut_m,\val_k\}$ for some $k$ and $m$.
In this case, under the data cap $\dcap$, we have
\begin{equation}\label{Equ: order slope}
\slope(\dcap,\type_1)\le\slope(\dcap,\type_2)\le...\le\slope(\dcap,\type_{KM}).
\end{equation}

\begin{lemma}\label{Lemma: independent of Q}
	The new user ordering in (\ref{Equ: order}) does not depends on the data cap.
	That is, for any $\dcap\ne\dcap'$, we have
	\begin{equation}
	\type_i(\dcap)=\type_i(\dcap'),\ \forall\ 1\le i\le KM.
	\end{equation}
\end{lemma}

Lemma \ref{Lemma: independent of Q} indicates that the user ordering in (\ref{Equ: order}) does not change, even though the value of $\slope(\dcap,\type_i)$ would change with the data cap $\dcap$.
Intuitively, this is because that a user's willingness-to-pay $\slope(\dcap,\type_i)$ in (\ref{Equ: slope}) has a separable structure between the user types (i.e., $\val$ and $\cut$) and the data cap $\dcap$.
The proof of Lemma \ref{Lemma: independent of Q} is given in Appendix \ref{Appendix: User type}.

For notation simplicity, in the following, we will directly use $\type_i$ to denote a user type under the ordering specified in (\ref{Equ: order slope}),  and denote $\contract_i=\{\dcap_i,\pcap_i\}$  the contract item intended for the type-$\type_i$ users.

To have a better understanding on the new user ordering, we use Fig. \ref{fig: market} to illustrate how $(\cut_m,\val_k)$ maps to $\type_i$.
There are three different market modes depending on the relationship between the extreme valuations ($\val_1$ and $\val_K$) and the overage fee $\adfee$, i.e.,  $\val_K<\adfee$ as in Fig. \ref{fig: ThetaBeta_low}, $\val_1<\adfee<\val_K$ as in Fig. \ref{fig: ThetaBeta_cross}, and $\adfee<\val_1$ as in  Fig. \ref{fig: ThetaBeta_high}.
Specifically, the arrows in Fig. \ref{fig: market} point to the direction where the user's MRS $\slope(\dcap,\cut,\val)$ increases, the blue square denotes the \emph{minimum willingness-to-pay} user type-$\type_1$, and the red star denotes the \emph{largest willingness-to-pay} user type-$\type_{KM}$.
The following proposition summarizes the mapping from $(\cut_m,\val_k)$ to $\type_1$ and $\type_{KM}$.
The proof is given in Appendix \ref{Appendix: User type}.
\begin{proposition}\label{Proposition: type^1, type^KM}
	Under the three market modes, the type-$\type_1$ and type-$\type_{KM}$ users have their private information as follows:
	\begin{equation}
	\begin{cases}
	\type_1=\{\cut_{M},\val_1\}, \type_{KM}=\{\cut_1,\val_{K}\},	& \text{if}\ \val_1<\val_{K}<\adfee, \\
	\type_1=\{\cut_{M},\val_1\}, \type_{KM}=\{\cut_{M},\val_{K}\},	& \text{if}\ \val_1<\adfee<\val_{K}, \\
	\type_1=\{\cut_1,\val_1\},\  \type_{KM}=\{\cut_{M},\val_{K}\},	& \text{if}\ \adfee<\val_1<\val_{K}.
	\end{cases}
	\end{equation}
\end{proposition}

Furthermore, the green triangles in Fig. \ref{fig: market} denote the \emph{smallest-payoff} user type $\type_{\epsilon}(\dcap,\pcap)$ given the contract item $(\dcap,\pcap)$, defined as follows 
\begin{equation}\label{Equ: smallest-payoff user}
\type_{\epsilon}(\dcap,\pcap) \eq \arg\min\limits_{\type_i} \payoff( \dcap,\pcap,\type_i ).
\end{equation}

Lemma \ref{Lemma: u independent of Q} indicates that the \textit{smallest-payoff} user type $\type_{\epsilon}(\dcap,\pcap)$ does not change with data cap or subscription fee.
Similar to Lemma \ref{Lemma: independent of Q}, this is because the separable structure between the user types (i.e., $\val$ and $\cut$) and the contract item (i.e., $\dcap$ and $\pcap$).
For notation simplicity, we will use $\type_{\epsilon}$ in the following.
The proof of Lemma \ref{Lemma: u independent of Q} is in Appendix \ref{Appendix: User type}.
\begin{lemma}\label{Lemma: u independent of Q}
	The smallest-payoff user defined in (\ref{Equ: smallest-payoff user}) does not depends on the data cap or the subscription fee, i.e.,
	\begin{equation}
	\type_{\epsilon}(\dcap,\pcap)=\type_{\epsilon}(\dcap',\pcap'),\ \forall (\dcap',\pcap')\ne(\dcap,\pcap).
	\end{equation}
\end{lemma}

Proposition \ref{Proposition: type^u} presents the mapping from $(\cut_m,\val_k)$ to $\type_{\epsilon}$.
The proof is in Appendix \ref{Appendix: User type}.
\begin{proposition} \label{Proposition: type^u}
	Under the three market modes, the type-$\type_{\epsilon}$ user has the private information as follows:
	\begin{equation}
	\begin{cases}
	\type_{\epsilon}=\{\cut_1,\val_1\},		& \text{if}\ \val_1<\val_{K}<\adfee, \\
	\type_{\epsilon}=\{\cut_1,\val_1\},		& \text{if}\ \val_1<\adfee<\val_{K}, \\
	\type_{\epsilon}=\{\cut_{M},\val_1\},	& \text{if}\ \adfee<\val_1<\val_{K}.
	\end{cases}
	\end{equation}
\end{proposition}

Next we study the necessary conditions for a contract to be feasible based on users' \textit{willingness-to-pay}.
\subsection{Necessary Conditions \label{SubSection: Necessary Condition}}
Lemmas \ref{Lemma: Cap-Price} and \ref{Lemma: type-Cap} present  two necessary conditions for a contract to be feasible (satisfying IC and IR conditions).
The proofs are given in Appendix \ref{Appendix: Necessary}.
%==========================================
\begin{lemma}\label{Lemma: Cap-Price}
	For any feasible contract $\Phi(\Sbeta,\Stheta)$, $\dcap_i<\dcap_j$ if and only if  $\ \pcap_i<\pcap_j$.
\end{lemma}
%==========================================
%==========================================
\begin{lemma}\label{Lemma: type-Cap}
	For any feasible contract $\Phi(\Sbeta,\Stheta)$, if $\slope(\dcap,\type_i)>\slope(\dcap,\type_j)$ for all $\dcap$, then $\dcap_i\ge\dcap_j$.
\end{lemma}
%==========================================

Lemma \ref{Lemma: Cap-Price} reveals that a larger data cap corresponds to a higher subscription fee in the feasible contract, which is intuitive.
Lemma \ref{Lemma: type-Cap} shows that a user with a stronger \textit{willingness-to-pay} for the data cap deserves a larger data cap in the feasible contract.
Next we provide a proof sketch for Lemma \ref{Lemma: type-Cap} to show the key insights.

\begin{proof}[\textbf{Proof Sketch of Lemma \ref{Lemma: type-Cap}}]
We illustrate the key insights of Lemma \ref{Lemma: type-Cap} based on the contract plane in Fig. \ref{fig: illustration_slope}.
\begin{itemize}
	\item For a type-$\type_j$ user, we assume that the red dot in Fig. \ref{fig: illustration_slope} is the contract item $\contract_j$ intended for this user type, and the red circle curve $l_j$ represents his indifference curve with a payoff equal to that of selecting $\contract_j$.
	\item For a type-$\type_i$ user, the blue square curve $l_i$ is his indifference curve with a payoff equal to this user choosing the red dot contract item $\contract_j$ (\textit{not intended for his type}).
\end{itemize}

It is obvious that $l_i$ is steeper than $l_j$; mathematically speaking, $\slope(\dcap,\type_i)>\slope(\dcap,\type_j)$ for all $\dcap$ (which is the condition in Lemma \ref{Lemma: type-Cap}).
That is, comparing with the type-$\type_j$ users, the type-$\type_i$ users have a stronger willingness-to-pay under any data cap.
Moreover, as a user's indifference curve shifts downward, his payoff increases because of the decreasing subscription fee.

\begin{figure}
	\centering
			\setlength{\abovecaptionskip}{2pt}
			\setlength{\belowcaptionskip}{0pt}
	\includegraphics[width=0.6\linewidth]{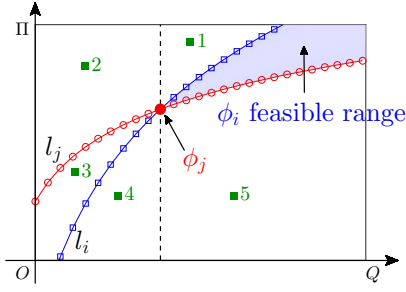}
	\caption{An illustration for Lemma \ref{Lemma: type-Cap}.}
	\label{fig: illustration_slope}
\end{figure}

Next we will show that \textit{to ensure the PIC condition $\contract_i\ic\contract_j$, the contract item $\contract_i$ (intended for the type-$\type_i$ users) must locate \textit{below} (or on) the blue square curve $l_i$ and \textit{above} (or on) the red circle curve $l_j$, i.e., in the blue region of Fig. \ref{fig: illustration_slope}}.
We prove this by contradiction. 
Assuming that this is  not true, then we need to consider the following two scenarios: 
\begin{itemize}
	\item \textbf{Scenario 1:} The contract item $\contract_i$ is above the blue square curve $l_i$, such as the green squares labeled 1, 2, 3 in Fig. \ref{fig: illustration_slope}. 
	%The square curve $\hat{l}_i$ is the corresponding indifference curve of the type-$\type_i$ user with a payoff equal to that of selecting $\contract_i$. 
	%In this case, $\hat{l}_i$  is above $l_i$, which implies that  $\payoff(\contract_j,\type_i)>\payoff(\contract_i,\type_i)$. 
	In this case, the indifference curve $l_i$ for the type-$\type_i$  should shift upward (with a decreasing payoff) to touch one of the three green squares.
	However, the type-$\type_i$ user can achieve a higher payoff (comparing with selecting $\contract_i$) by selecting the red dot contract item $\contract_j$, which violates the PIC condition for the type-$\type_i$ user.
	\item \textbf{Scenario 2:}  The contract item $\contract_i$ is below the red circle curve ${l}_j$, such as the green squares labeled 4 and 5 in Fig. \ref{fig: illustration_slope}. 
	%The triangle curve $\hat{l}_j$ is the corresponding indifference curve of the type-$\type_j$ user with a payoff equal to that of selecting $\contract_i$. 
	%In this case, $\hat{l}_j$ is below $l_j$, which means that $\payoff(\contract_i,\type_j)>\payoff(\contract_j,\type_j)$. 
	In this case, the indifference curve $l_j$ for the type-$\type_j$ user  should shift downward (with an increasing payoff) to touch one of the three green squares.
	Therefore, the type-$\type_j$ user can achieve a higher payoff by selecting the green square contract item $\contract_i$, which violates the PIC condition for the type-$\type_j$ users.
\end{itemize}

The above discussion indicates that the contract item $\contract_i$ must locate in the blue area, which is on the right of the dash line.
Thus $\dcap_i\ge\dcap_j$, as Lemma \ref{Lemma: type-Cap} implies.	
\end{proof}

According to Lemma \ref{Lemma: Cap-Price} and Lemma \ref{Lemma: type-Cap}, we summarize the \textit{necessary conditions} for a feasible contract as follows:
\begin{theorem}[Necessary Conditions for Feasibility]\label{Theorem: necessary conditions type}
	The feasible contract $\Phi(\Sbeta,\Stheta)$ has the following structure
	\begin{equation}\label{Equ: necessary conditions type}
	\left\{
	\begin{aligned}
	&\dcap_1\le\dcap_2\le...\le\dcap_{KM} ,\\
	&\pcap_1\le\pcap_2\le...\le\pcap_{KM}.
	\end{aligned}
	\right.
	\end{equation}
\end{theorem}

\subsection{Sufficient Conditions \label{SubSection: Sufficient Condition}}
Next we derive the sufficient conditions for the feasible contract though the following two transitivity properties for Pairwise Incentive Compatibility (PIC) and Individual Rationality (IR).
The proofs are given in Appendix \ref{Appendix: B}.
%==========================================
\begin{lemma}[PIC-Transitivity]\label{Lemma: IC-Transitivity}
	Suppose the necessary conditions in Theorem \ref{Theorem: necessary conditions type} hold, then for any $i_1<i_2<i_3$, the following is true
	\begin{equation}
	\text{if}\ \contract_{i_1}\ic\contract_{i_2} \text{ and }\contract_{i_2}\ic\contract_{i_3}, \text{ then }\contract_{i_1}\ic\contract_{i_3}.
	\end{equation}
\end{lemma}
%==========================================

The above PIC transitivity property makes the contract problem (i.e., Problem \ref{Problem: Optimal Contract}) more tractable.
It shows that we can reduce a total of  ${KM(KM-1)/2}$ PIC conditions to  a total of  ${KM-1}$ PIC conditions for the neighbor user type pairs, i.e., $\phi_i\ic\phi_{i+1},\ i=1,2,...,{KM-1}$.

We presents the IR transitivity in Lemma \ref{Lemma: IR-Transitivity}.
%==========================================
\begin{lemma}[IR-Transitivity]\label{Lemma: IR-Transitivity}
	Suppose the necessary conditions in Theorem \ref{Theorem: necessary conditions type} and all PIC conditions hold, then the following is true, 	
	$$\text{if } \phi_{\epsilon}\ir, \text{ then } \phi_i\ir,\ \forall\ i\ne \epsilon.$$
\end{lemma}
%==========================================

Recall that the user type-$\type_{\epsilon}$, defined in (\ref{Equ: smallest-payoff user}), achieves the smallest payoff among all the user types for any given contract item.
Lemma \ref{Lemma: IR-Transitivity} implies that once we can guarantee all the PIC conditions, then we only need to further ensure that the IR constraint for the \textit{smallest-payoff} type-$\type_{\epsilon}$ users. 
This allows us to reduce a total of  $KM$ IR conditions to one IR condition $\phi_{\epsilon}\ir$.

Before we present the \textit{sufficient conditions}  for the feasible contract, we first introduce a user's \textit{virtual payoff increment}.
Recall that $L(\dcap,\type_i)$ defined in (\ref{Equ: Virtual Payoff}) denotes the type-$\type_i$ user's virtual payoff.
We define $\eta^-(\type_i,\dcap_i,\dcap_{i-1})$ and $\eta^+(\type_i,\dcap_i,\dcap_{i+1})$ as the type-$\type_i$ user's virtual payoff increments between selecting the contract item $\contract_i$ and the contract items intended for his neighbor user types (i.e., $\contract_{i-1}$ and $\contract_{i+1}$), as follows
\begin{subequations}\label{Equ: Virtual payoff increament}
\begin{align}
& \eta^-(\type_i,\dcap_i,\dcap_{i-1})=L(\dcap_i,\type_i)-L(\dcap_{i-1},\type_i), \label{Equ: Virtual payoff increament eta-} \\
& \eta^+(\type_i,\dcap_i,\dcap_{i+1})=L(\dcap_i,\type_i)-L(\dcap_{i+1},\type_i).\label{Equ: Virtual payoff increament eta+}
\end{align}
\end{subequations}	

Based on Lemmas \ref{Lemma: Cap-Price}$\sim$\ref{Lemma: IR-Transitivity}, we derive the following \textit{sufficient conditions} for a contract to be feasible.
\begin{theorem}[Sufficient Conditions for Feasibility]\label{Theorem: sufficien condition}
	The contract $\Contrset(\Sbeta,\Stheta)$ is feasible if all the following conditions hold,
	\begin{enumerate}
		\item $\dcap_1\le\dcap_2\le...\le\dcap_{KM}$,
		%%%%%%%%%%%%%%%%%%%%%%%%%%%%%%%%%%%%%%%%%%%%%%%%%%%%%%%%%%%
		\item for $i=\epsilon$,
		\begin{equation}\label{Equ: Sufficient IR u}
		\pcap_{\epsilon}\le L(\dcap_{\epsilon},\type_{\epsilon}),
		\end{equation}
		
		%%%%%%%%%%%%%%%%%%%%%%%%%%%%%%%%%%%%%%%%%%%%%%%%%%%%%%%%%%%
		\item for all $i=1,2,...,\epsilon-1$,
		\begin{subequations}\label{Equ: Sufficient IC k<u}
			\begin{align}
			& \pcap_i\le \pcap_{i+1}+ \eta^+(\type_i,\dcap_i,\dcap_{i+1}),			\label{SubEqu: Sufficient IC k<u 1st} \\
			& \pcap_i\ge \pcap_{i+1}- \eta^-(\type_{i+1},\dcap_{i+1},\dcap_i).		\label{SubEqu: Sufficient IC k<u 2nd}
			\end{align}	
		\end{subequations}
				
		%%%%%%%%%%%%%%%%%%%%%%%%%%%%%%%%%%%%%%%%%%%%%%%%%%%%%%%%%%%
		\item for all $i=\epsilon+1,\epsilon+2,...,{KM}$,
		\begin{subequations}\label{Equ: Sufficient IC k>u}
		\begin{align}
		& \pcap_i\le \pcap_{i-1}+ \eta^-(\type_i,\dcap_i,\dcap_{i-1}),     \label{SubEqu: Sufficient IC k>u 1st} \\
		& \pcap_i\ge \pcap_{i-1}- \eta^+(\type_{i-1},\dcap_{i-1},\dcap_i), \label{SubEqu: Sufficient IC k>u 2nd}
		\end{align}	
		\end{subequations}

	\end{enumerate}
\end{theorem}

Now we discuss the intuitions of Theorem \ref{Theorem: sufficien condition}. 
\textit{Condition 1)} satisfies the necessary conditions in Theorem \ref{Theorem: necessary conditions type}.
\textit{Condition 2)} guarantees the IR condition for the type-$\type_{\epsilon}$ users, i.e., $\contract_{\epsilon}\ir$, which is sufficient for the IR conditions of all other user types according to Lemma \ref{Lemma: IR-Transitivity}.
\textit{Condition 3)} and \textit{Condition 4)} guarantee the PIC condition for the neighbor user types, i.e., $\contract_i\ic\contract_{i+1},\ \forall\ 1\le i\le{KM-1}$, which is sufficient for the global IC condition according to Lemma \ref{Lemma: IC-Transitivity}.
Specifically, the inequality (\ref{SubEqu: Sufficient IC k<u 1st}) ensures that the type-$\type_i$ user will not select the contract item $\contract_{i+1}$, i.e., $\payoff(\contract_i,\type_i)\ge\payoff(\contract_{i+1},\type_i)$; 
the inequality (\ref{SubEqu: Sufficient IC k<u 2nd}) ensures the type-$\type_{i+1}$ user will not select the contract item $\contract_{i}$, i.e., $\payoff(\contract_{i+1},\type_{i+1})\ge\payoff(\contract_{i},\type_{i+1})$.
Similar intuitions apply to (\ref{Equ: Sufficient IC k>u}).

So far we have derived the necessary and sufficient conditions for a feasible contract.
Next we will analyze the optimality of the contract.

\section{Contract Optimality \label{Section: Contract Optimality}}  
We will study the MNO's optimal contract problem (i.e., Problem \ref{Problem: Optimal Contract}) based on the necessary and sufficient conditions for a feasible contract.
To reveal the key insights, we will investigate the contract optimality in the following two steps.
\begin{itemize}
	\item First, in Problem \ref{Problem: Optimal Pricing type}, we derive the MNO's optimal prices $\{\pcap_i^*(\bm{\dcap}),1\leq i\leq{KM}\}$ given a \textit{feasible} choice of  data caps $\bm{\dcap}=\{\dcap_i,1\le i\le KM\}$ where $\dcap_1\le\dcap_2\le...\le\dcap_{KM}$. 
	\item Second, in Problem \ref{Problem: Optimal cap type}, we substitute the optimal prices $\{\pcap_i^*(\bm{\dcap}),1\leq i\leq{KM}\}$ to the MNO's profit function and derive the optimal data cap $\bm{\dcap}^*=\{\dcap_i^*,1\le i\le{KM}\}$.
\end{itemize}

\subsection{Optimal Pricing}
In Problem \ref{Problem: Optimal Pricing type}, we compute  the MNO's optimal prices, denoted by $\{\pcap_i^*(\bm{\dcap}), 1\le i \le {KM}\}$, given a feasible data cap allocation $\bm{\dcap}$, i.e., $\dcap_1\le\dcap_2\le...\le\dcap_{KM}$.
Note that the constraints (\ref{Equ: Sufficient IR u}), (\ref{Equ: Sufficient IC k>u}), and (\ref{Equ: Sufficient IC k<u}) are the sufficient conditions in Theorem \ref{Theorem: sufficien condition}.
Hence the solution $\{\pcap_i^*(\bm{\dcap}), 1\le i \le {KM}\}$ together with the given data cap $\bm{\dcap}$ must be a feasible contract.

\begin{problem}[Optimal Prices]\label{Problem: Optimal Pricing type}
	\begin{equation}\label{Equ: Optimal price type}
	\begin{aligned}
	&\max  \sum\limits_{i=1}^{KM} q(\type_i) \Big[ \pcap_i+P(\dcap_i,\type_i) -c\cdot U(\dcap_i,\type_i)-J( \dcap_i ) \Big] \\
	& \ \textit{s.t. }\ (\ref{Equ: Sufficient IR u}),(\ref{Equ: Sufficient IC k<u}),(\ref{Equ: Sufficient IC k>u})\\
	& \ \textit{var. }\ \pcap_i,1\le  i\le {KM}. 
	\end{aligned}
	\end{equation}
\end{problem} 

Next we characterize the optimal prices $\{\pcap_i^*(\bm{\dcap}), 1\le i \le {KM}\}$ in Theorem \ref{Theorem: Optimal Pricing type}.
The proof is given in Appendix \ref{Appendix: Optimal Pricing}.

%==========================================
\begin{theorem}[Optimal Pricing Policy]\label{Theorem: Optimal Pricing type}
	Given a set of feasible data caps $\bm{\dcap}$ satisfying $\dcap_1\le\dcap_2\le...\le\dcap_{KM}$. 
	The optimal pricing policy for the MNO, denoted by $\{\pcap_i^*(\bm{\dcap}),1\le i\le {KM}\}$, is
	\begin{subequations}\label{Equ: Optimal Pricing Policy}
	\begin{align}
	& \pcap^*_i(\bm{\dcap})=L\left(\dcap_i,\type_i\right),  &\text{if } i= \epsilon,\label{Equ: Optimal Pricing Policy 1}\\
	& \pcap^*_i(\bm{\dcap})= \pcap^*_{i+1}(\bm{\dcap})+\eta^+\left(\type_i,\dcap_i,\dcap_{i+1}\right), 	&\text{if } i< \epsilon,\\
	& \pcap^*_i(\bm{\dcap})= \pcap^*_{i-1}(\bm{\dcap})+\eta^-\left(\type_i,\dcap_i,\dcap_{i-1}\right),	&\text{if } i> \epsilon.	
	\end{align}
	\end{subequations}	
\end{theorem}
%==========================================

Comparing Theorem \ref{Theorem: sufficien condition} and Theorem \ref{Theorem: Optimal Pricing type}, we notice that, given a set of feasible data caps $\bm{\dcap}$, the MNO should charge the highest prices satisfying  the IC and IR conditions.

Next we further study the MNO's optimal data caps $\bm{\dcap}$ based on the optimal prices $\{\pcap_i^*(\bm{\dcap}),1\le i\le {KM}\}$ in (\ref{Equ: Optimal Pricing Policy}).

\subsection{Optimal Data Caps}
For notation simplicity, we first introduce the concept of \textit{virtual payoff difference}.
For a given data cap $\dcap$, the virtual payoff differences between the type-$\type_i$ user and his neighbor user types (i.e., $\type_{i-1}$ and $\type_{i+1}$) are defined as
\begin{subequations}\label{Equ: virtual payoff differences}
\begin{align}
& \rho^-_i(\dcap) \triangleq L( \dcap,\type_i )-L( \dcap,\type_{i-1} ) ,  \\
& \rho^+_i(\dcap) \triangleq L( \dcap,\type_i )-L( \dcap,\type_{i+1} ) .
\end{align}
\end{subequations}

We substitute the optimal prices (\ref{Equ: Optimal Pricing Policy}) derived in Theorem \ref{Theorem: Optimal Pricing type} into the objective function of Problem \ref{Problem: Optimal Pricing type}, and write the MNO's objective function (i.e., the total profit) as follows: 
\begin{equation}
\begin{aligned}
\sum_{i=1}^{KM} G_i\left(\dcap_i\right) ,
\end{aligned}
\end{equation}
where $G_i(\cdot)$ is given by (\ref{Equ: Gi type}), and $h^i=\textstyle\sum_{t=1}^{i-1}q(\type_t)$ and $h_i=\sum_{t=i+1}^{KM}q(\type_t)$ are two constants related to the distribution of the user types.
Thus we  get the following optimization problem over the $KM$ data caps.
 
\begin{figure*}
\begin{equation}\label{Equ: Gi type}
\begin{aligned}
& G_i(\dcap)= \left\{
\begin{aligned}
%========================================================================================================================
& q(\type_i)V(\dcap,\type_i) -q(\type_i)\left[ c\cdot U(\dcap,\type_i ) + J(\dcap) \right], \qquad\qquad\qquad\qquad\quad\ \ \text{if } i\in\{1,KM\}, \\
%========================================================================================================================
& q(\type_i)V(\dcap,\type_i) + h^i  \rho^-_i(\dcap) -q(\type_i)\left[ c U(\dcap,\type_i ) + J(\dcap) \right], \qquad\qquad\quad\  \text{if } i\in\{2,3,...,\epsilon-1\}, \\
%========================================================================================================================
& q(\type_i)V(\dcap,\type_i) + h^i \rho^-_i(\dcap)+ h_i \rho^+_i(\dcap) -q(\type_i)\left[ c U(\dcap,\type_i ) + J(\dcap) \right], \ \text{if } i={\epsilon}, \\
%========================================================================================================================
& q(\type_i)V(\dcap,\type_i) + h_i \rho^+_i(\dcap) -q(\type_i)\left[ c U(\dcap,\type_i ) + J(\dcap) \right], \qquad\qquad\quad\ \   \text{if } i\in\{\epsilon+1,\epsilon+2,...,KM-1\},
%========================================================================================================================
%	& p(\type_{i})V(\dcap,\type_{i})-p(\type_{i})C(\dcap), & i={KM}. %========================================================================================================================
\end{aligned}
\right. 
\end{aligned}
\end{equation}
\hrulefill
\vspace{4pt}
\end{figure*}

\begin{problem}[Optimal Data Caps]\label{Problem: Optimal cap type}
	\begin{subequations}\label{Equ: Optimal cap type}
	\begin{align} 
	&  \max  \sum\limits_{i=1}^{KM}G_i\left(\dcap_i\right) 					\label{SubEqu: Optimal cap type, objective}	\\
	& \ \textit{s.t. }\ \ \dcap_1\le\dcap_2\le...\le\dcap_{KM}\le\dmax 		\label{SubEqu: Optimal cap type, monotonicity}	 \\
	& \quad\quad\ \dcap_i\in\mathbb{N}, \ \forall\ i\in\{1,2,...,KM\} 	\label{SubEqu: Optimal cap type, integel cap}	 \\
	& \ \textit{var. }\  \dcap_i,1\le i\le {KM}.		
	\end{align}
	\end{subequations}
\end{problem}

Problem \ref{Problem: Optimal cap type} is a nonlinear integer programming with two special structures.
First, the objective function has a separable structure over each decision variable $\dcap_i$.
Second, the decision variables are monotonic.
Moreover, the convexity of Problem \ref{Problem: Optimal cap type} depends on all user types $\type_i$ for all $1\le i\le KM$ and the corresponding distribution $q(\type_i)$ for all $1\le i\le KM$.

In previous literature (e.g., \cite{gao2011spectrum,duan2014cooperative,zhang2015contract,zhang2017non}), the commonly used approach to solving Problem \ref{Problem: Optimal cap type} is \textit{monotonicity relaxation}.
The main idea is to first relax the monotonicity constraints (\ref{SubEqu: Optimal cap type, monotonicity}) and maximize each $G_i(\cdot)$ over the corresponding decision variable $\dcap_i$. 
If the solution obtained under the relaxation  violates the monotonicity constraints (\ref{SubEqu: Optimal cap type, monotonicity}), then one needs to adjust  the solution according to the algorithm proposed in \cite{gao2011spectrum}  to become feasible.
We refer interested readers to Appendix \ref{Appendix: Monotonicity Relaxation} for more details.
In general, the monotonicity relaxation approach is very efficient, since it only needs to deal with several single-variable optimization problems. 
However, the adjusted solution is only a locally optimal solution when the problem is not convex \cite{Zhiyuan2018MobiHoc}.
Moreover, it is difficult to analytically characterize the sub-optimality gap of the solution. 
To obtain the globally optimal solution of Problem \ref{Problem: Optimal cap type} efficiently, in Section  \ref{Subsection: Dynamic Quota Assignment Algorithm}, we will propose the Dynamic Quota Allocation Algorithm, which is one of the major contributions in this paper.

\subsection{Dynamic Quota Allocation (DQA) Algorithm\label{Subsection: Dynamic Quota Assignment Algorithm}}

\subsubsection{Basic Idea}
The basic idea of the DQA Algorithm comes from dynamic programming, i.e., breaking the original problem down into simpler sub-problems in a recursive manner \cite{bellman2013dynamic}.
Specifically, we will decompose  Problem \ref{Problem: Optimal cap type} by utilizing the separability of objective (\ref{SubEqu: Optimal cap type, objective}) and the monotonicity constraints (\ref{SubEqu: Optimal cap type, monotonicity}).
Next we introduce how to define the proper sub-problems.
\subsubsection{Level-($n,q$) Subproblem}
In the DQA Algorithm, we refer to Problem \ref{Problem: sub} as the level-($n,q$) sub-problem of Problem \ref{Problem: Optimal cap type}.
Basically, the level-($n,q$) sub-problem focuses on the optimal data caps for the smallest $n$ user types (i.e., type-$1$ to type-$n$, where $1\leq n \leq KM$) under the data cap upper bound $q$ ($0\le q\le \dmax$).
Recall that there are a total of $KM$ types of users and $\dmax$ is users' maximal possible monthly data demand.
The special case of the  level-($KM,\dmax$) sub-problem is equivalent to Problem \ref{Problem: Optimal cap type}, since the MNO does not need to offer any data cap larger than $\dmax$.

\begin{problem}[Level-$(n,q)$ Sub-problem]\label{Problem: sub}
	Given $1\le n \le KM$ and $0\le q\le \dmax$, the level-$(n,q)$ sub-problem is
	\begin{subequations}\label{Equ: subproblem}
		\begin{align}
		H(n,q)\triangleq\arg\max\   & \sum_{i=1}^{n}G_i( \dcap_i ) 						\label{SubEqu: subproblem, objective}\\
		\text{s.t. }\ & \dcap_1\le\dcap_2\le...\le\dcap_n\le q 	\label{SubEqu: subproblem, monotonicity} \\
		& \dcap_i\in\mathbb{N},\ \forall\ i\in\{1,2,...,n\} \\
		\text{var: }	&  \dcap_i,1\le i\le n .
		\end{align}
	\end{subequations}
	
\end{problem}

Here we denote $H(n,q)$ and $\bm{\dcap}^\star(n,q)=\{\dcap^\star_i(n,q),1\le i\le n\}$ as the optimal value and the optimal solution of the level-($n,q$) sub-problem (\ref{Equ: subproblem}), respectively.	
Since the level-($KM,\dmax$) sub-problem is equivalent to Problem \ref{Problem: Optimal cap type}, we have
\begin{itemize}
	\item The optimal value of Problem \ref{Problem: Optimal cap type} is $H(KM,\dmax)$.
	\item The optimal data caps in Problem \ref{Problem: Optimal cap type} is $\bm{\dcap}^\star(KM,\dmax)$, i.e.,
	$\dcap_i^*=\dcap_i^\star(KM,\dmax)$ for all $1\le i\le KM$.
	
\end{itemize}

In the following, we will show that if we know $H(n,q)$ for all $1\le n\le KM$ and $0\le q\le \dmax$, then we can directly find $\bm{\dcap}^\star(KM,\dmax)$.
To present this connection clearly, we first introduce some properties of $H(n,q)$ in Propositions \ref{Proposition: function H Recursiveness} and  \ref{Proposition: function H Monotonicity}.
The proofs are given in Appendix \ref{Appendix: H(n,q)}.

\begin{proposition}\label{Proposition: function H Recursiveness}
	For any $2\le n\le KM$ and $0\le q\le \dmax$, $H(n,q)$ has the following recursive relation
	\begin{subequations}\label{Equ: H Recursiveness}
	\begin{align}
	H(n,q) = \max\limits_{x\in\mathbb{N}}\ & H(n-1, x ) + G_n( x ) \\
	\text{s.t. } & x\le q.\label{Equ: H Recursiveness constrait}
	\end{align}
	\end{subequations}
\end{proposition}

The proof of Proposition \ref{Proposition: function H Recursiveness} follows the definition of the level-($n,q$) sub-problem in (\ref{Equ: subproblem}).
\begin{proposition}\label{Proposition: function H Monotonicity}
	Given any $n\in\{1,2,...,KM\}$, we have
	\begin{itemize}
		\item Function  $H(n,q)$ is non-decreasing in the data cap $q$.
		\item There exists a critical point $\hat{q}_n$ such that $H(n,q)$ does not change for any $q\ge\hat{q}_n$.
	\end{itemize}
	
\end{proposition}

The intuitions behind Proposition \ref{Proposition: function H Monotonicity} are two-fold.
\begin{itemize}
	\item First, the non-decreasing property of $H(n,q)$ results from the constraints (\ref{SubEqu: subproblem, monotonicity}).
	Mathematically, $q$ in (\ref{SubEqu: subproblem, monotonicity}) defines the domain upper bound of the level-($n,q$) sub-problem.
	That is, a larger $q$ in (\ref{SubEqu: subproblem, monotonicity}) corresponds to a larger feasible domain, hence a no smaller optimal value $H(n,q)$.
	\item Second, $H(n,q)$ will not increase in $q$ anymore if the optimal solution of the level-($n,q$) sub-problem is smaller than the domain upper bound $q$. 
%	In this case, further increasing the upper bound $q$ will not increase $H(n,q)$.
	Basically, $\hat{q}_n$ equals to the $n$-th element of the optimal solution $\bm{\dcap}^\star(n,\dmax)$ for the level-($n,\dmax$) sub-problem, i.e.,
	\begin{equation}\label{Equ: critical point q_hat}
	\hat{q}_n=\dcap_n^\star(n,\dmax).
	\end{equation}
\end{itemize}

%\begin{proposition}\label{Proposition: H and Problem solution}
%	Denote $\{\dcap^*_i,1\le i \le KM\}$ the optimal data caps of Problem \ref{Problem: Optimal cap type}.
%	Then we have 
%	\begin{equation}
%	\begin{aligned}
%	H(KM,\dmax) = H\left(n-1,{\dcap}^*_{n}\right)+\sum_{i=n}^{KM}G_i\left({\dcap}^*_i\right),\ \forall n\ge2.
%	\end{aligned}
%	\end{equation}
%\end{proposition}

Based on the recursiveness shown in Propositions \ref{Proposition: function H Recursiveness} and the critical points $\{\hat{q}_i,1\le i\le KM\}$ shown in Proposition \ref{Proposition: function H Monotonicity}, we are able to find the optimal solution of the level-($KM,\dmax$) sub-problem (which is the same as  Problem \ref{Problem: Optimal cap type}) according to  Theorem \ref{Theorem: solution KM,D}.
The proof is given in Appendix \ref{Appendix: H(n,q)}.
\begin{theorem}\label{Theorem: solution KM,D}
	The optimal solution $\{\dcap^\star_i(KM,\dmax),1\le i \le KM\}$ of the level-($KM,\dmax$) sub-problem is  
	\begin{equation}\label{Equ: solution KM,D}
	\dcap_i^\star(KM,\dmax) = 
	\begin{cases}
	\hat{q}_i		,									&\text{if } i=KM,\\
	\min\{ \hat{q}_i,\dcap_{i+1}^\star(KM,\dmax) \},		&\text{if } i< KM.
	\end{cases}
	\end{equation}
\end{theorem}

We elaborate Theorem \ref{Theorem: solution KM,D} as follows:
\begin{itemize}
	\item For $i=KM$, according to (\ref{Equ: critical point q_hat}), we know that the optimal data cap $\dcap_{KM}^\star(KM,\dmax)$ is the same as the critical point mentioned in Proposition \ref{Proposition: function H Monotonicity}, i.e., $\dcap_{KM}^\star(KM,\dmax)=\hat{q}_{KM}$.
	
	\item For the other user types, i.e., $i<KM$, according to (\ref{Equ: H Recursiveness constrait}) in Proposition \ref{Proposition: function H Recursiveness}, the optimal data cap $\dcap_i^\star(KM,\dmax)$ is the smaller one between the critical point $\hat{q}_i$ and the optimal data cap for the next user type  $\dcap_{i+1}^\star(KM,\dmax)$.
\end{itemize}

Here we want to emphasize that (\ref{Equ: solution KM,D}) only needs the $KM$ critical points $\{\hat{q}_i,1\le i\le KM\}$, which can be easily obtained from the table of $H(n,q)$ for all $1\le n\le KM$ and $0\le q\le \dmax$.
In Appendix \ref{Appendix: Numerical example}, we provide a numerical example to demonstrate how to find the optimal data caps based on the table $H(n,q)$.

The remaining question is how to compute the table of $H(n,q)$. 
We solve this problem by proposing the DQA Algorithm next. 

\begin{algorithm}  
	\caption{Dynamic Quota Allocation (DQA)}\label{Algorithm: bottom up} 
	\SetKwInOut{Input}{Input}
	\SetKwInOut{Output}{Output}
	\Input{All user types $\type_i$ and the distribution $q(\type_i)$.}  
	\Output{ $H(n,q)$ for all $1\le n\le  KM$ and $0\le q\le \dmax$. } 
	\textbf{Initial} $H(n,q)=0,\ \forall\ 1\le n \le KM$ and $0\le q\le \dmax$.	\\
	\For {$n=1$ \KwTo $KM$ } 
	{
		\For{$q=0$ \KwTo $\dmax$} 
		{
			\eIf{$n = 1$}
			{
				$H(n,q):= \max\limits_{x\le q} G_n(x)$. \\
				%				${\bf H}[n,m]:=  G_n\left({\bm s}[n]\right)$. \\
			}
			{
				$H(n,q):= \arg\max\limits_{x\le q} H(n-1,x)+ G_n(x) $. \\
				%				${\bf H}[n,m]:=  G_n\left({\bm s}[n]\right)$. \\				
			}			
		}
	}
	%	Compute the optimal data caps according to (\ref{Equ: solution KM,D}).
\end{algorithm}  

\begin{figure*}
	\centering
	\begin{minipage}{0.3\textwidth}
		\centering
	\setlength{\abovecaptionskip}{3pt}
	\setlength{\belowcaptionskip}{0pt}
		\includegraphics[width=0.93\linewidth]{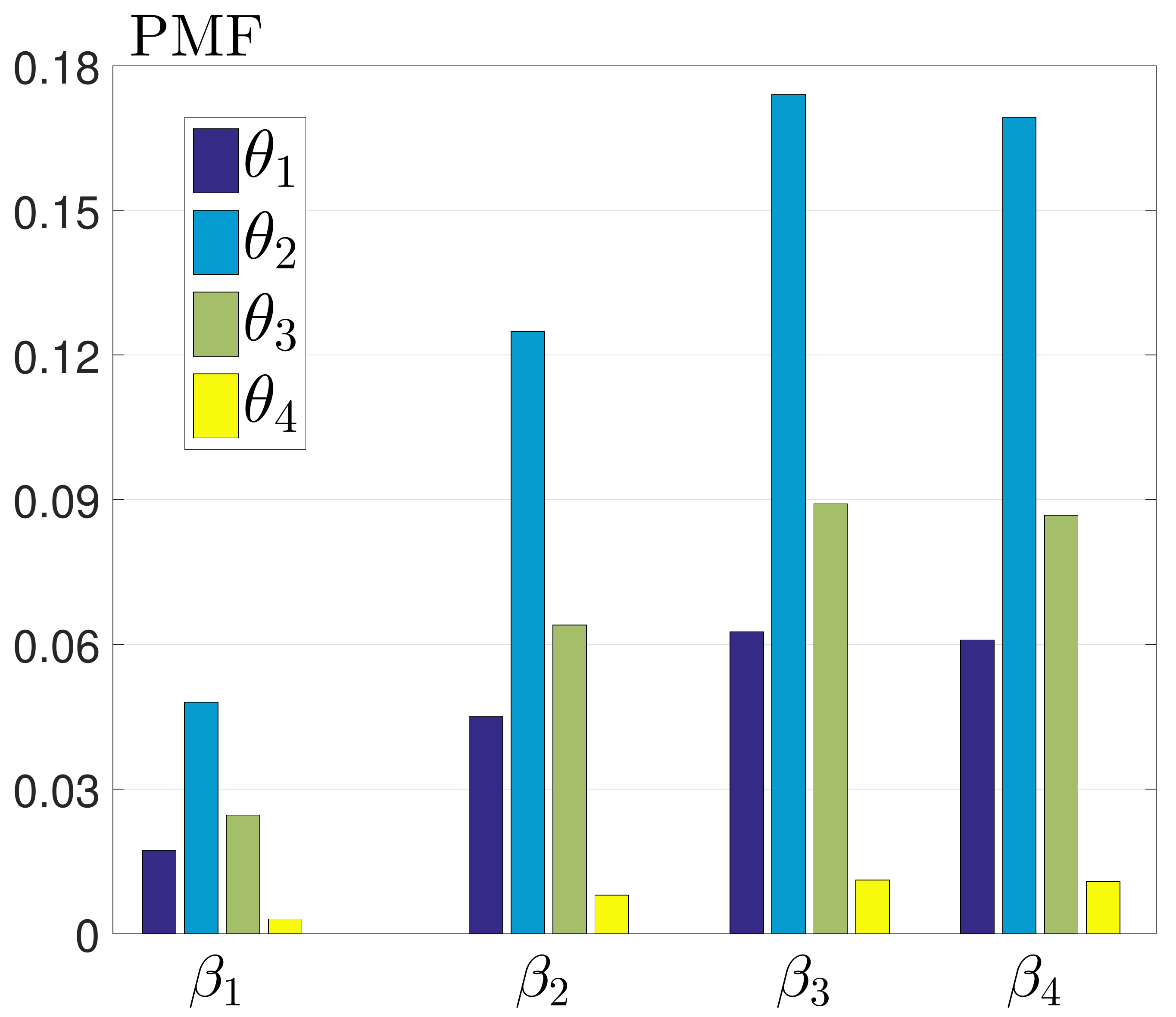}
		\caption{User type distribution.}
		\label{fig: PDF}
	\end{minipage}\qquad
	\begin{minipage}{0.3\textwidth}
		\centering
		\setlength{\abovecaptionskip}{3pt}
		\setlength{\belowcaptionskip}{0pt}
		\includegraphics[width=0.93\linewidth]{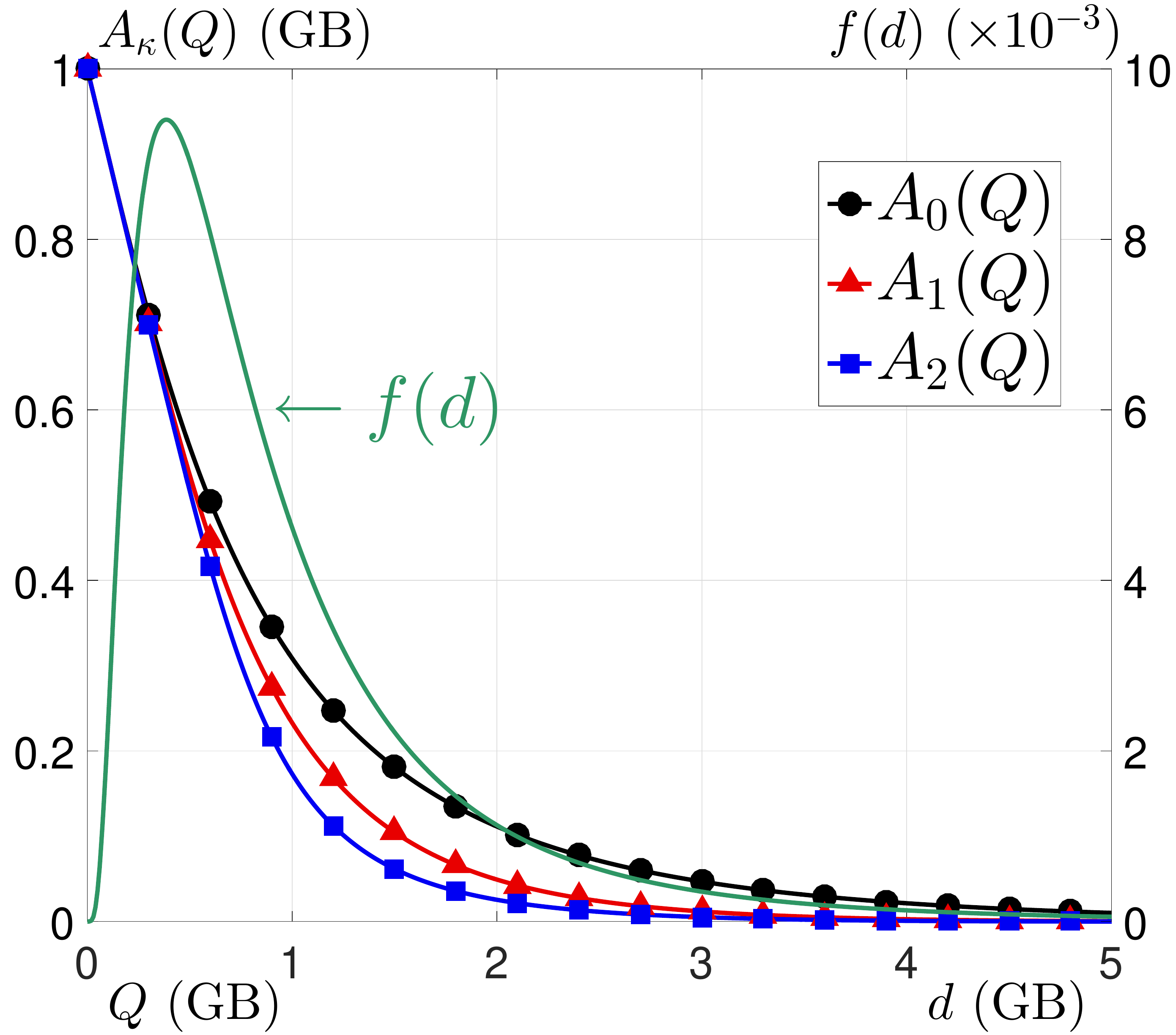}
		\caption{$f(d)$ vs. $d$ and $A_{\mechanism}(\dcap)$ vs. $\dcap$.}
		\label{fig: Demand}
	\end{minipage}\qquad
	\begin{minipage}{0.3\textwidth}
		\centering
		\setlength{\abovecaptionskip}{3pt}
		\setlength{\belowcaptionskip}{0pt}
		{\includegraphics[width=0.93\linewidth]{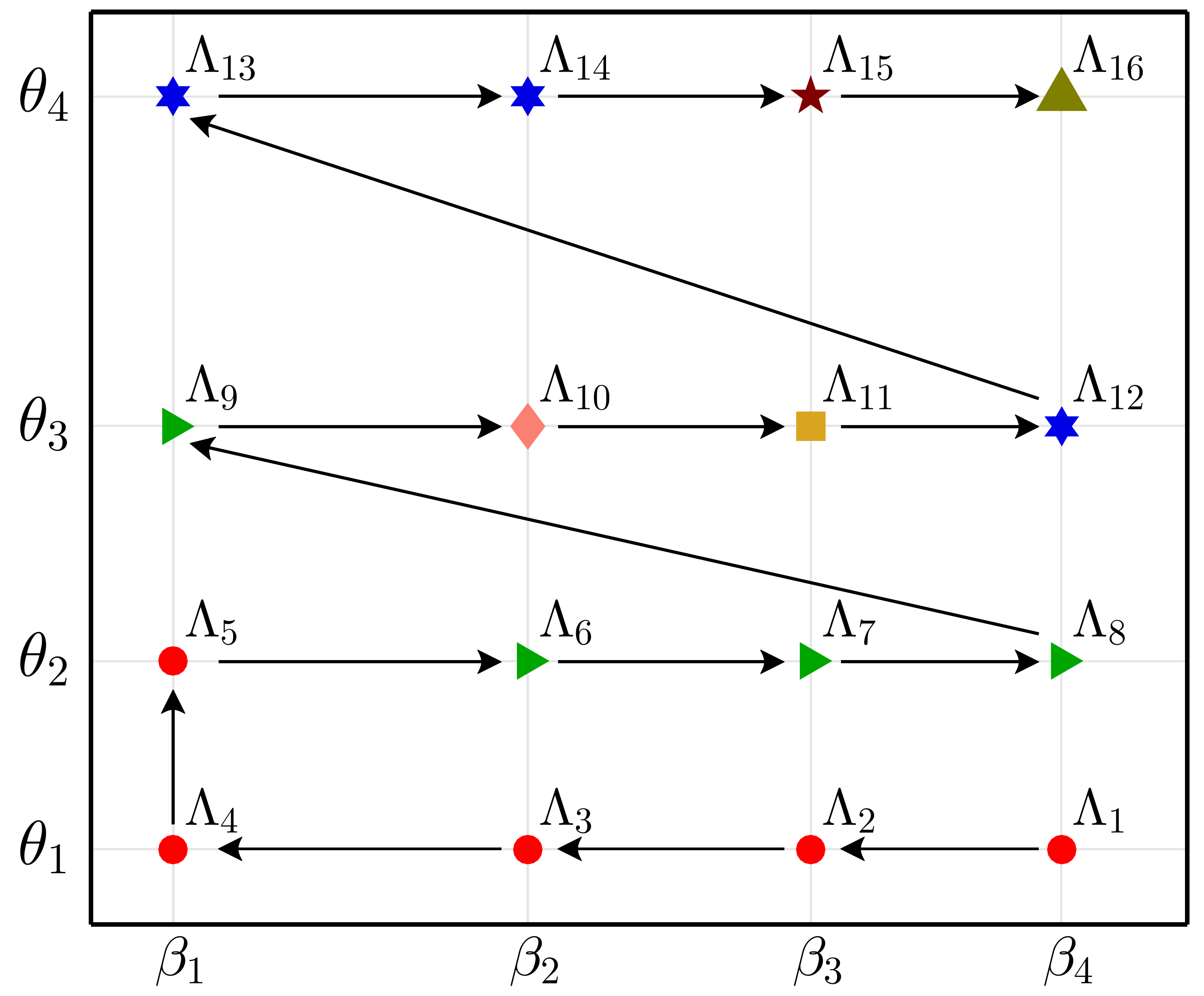}} 
		\caption{Structure of optimal contract.}
		\label{fig: CapStructure}
	\end{minipage}  
\end{figure*}

\subsubsection{DQA Algorithm}
To compute $H(n,q)$ efficiently, we need to take the advantage of its recursiveness (in Proposition \ref{Proposition: function H Recursiveness}) again.
The detailed process is shown in Algorithm \ref{Algorithm: bottom up}.
Specifically, the input of Algorithm \ref{Algorithm: bottom up} includes all of the user types $\type_i$ and the corresponding distribution (probability mass function)  $q(\type_i)$.
The output of this algorithm is the table of $H(n,q)$ for all $1\le n\le KM$ and $0\le q\le \dmax$.
In Line 5, we compute $H(1,q)$ for all $q\in\{0,1,2,...,\dmax\}$.
In Line 7, we compute $H(n,q)$ for all $n\ge2$ by utilizing the recursiveness in Proposition \ref{Proposition: function H Recursiveness}.

Algorithm \ref{Algorithm: bottom up} has a computational complexity of $O(KM|\mathcal{Q}|^2)$, where $KM$ is the number of user types and the set $\mathcal{\dcap}$ consists of all the possible data caps. 
It is actually quite efficient in the implementation process, since the MNO usually set the data caps to be the nearest hundreds of MB (e.g., 100MB, 500MB, and 1GB).
For example, suppose that the maximal data demand is $\dmax=10$GB (which is large enough in most cases).
If the MNO would optimize the data cap with 1MB as the minimal unit, then there are a total of $|\mathcal{\dcap}|=10001$ possible data caps (i.e., $\mathcal{\dcap}=\{$0MB, 1MB, 2MB, 3MB,..., 10000MB$\}$) to be considered in this algorithm.
If the MNO would optimize the data cap with 100MB as the minimal unit, then there are only a total of $|\mathcal{\dcap}|=101$ possible data caps (i.e., $\mathcal{\dcap}=\{$0MB, 100MB, 200MB, 300MB,..., 10000MB$\}$).
Hence the algorithm is efficient in the implementation progress.

So far, we have completely solved the optimal contract.
Next we evaluate the proposed multi-dimensional contract.

%%%%%%%%%%%%%%%%%%%%%%%%%%%%%%%%%%%%%%%%%%%%%%%%%%%%%%%%%%%%%%%%%%%%%%%%%%%%%%%%%%%%%%%%%%%%%%%%%%%%%
%%%%%%%%%%%%%%%%%%%%%%%%%%%%%%%%%%%%%%%%%%%%%%%%%%%%%%%%%%%%%%%%%%%%%%%%%%%%%%%%%%%%%%%%%%%%%%%%%%%%%
%%%%%%%%%%%%%%%%%%%%%%%%%%%%%%%%%%%%%%%%%%%%%%%%%%%%%%%%%%%%%%%%%%%%%%%%%%%%%%%%%%%%%%%%%%%%%%%%%%%%%
\section{Numerical Results\label{Section: Numerical Results}} 
We evaluate the performance of the optimal contract based on some empirical data.
Specifically, we first illustrate the optimal contract structure in Section \ref{Subsection: Optimal Contract}, then investigate how the price discrimination and the time flexibility affect the MNO's profit and users' payoffs in Section \ref{Subsection: Impact of Price Discrimination and Time Flexibility}.

\subsection{Optimal Contract\label{Subsection: Optimal Contract}}
Next we introduce the estimated user types, data demand distribution, and the MNO's cost.
Then we illustrate the optimal contract structure.

\textbf{Estimated User Types:} 
According to the market survey results (based on over two thousand users of mainland China) in \cite{Zhiyuan2018TMC}, a large proportion of users' data valuations $\val$ is within the interval of $[15, 65]$ (in RMB/GB); most people would like to shrink approximately $70\%\sim100\%$ overage data consumption through alternative networks ($\cut$ value).
We follow \cite{ma2016economic} by using the $k$-means clustering method
%\footnote{The $k$-means clustering aims to partition $n$ observations into $k$ clusters, in which each observation belongs to the cluster with the nearest mean.} 
to partition the empirical data valuation $\val$ into four clusters with mean values $\Stheta=\{16.2, 36.1, 61.9, 96.3\}$, and partition the empirical network substitutability $\cut$ into four clusters with mean value $\Sbeta=\{0.51, 0.71, 0.84, 0.95\}$\footnote{Our previous work in \cite{Zhiyuan2018TMC} shows that the data valuation $\val$ and the network substitutability $\cut$ can be treated as independent with a Pearson correlation coefficient less than $0.05$.}.
Therefore, we consider a total of $KM=16$ user types,\footnote{In practice, the MNO can partition the empirical data into more clusters to increase the  accuracy  at the expense of additional complexity. Nevertheless, the MNO usually offers no more than ten data caps for implementation simplicity \cite{sen2013survey}.} and the corresponding distribution extracted from empirical data  is shown in Fig. \ref{fig: PDF}.
%Later on we simulate the optimal contract $\Phi(\Sbeta,\Stheta)$.

\textbf{Data Demand Distribution:} We set the minimum data unit as $1$MB.
Following the data analysis results in \cite{lambrecht2007does,nevo2016usage}, we suppose that users' monthly data demand follows a truncated log-normal distribution over the support of $[0,10^4]$ with a  mean $\dmean=10^3$, i.e., the average data demand $\dmean=1$GB and the maximal potential data demand $\dmax=10$GB \cite{Zhiyuan2018TMC}.
Fig. \ref{fig: Demand} shows the PMF $f(d)$ and the expected overage data consumption $A_{\mechanism}(\dcap)$ under different data mechanisms $\mechanism\in\{0,1,2\}$, which indicates that  $A_0(\dcap)\ge A_1(\dcap)\ge A_2(\dcap)$ for any $\dcap$.

\textbf{MNO's Cost:} As mentioned in (\ref{Equ: Cost}) of Section \ref{SubSection: MNO's Contract Formulation}, we take account of the capacity cost and the operational cost for the MNO.
To be consistent with our previous work \cite{Zhiyuan2018TMC}, we suppose that the capacity cost takes a linear form, i.e., $J(\dcap)=z\cdot\dcap$, where $z$ represents the MNO's marginal capacity cost.\footnote{Note that our method of solving the optimal contract is not limited to a specific form of the capacity cost $J(\dcap)$.}
In addition, $c$ represents the MNO's marginal operational cost.
Next we will vary the two parameters (i.e., $c$ and $z$) to illustrate their effects on the optimal contract and the corresponding MNO profit and user payoff.

\begin{figure*}
	\centering
	\setlength{\abovecaptionskip}{0pt}
	\setlength{\belowcaptionskip}{0pt}
	\subfigure[Data caps under $\mechanism=0$.]{\label{fig: General_Q_0}{\includegraphics[width=0.26\linewidth]{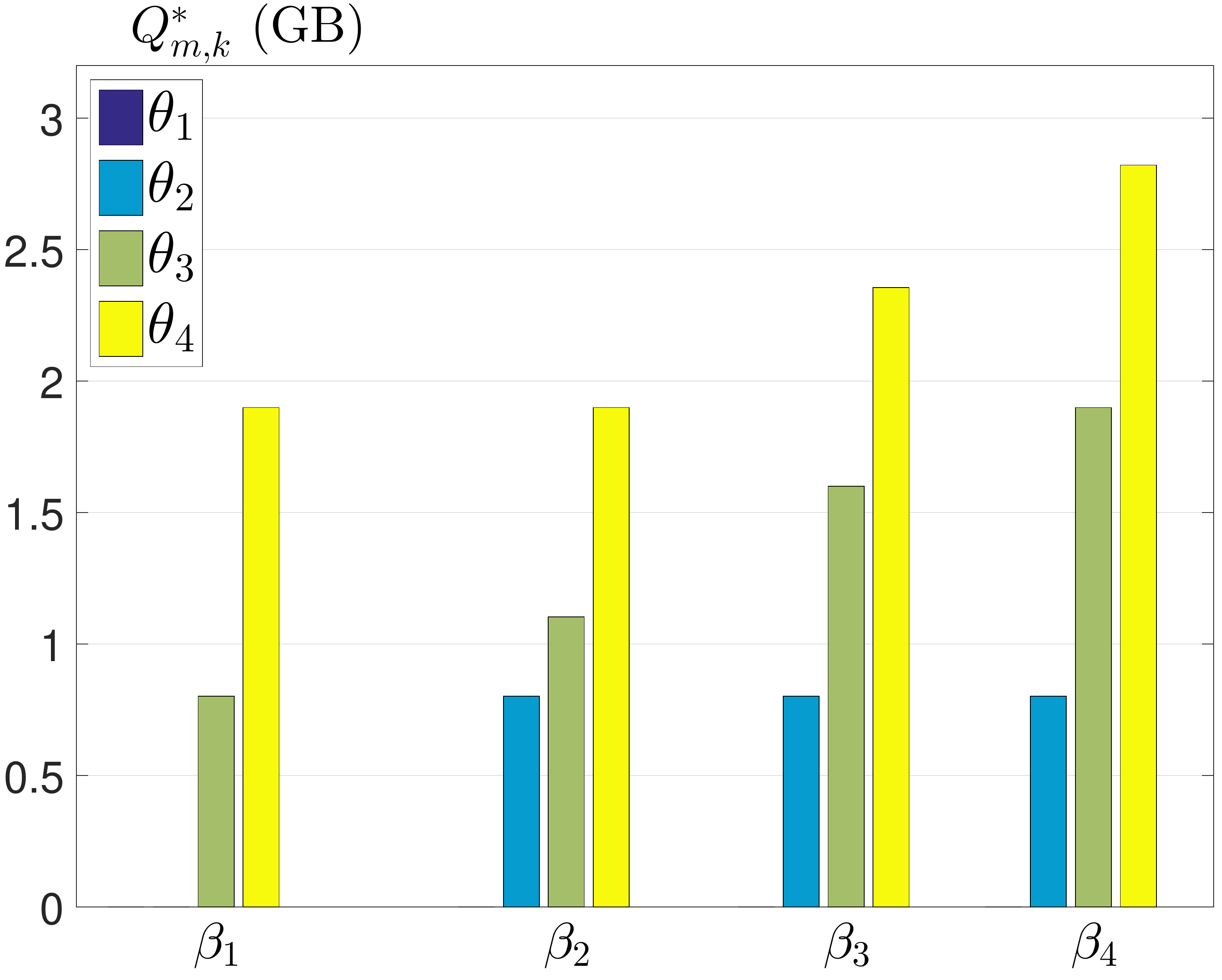}}} \qquad\quad
	\subfigure[Data caps under $\mechanism=1$.]{\label{fig: General_Q_1}{\includegraphics[width=0.26\linewidth]{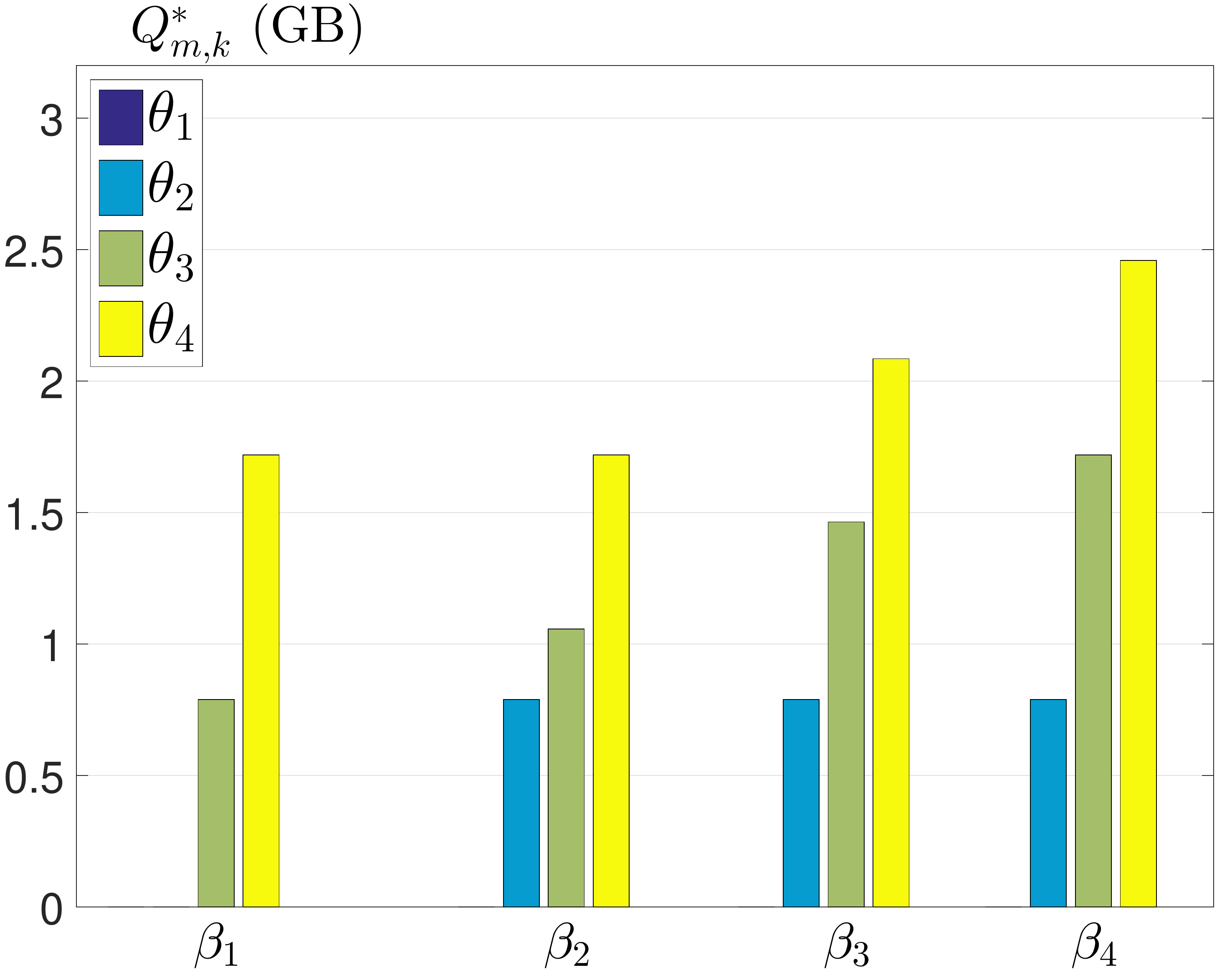}}} \qquad\quad
	\subfigure[Data caps under $\mechanism=2$.]{\label{fig: General_Q_2}{\includegraphics[width=0.26\linewidth]{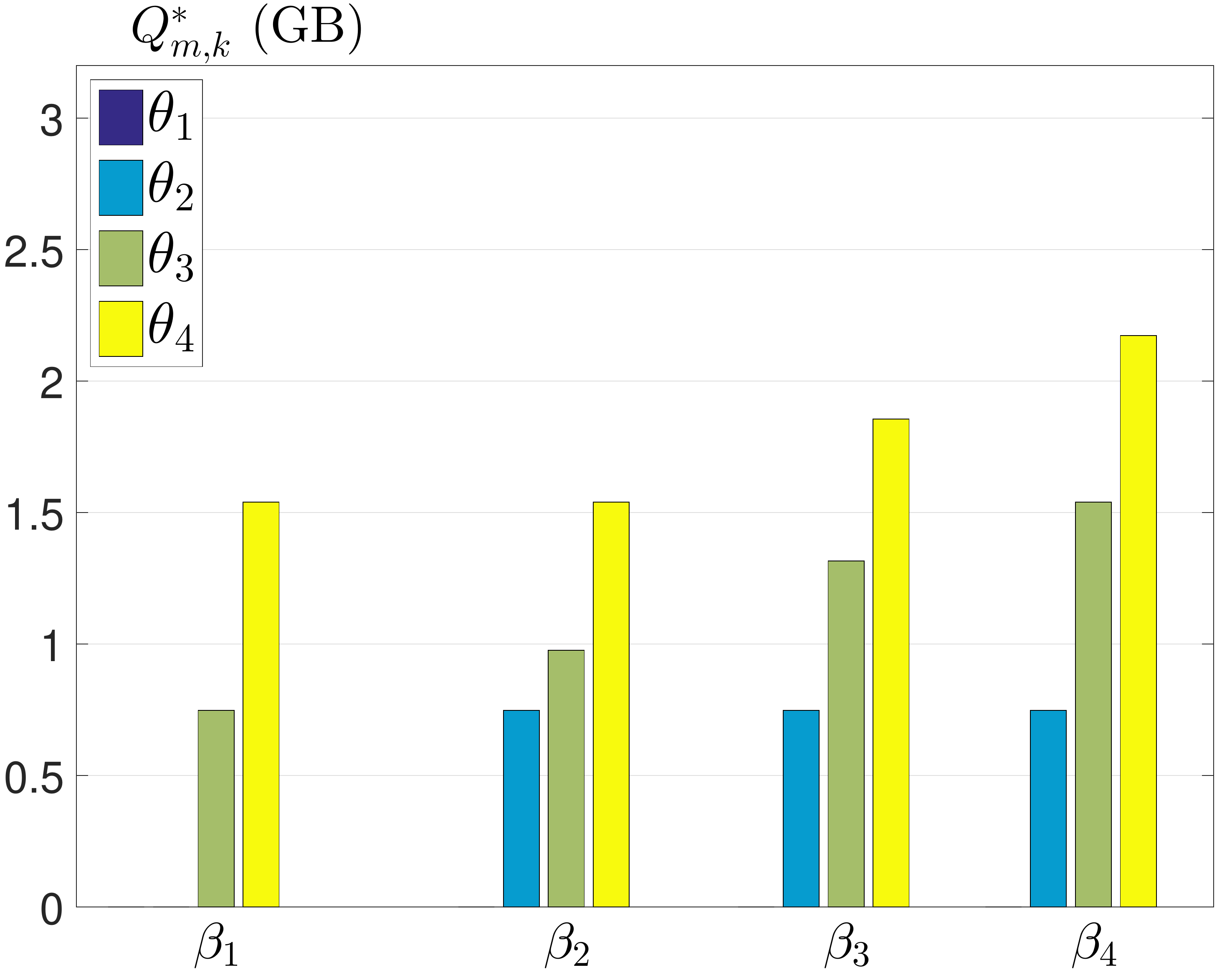}}} \\
	\subfigure[Subscription fees under $\mechanism=0$.]{\label{fig: General_PI_0}{\includegraphics[width=0.26\linewidth]{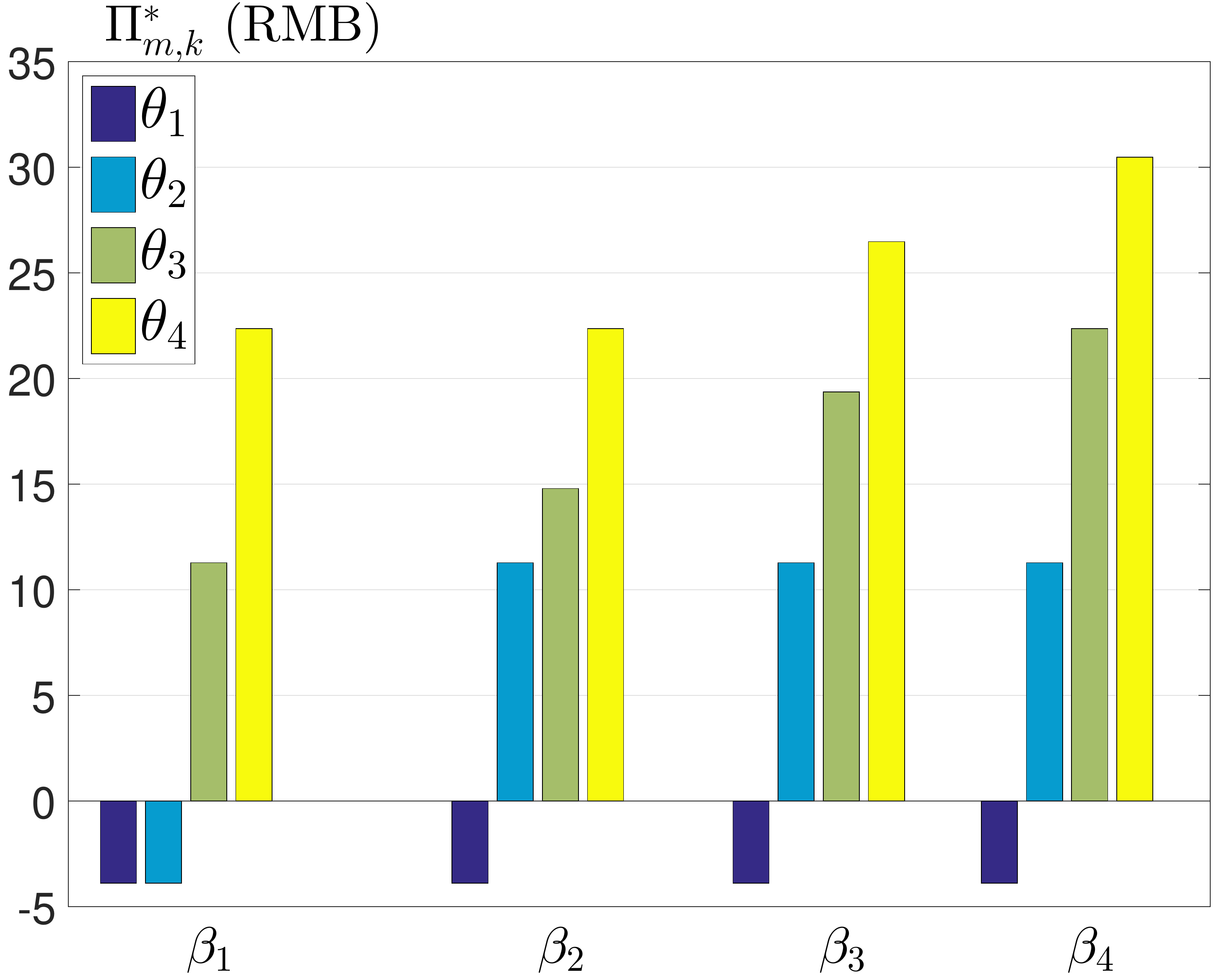}}} \qquad\quad
	\subfigure[Subscription fees under $\mechanism=1$.]{\label{fig: General_PI_1}{\includegraphics[width=0.26\linewidth]{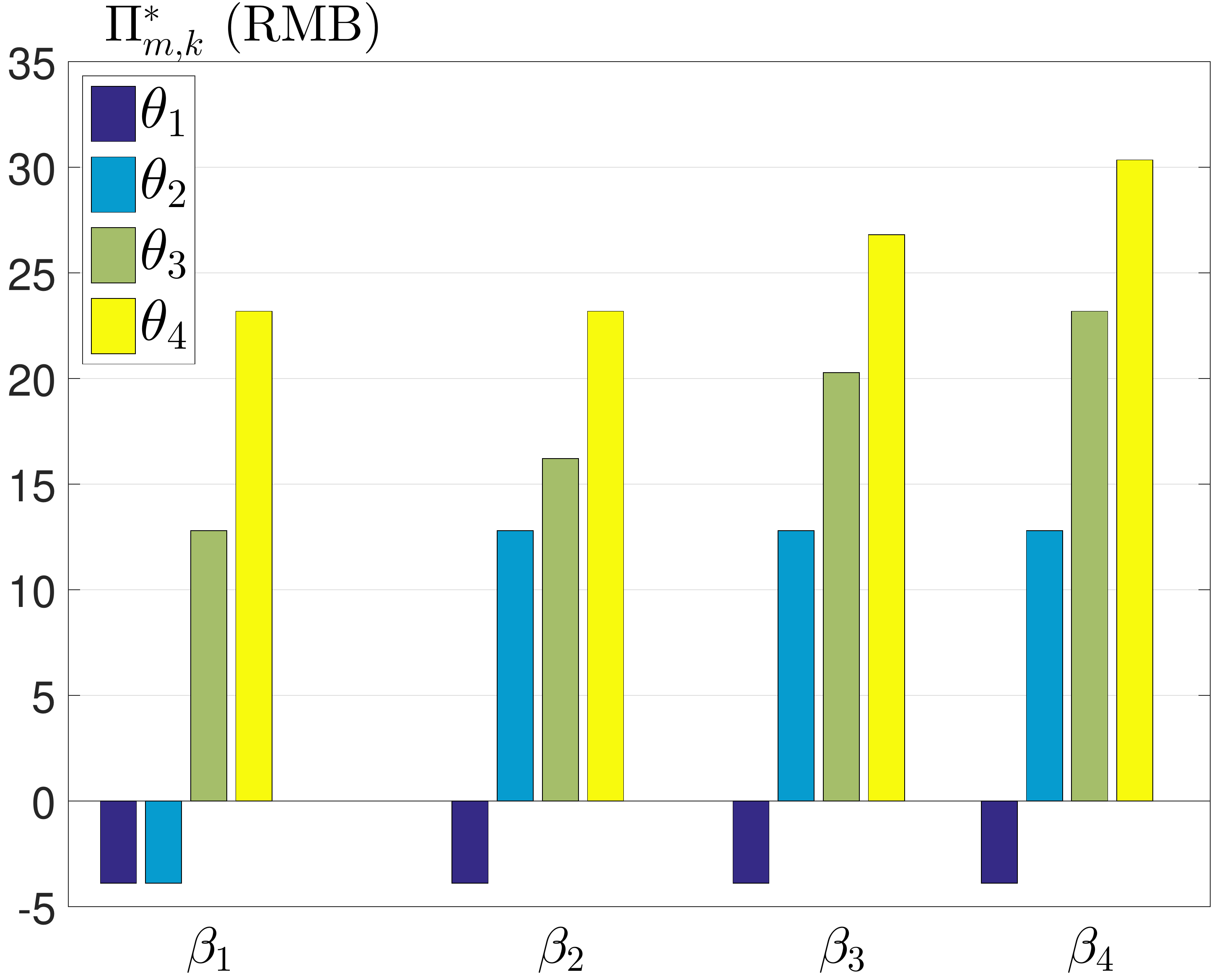}}} \qquad\quad
	\subfigure[Subscription fees under $\mechanism=2$.]{\label{fig: General_PI_2}{\includegraphics[width=0.26\linewidth]{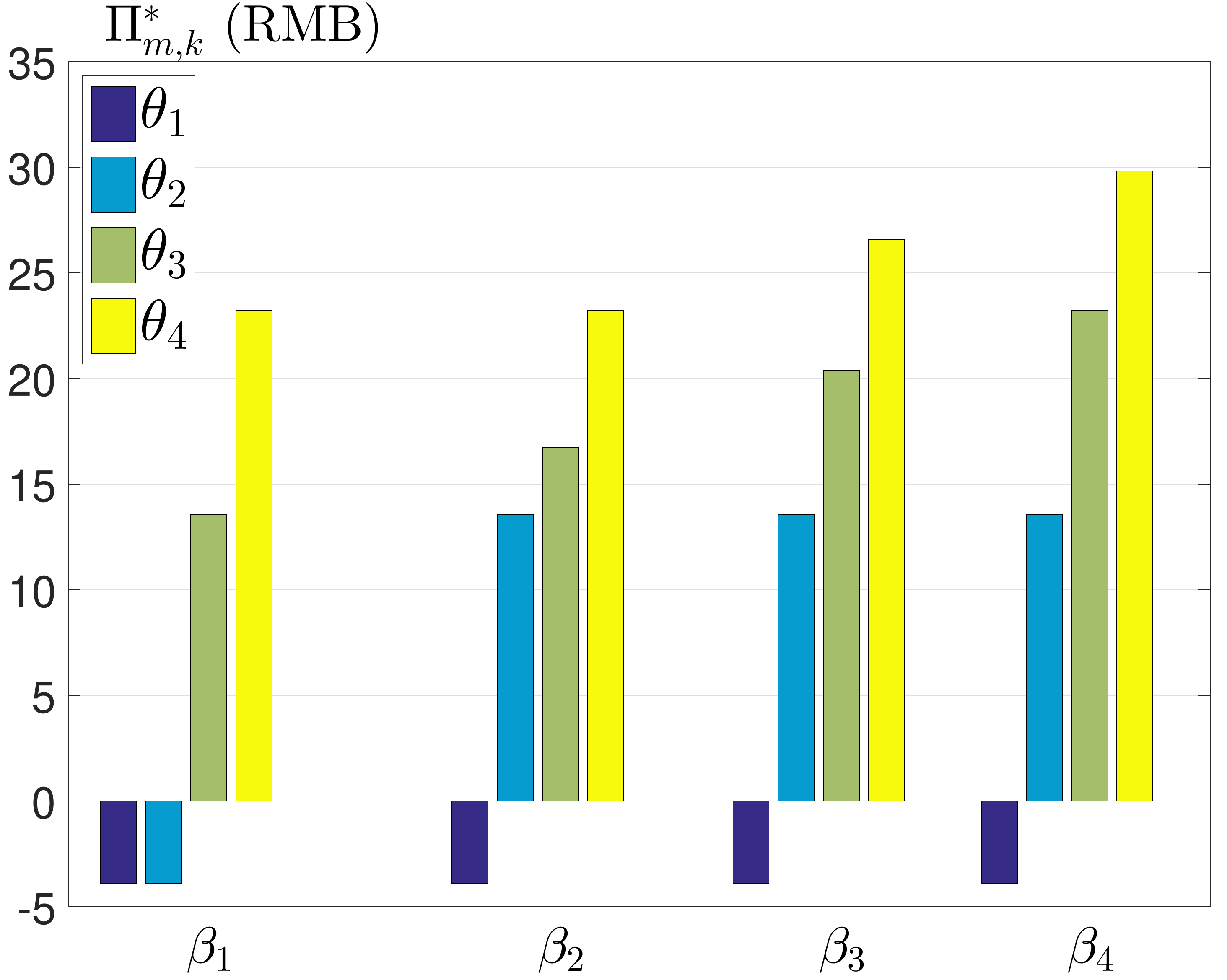}}} 	
	\caption{Optimal contract under the three data mechanisms  $\mechanism\in\{0,1,2\}$. }
	\label{fig: optimal contract}
\end{figure*}

Furthermore, we use the per-unit fee in the telecommunication market of China, i.e., $\adfee=30$ RMB/GB.
Based on the above setting, we will evaluate the optimal contract in the following three steps.

\subsubsection{Contract Structure}
We take the data mechanism $\mechanism=1$ as example to visualize the contract structure based on the users' types. 

Fig. \ref{fig: CapStructure} shows some properties of  the optimal contract item for each user type, given the MNO's cost $c=5$RMB/GB, $z=0.9$RMB/GB.
Specifically, the markers represent all the user types $\type_i=(\cut_m,\val_k)$, each of which corresponds to a network substitutability $\cut_m$ and a data valuation $\val_k$ in the horizontal and vertical axis, respectively.
%The dash line represents the per-unit fee $\adfee=40$ RMB/GB. 
Moreover, the arrows point to the non-decreasing direction of users' \textit{willingness-to-pay} as defined in (\ref{Equ: order}).
The markers of the same shape and color represent that the corresponding users types have the same contract item (i.e., pooling contract).
Therefore, the optimal contract contains seven different contract items for a total of $KM=16$ types of users.
%Particularly, the red solid circles in Fig. \ref{fig: CapStructure} corresponds to a zero-cap contract item, which represents that the optimal contract tends to provide the user types of weak \textit{willingness-to-pay} (i.e., $\type_1,\type_2,...,\type_{5}$) with the pure usage-based plan.

\subsubsection{Impact of Data Mechanisms}
Next we compare the  optimal contract under different data mechanisms $\mechanism\in\{0,1,2\}$.

Fig. \ref{fig: optimal contract} plots the optimal data caps (i.e., Figs. \ref{fig: General_Q_0}, \ref{fig: General_Q_1}, \ref{fig: General_Q_2}) under the three data mechanisms and the corresponding subscription fees (i.e., Figs. \ref{fig: General_PI_0}, \ref{fig: General_PI_1}, \ref{fig: General_PI_2}). 
We have the following observations:
\begin{itemize}
	\item For all three data mechanisms, the optimal contract offers some low valuation users (e.g.,  $\type_1$) a zero data cap (e.g., pure usage-based plan in Fig. \ref{fig: General_Q_0}), together with a negative price (e.g., the five negative bars in Fig. \ref{fig: General_PI_0}).
	The pure usage-based plan reduces the MNO's capacity cost due to the zero data cap.
	Meanwhile, the negative price serves as a price discount, which ensures the subscription of these users (still satisfying the IR condition)\footnote{In practice, the MNO may allow users to pay 100RMB and enjoy the usage-based data service that is equivalent
		to 120RMB, which is actually similar to the $-20$RMB subscription fee. On the other hand, the MNOs can also directly
		subsidize 20RMB for the usage-based subscribers. The current wireless data market is based on real-name registration, hence the negative subscription fee (or discount) is not a concern.}.

	\item For each data mechanism $\mechanism$, the optimal contract tends to offer the users who have small $\cut$ values hence poor alternative network choices (e.g., the type-$(\cut_1,\val_3)$ users) a small data cap (e.g., 0.9GB in Fig. \ref{fig: General_Q_0}) together a low subscription fee (e.g., 12RMB in Fig. \ref{fig: General_PI_0}). 
	As a result, these users will end up paying a lot of  overage fee.
	However, the optimal contract offers the users who have high $\cut$ values hence good alternative network choices (e.g., the type-$(\cut_4,\val_4)$ users) a large data cap (e.g., 2.8GB in Fig. \ref{fig: General_Q_0}) together with a high subscription fee (e.g., 30RMB in Fig. \ref{fig: General_PI_0}).
	
	\item Under the optimal contract, the better time flexibility (i.e., a larger value of $\mechanism$) enables the MNO to offer a smaller data cap for the same type of users. 
	For example, the optimal data cap for type-$(\cut_4,\val_4)$ users is 2.8GB, 2.5GB, and 2.2GB in Fig. \ref{fig: General_Q_0}, \ref{fig: General_Q_1}, and \ref{fig: General_Q_2}, respectively.
	The MNO reduces its capacity cost by offering a better time flexibility (i.e., $\mechanism=0\rightarrow1\rightarrow2$).
\end{itemize}

\begin{figure*}
	\centering
	\begin{minipage}{0.49\textwidth}
		\centering
		\setlength{\abovecaptionskip}{0pt}
		\setlength{\belowcaptionskip}{0pt}
		\subfigure[Optimal data caps.]{\label{fig: c_Cap}{\includegraphics[height=0.4\linewidth]{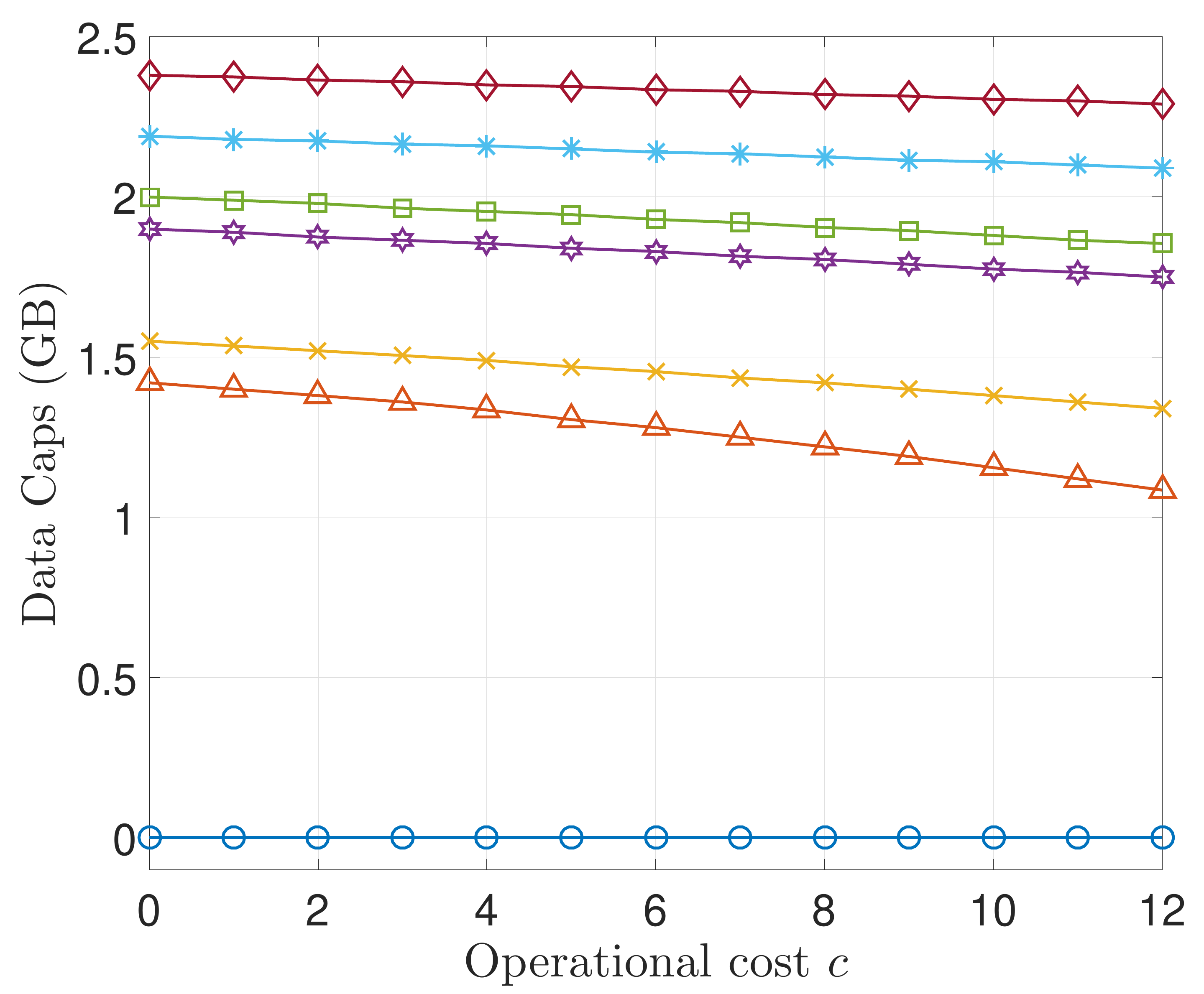}}} \
		\subfigure[Optimal subscripton fees.]{\label{fig: c_Subscription}{\includegraphics[height=0.4\linewidth]{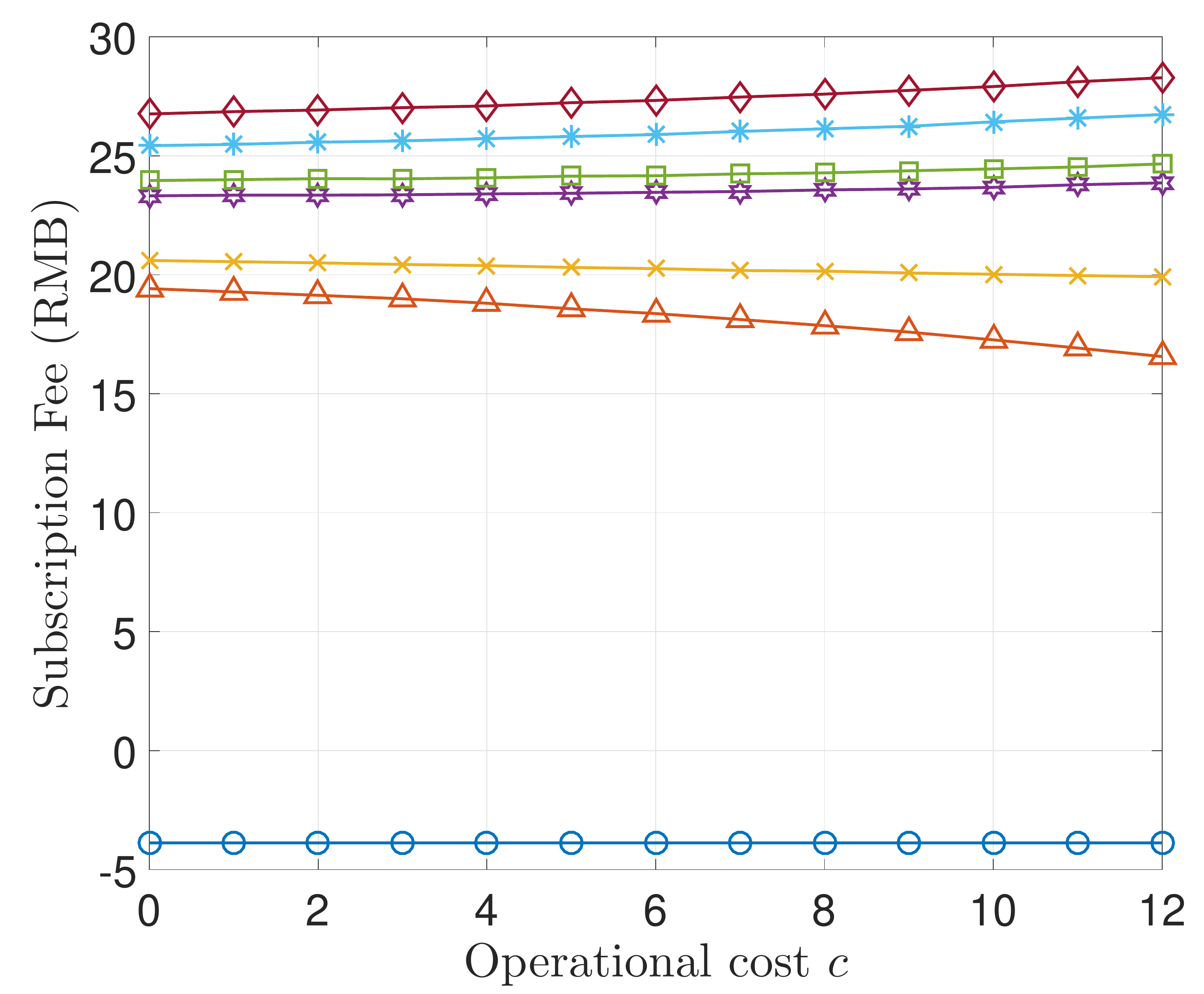}}} 
		\caption{Impact of MNO's operational cost $c$.}
		\label{fig: impact of c}
	\end{minipage}\quad
	\begin{minipage}{0.49\textwidth}
		\centering
		\setlength{\abovecaptionskip}{0pt}
		\setlength{\belowcaptionskip}{0pt}
		\subfigure[MNO's profit.]{\label{fig: c_Profit}{\includegraphics[width=0.48\linewidth]{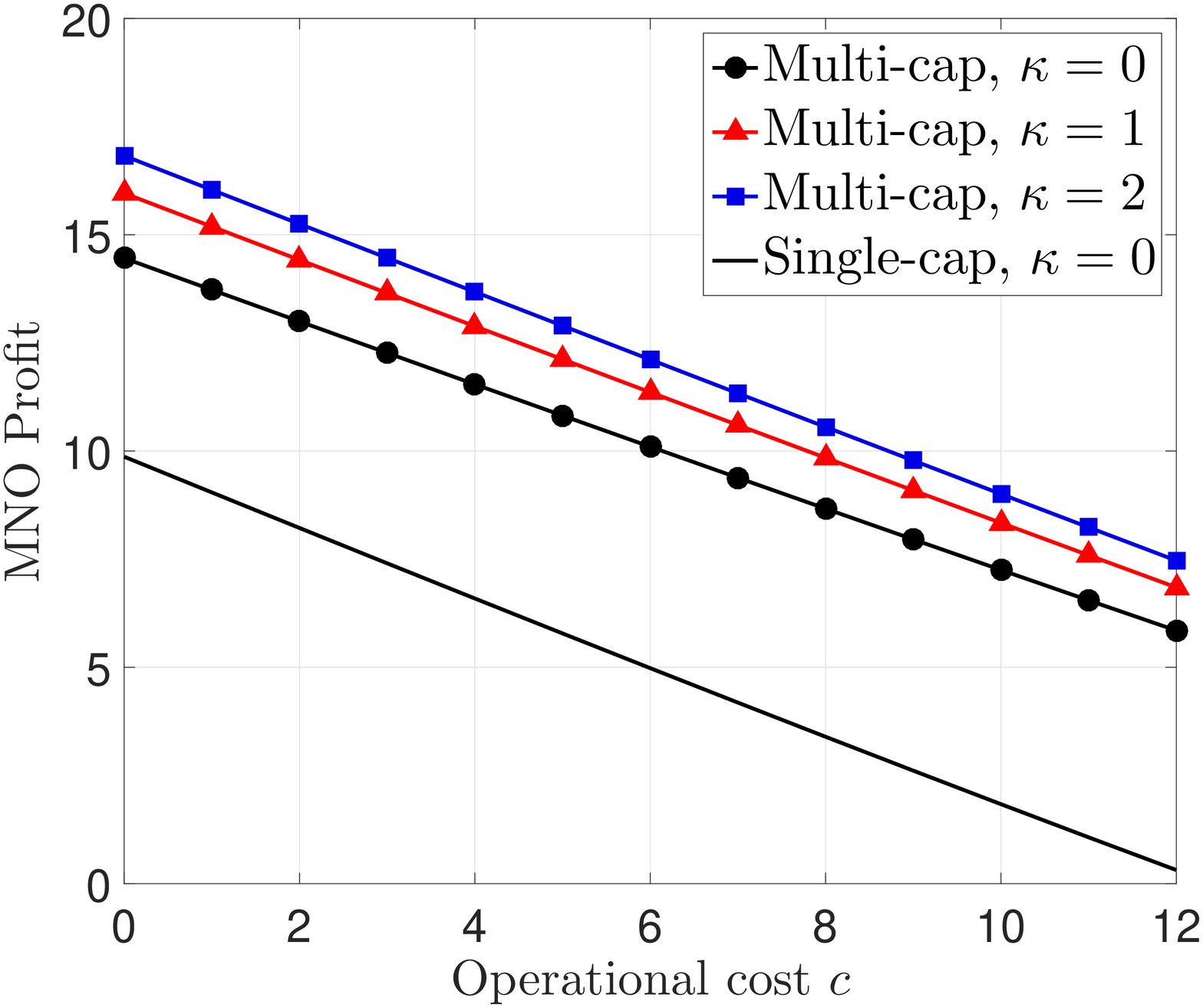}}} \ 
		\subfigure[All users' expected payoff.]{\label{fig: c_Payoff}{\includegraphics[width=0.48\linewidth]{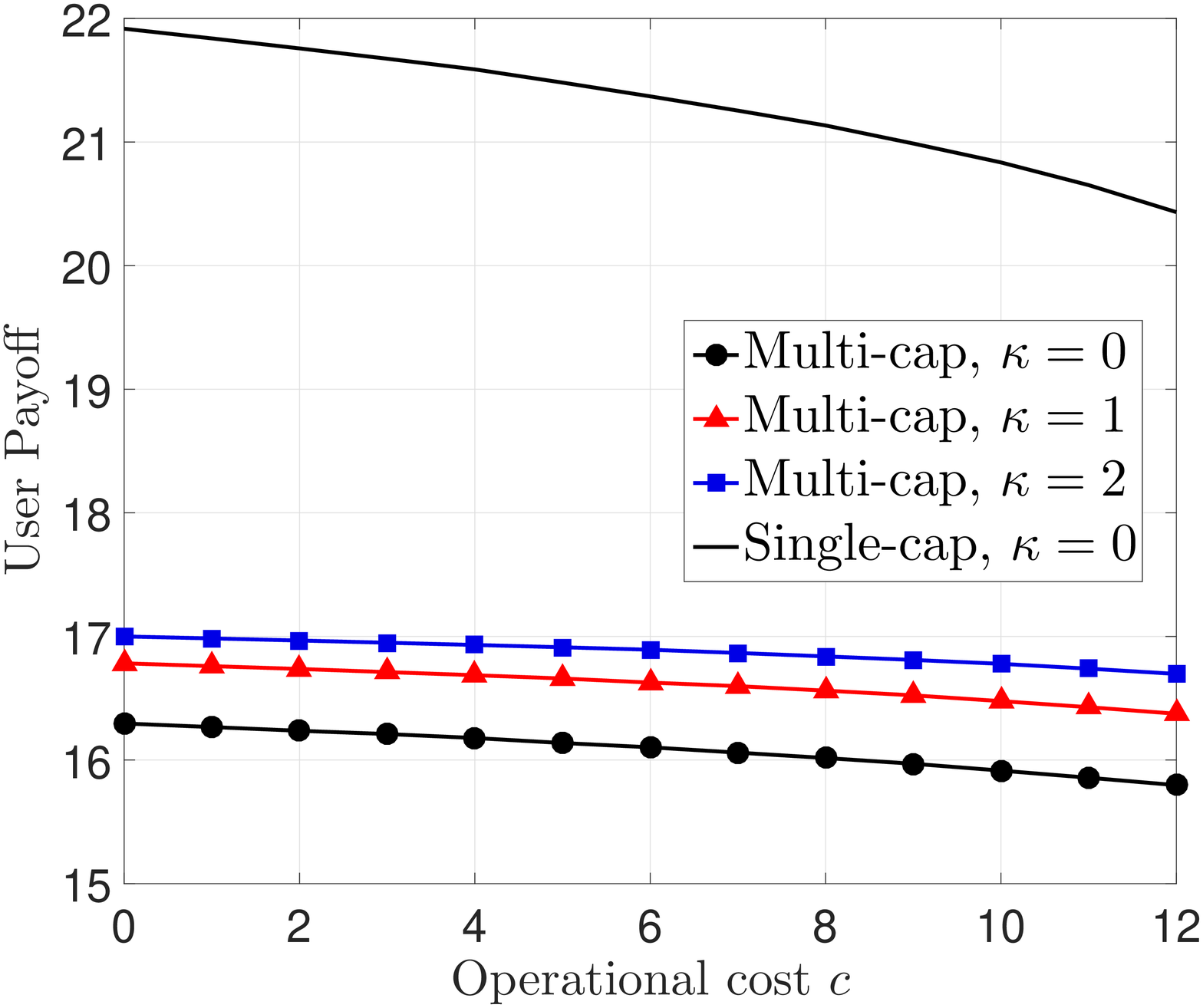}}} 
		\caption{Impact of the MNO's operational cost $c$.}
		\label{fig: c profit payoff}
	\end{minipage} 
\end{figure*}

\subsubsection{Impact of MNO's Costs}\label{Subsubsection: Impact of MNO's Costs}
We take the data mechanism $\mechanism=1$ as an example to investigate how the MNO's costs (i.e., both $c$ and $z$) affect the optimal contract items.

Fig. \ref{fig: impact of c} shows the impact of the MNO's operational cost $c$.
Specifically, there are a total of seven different contract items in the optimal contract.
The seven curves in Fig. \ref{fig: c_Cap} represent the different data caps.
We note that overall the optimal data caps (except the zero cap) decreases as the MNO's operational cost increases.
Fig. \ref{fig: c_Subscription} plots the corresponding subscription fees versus the MNO's operational cost.
We find that
\begin{itemize}
	\item The subscription fee of the zero-cap contract item (i.e., the bottom blue circle curve in Fig. \ref{fig: c_Subscription}) does not change in the operation cost $c$.
	This results from the individual rationality condition as in (\ref{Equ: Optimal Pricing Policy 1}).
	\item The subscription fees of small-cap contract items (e.g., the orange triangle and yellow cross curves in Fig. \ref{fig: c_Subscription}) decrease as the MNO's costs increase.
		While the subscription fees of the large-cap contract items (the remaining curves)  increase in the MNO's costs.
		Therefore, the large-cap contract items become  less economical to the users (in terms of the average price $\pcap/\dcap$) as the MNO's costs increase.
		This means that a profit-maximizing MNO tends to compensate its operational cost by charging those users who are willing to pay for the  large-cap contract items.
\end{itemize}

We also investigate the impact of the MNO's capacity cost $z$.
The insights are similar to those from Fig. \ref{fig: impact of c}.
We refer interested readers to Appendix \ref{Appendix: simulation} for more details.

\subsection{Impact of Price Discrimination and Time Flexibility\label{Subsection: Impact of Price Discrimination and Time Flexibility}}
We evaluate the effect of the price discrimination and the time flexibility on the MNO's profit and users' payoffs.

We consider four scenarios as shown in Table \ref{table: four cases}.
Scenario (i) represents the benchmark  single-cap scheme under the traditional data mechanism $\mechanism=0$ (studied in our previous work \cite{Zhiyuan2018TMC}).
Scenarios (ii), (iii), and (iv) represent the multi-cap scheme under different data mechanisms (studied in this paper).

\begin{table}
	\setlength{\abovecaptionskip}{1pt}
	\setlength{\belowcaptionskip}{0pt}
	\renewcommand{\arraystretch}{1}
	\caption{Four scenarios. }
	\label{table: four cases}
	\centering
	\begin{tabular}{ccccccccc}
		\toprule
		Scenario			& Multi-cap or Single-cap		& Data Mechanism \\
		\midrule
		(i)				& Single						& $\mechanism=0$ \\
		(ii)			& Multiple						& $\mechanism=0$ \\
		(iii)			& Multiple						& $\mechanism=1$ \\
		(iv)			& Multiple						& $\mechanism=2$ \\
		\bottomrule
	\end{tabular}
\end{table}

%We first evaluate the 
Fig. \ref{fig: c profit payoff} plots the MNO's profit and user's payoffs in the four scenarios under different operational cost $c$.
%, the horizontal axises in the two sub-figures represent the MNO's marginal operational cost.
\begin{itemize}
	\item Fig. \ref{fig: c_Profit} plots MNO's profits in the four scenarios.
		Overall, the MNO's profits decrease in its operational cost $c$.
		By comparing the single-cap traditional pricing benchmark and the multi-cap traditional pricing scheme, we note that the price discrimination under our optimal contract can significantly increase the MNO's profit (180\% on average).
		By comparing the three multi-cap curves, we find that the MNO obtains a higher profit under a more time-flexible data mechanism.
		Specifically, compared with Scenario (ii) (i.e., the black circle curve), MNO's profits increases by 15\% on average in Scenario (iii) (i.e., the red triangle curve) and 25\% on average in Scenario (iv) (i.e., blue square curve).
		This implies that under the multi-cap scheme, offering a better time flexibility can further improve the MNO's profit.
	\item Fig. \ref{fig: c_Payoff} plots the users' total expected payoff in four scenarios.
		First, we observe that  users' payoff decreases in the MNO's operational cost.
		By comparing the single-cap traditional pricing benchmark and the multi-cap traditional pricing scheme, we notice that the price discrimination under our optimal contract reduces users' expected payoff (23\% on average), which means that the MNO captures more consumer surplus through the price discrimination.
		Comparing the three multi-cap schemes, we find that the time-flexible data mechanisms can improve the users' payoff. 
		Specifically, compared with Scenario (ii) (i.e., the circle curve), users' payoff increases by 5.1\% on average in Scenario (iii) (i.e., the triangle curve) and 8.2\% on average in Scenario (iv) (i.e., square curve).
%		Hence a better time flexibility leads to a higher expected payoff for the user market.
\end{itemize}

We also evaluate the performance under different capacity cost $z$, which leads to similar insights.
We refer interested readers to Appendix \ref{Appendix: simulation} for more details.
Furthermore, we also evaluate the impact of the number of user types on the optimal contract performance.
Due to page limit, we refer interested readers to Appendix \ref{Appendix: type} for more details.

\section{Conclusion and Future Works\label{Section: Conclusion}}  
In this paper we studied how the MNO optimizes its multi-cap  data plans under the time-flexible data mechanisms.
Specifically, we consider an asymmetric information scenario, where each use is associated with two-dimensional private information, i.e., his data valuation and network substitutability.
We formulate the MNO's optimal multi-cap design as a multi-dimensional contract problem and derive the optimal contract under different data mechanisms.
Our analysis revealed that the slope of a user's indifference curve on the contract plane corresponds to his willingness-to-pay, and the feasible contract (satisfying IC and IR conditions) would offer a larger data cap to the user with the stronger willingness-to-pay.
Moreover, we proposed an efficient algorithm to solve the contract problem optimally.
%The numerical results show that under the multi-cap scheme, the time flexibility can further increase the MNO's profit and the users' payoffs.

In the future, we have two directions to extend the results of this paper.
\begin{itemize}
	\item First, we will collect more empirical data related to users' data demand distributions and try to relax the current homogeneous assumption of the data demand distribution.
	This will lead  to a new contract problem with three-dimensional private information, which will be much more challenging to solve. 
	
	\item Second, we will consider the competitive market.
	So far we have shown that the time flexibility increases the MNO's profit and users' payoffs under the multi-cap scheme, it is still necessary to analyze the role of price discrimination and time flexibility on MNOs' market competition.
	This will build upon our previous analysis of  the competitive market under the single-cap scheme. 
\end{itemize}

\bibliography{ref}
\bibliographystyle{IEEEtran}

\vspace{-30pt}
\begin{IEEEbiography}[{	\includegraphics[width=1in,height=1.25in,clip,keepaspectratio]{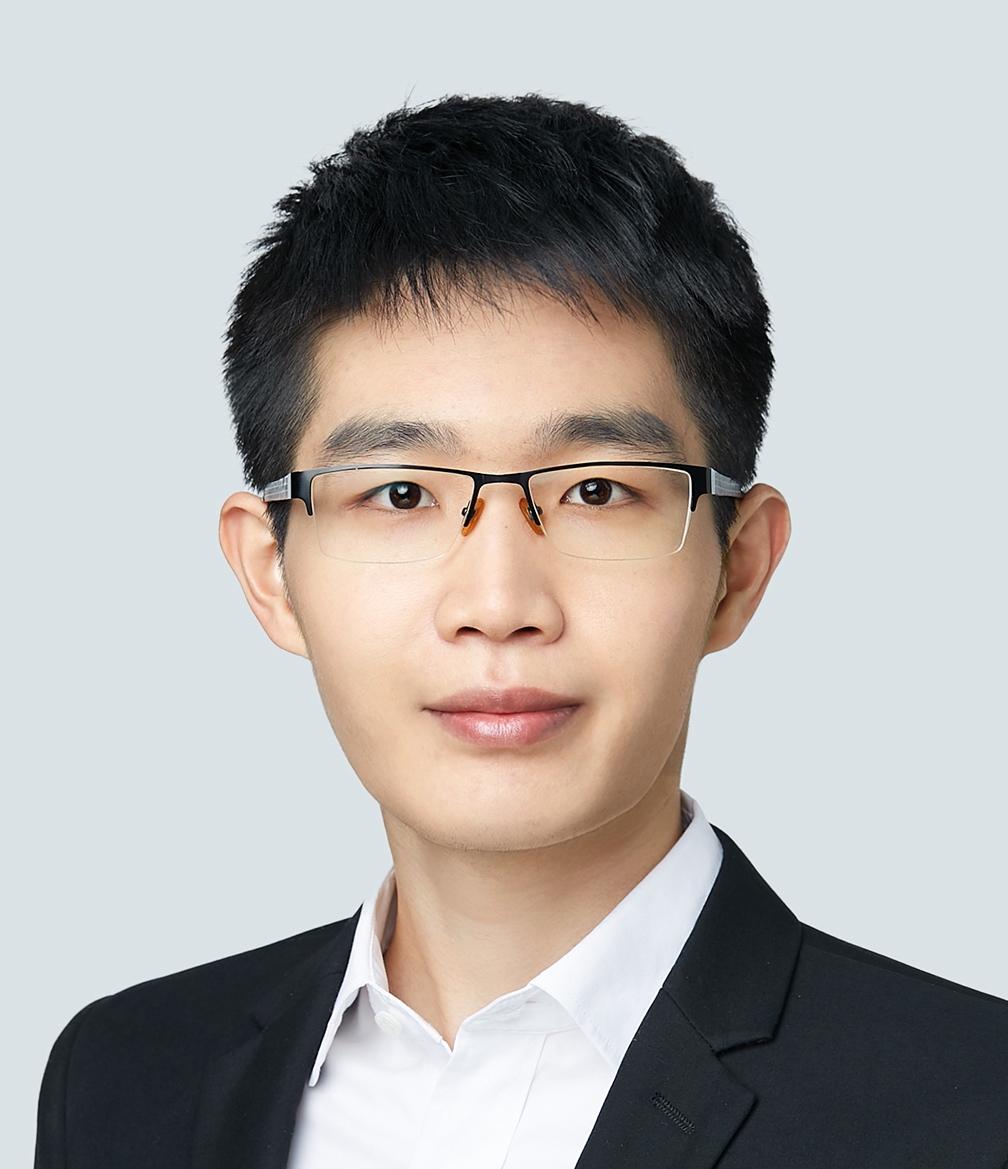}	}]{Zhiyuan Wang}
	received the B.S. degree from Southeast University, Nanjing, China, in 2016.
	He is currently working toward the Ph.D. degree with the Department of Information Engineering, The Chinese University of Hong Kong, Shatin, Hong Kong.  
	His research interests include the field of network economics and game theory, with current emphasis on smart data pricing and mobile edge computing. 
	He is the recipient of the Hong Kong PhD Fellowship.
\end{IEEEbiography}

\vspace{-30pt}
\begin{IEEEbiography}[{	\includegraphics[width=1in,height=1.25in,clip,keepaspectratio]{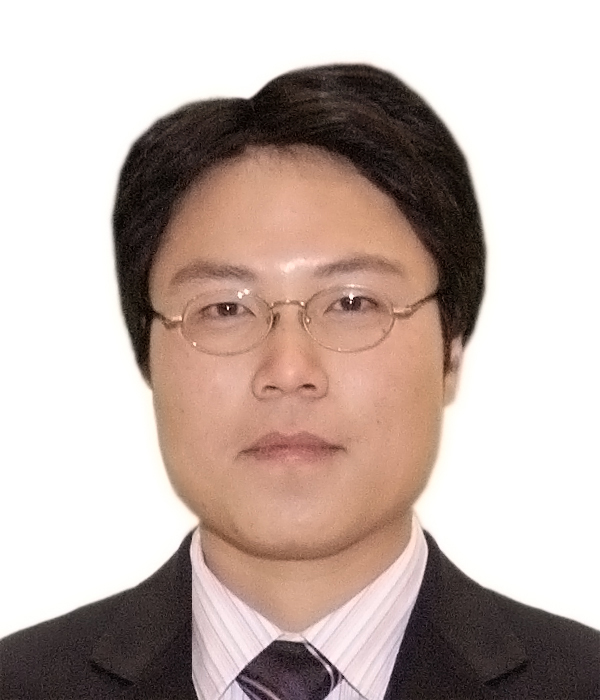}	}]	{Lin Gao}
	(S'08-M'10-SM'16) is an Associate Professor with the School of Electronic and Information Engineering, Harbin Institute of Technology, 	Shenzhen, China. He received the Ph.D. degree in Electronic Engineering from Shanghai Jiao Tong University in 2010. His main research
	interests are in the area of network economics 	and games, with applications in wireless communications 	and networking. He received the	IEEE ComSoc Asia-Pacific Outstanding Young 	Researcher Award in 2016.
\end{IEEEbiography}

\vspace{-25pt}
\begin{IEEEbiography}
	[{	\includegraphics[width=1in,height=1.25in,clip,keepaspectratio]{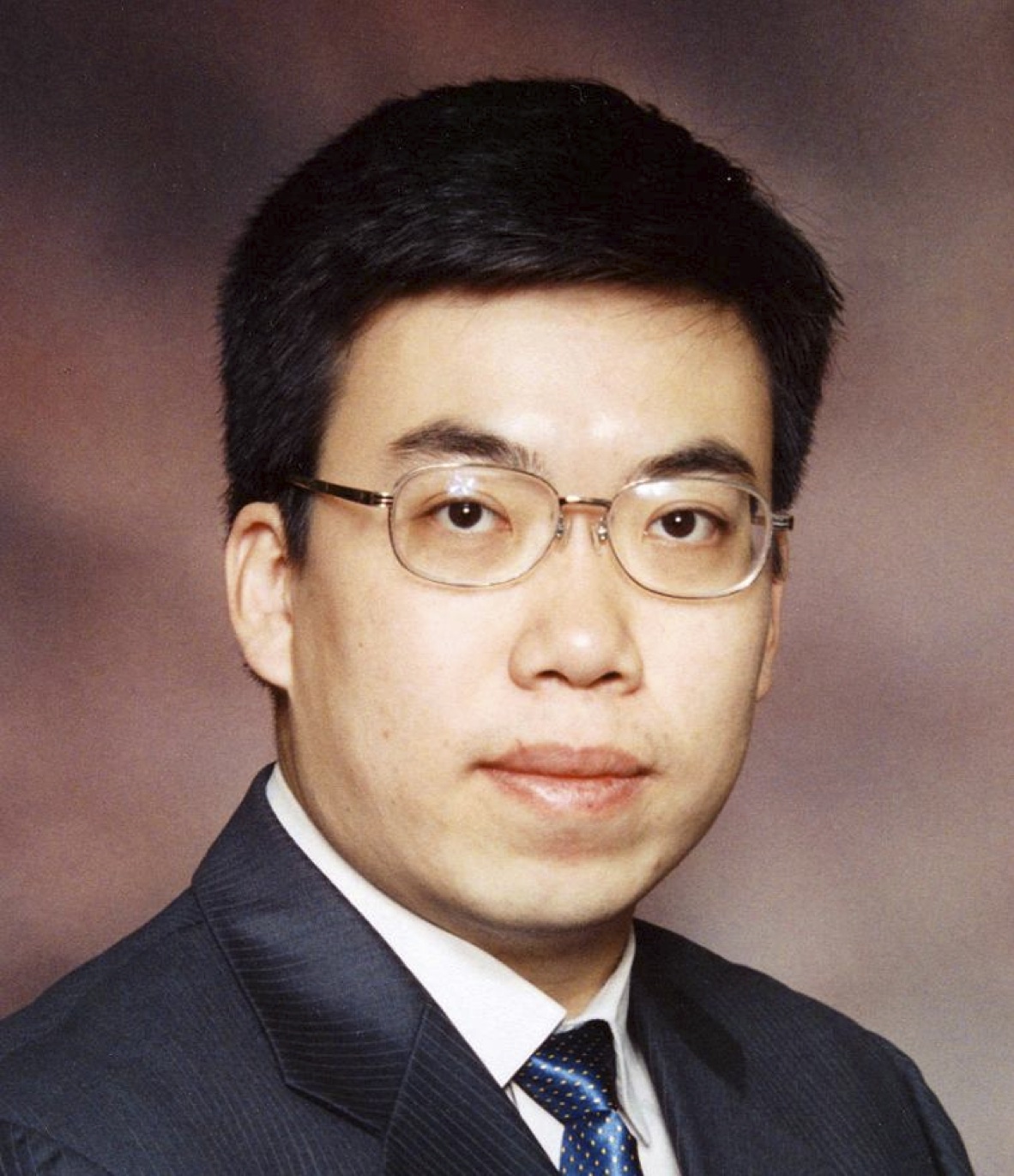}	}]	{Jianwei Huang}
	(F'16) is a Presidential Chair Professor and Associate Dean of the School of Science and Engineering, The Chinese University of Hong Kong, Shenzhen. 
	He is also a Professor in the 	Department of Information Engineering at The 	Chinese University of Hong Kong. 
	He is the
	co-author of 9 Best Paper Awards, including 	IEEE Marconi Prize Paper Award in Wireless
	Communications 2011. He has co-authored six 	books, including the textbook on ``Wireless Network
	Pricing''. He has served as the Chair of 	IEEE Technical Committee on Cognitive Networks and Technical Committee on Multimedia Communications. 
	He has been  an IEEE ComSoc 	Distinguished Lecturer and a Thomson Reuters 	Highly Cited Researcher.
\end{IEEEbiography}

\newpage
\appendices

%%%%%%%%%%%%%%%%%%%%%%%%%%%%%%%%%%%%%%%%%%%%%%%%%%%%%%%%%%%%%%%%%%%%%%%%%%%
%%%%%%%%%%%%%%%%%%%%%%%%%%%%%%%%%%%%%%%%%%%%%%%%%%%%%%%%%%%%%%%%%%%%%%%%%%%
%%%%%%%%%%%%%%%%%%%%%%%%%%%%%%%%%%%%%%%%%%%%%%%%%%%%%%%%%%%%%%%%%%%%%%%%%%%
\section{\label{Appendix: User type}}
\begin{proof}[\textbf{Proof of Lemma \ref{Lemma: independent of Q}}]
	We prove this lemma by showing that if $\slope(\dcap,\type_i)<\slope(\dcap,\type_j)$, then  $\slope(\dcap',\type_i)<\slope(\dcap',\type_j)$ for any $i,j\in\{1,2,...,KM\}$ and any $\dcap'\ne\dcap$.
	
	In this paper, we take into account two-dimensional user type.
	For a type-$\type_i$ user, we denote $\cut^i$ and $\val^i$ as his network substitutability and data valuation, respectively.
	That is, $\type_i=\{\cut^i,\val^i\}$.
	Similarly, we have $\type_j=\{\cut^j,\val^j\}$ for the type-$\type_j$ user.
	
	Recall that the user's willingness-to-pay $\slope(\dcap,\cut,\val)$ is 
	\begin{equation}
		\slope(\dcap,\cut,\val) = -\left[\val\cut+\adfee(1-\cut)\right]\frac{\partial A(\dcap)}{\partial \dcap},
	\end{equation}
	which has a separable structure between user's private information $(\cut,\val)$ and the data cap $\dcap$.
	Therefore, $\slope(\dcap,\type_i)<\slope(\dcap,\type_j)$ implies that
	\begin{equation}
	\begin{aligned}
		&-\left[\val^i\cut^i+\adfee(1-\cut^i)\right]\frac{\partial A(\dcap)}{\partial \dcap} < \\
		&\qquad\qquad\qquad\qquad -\left[\val^j\cut^j+\adfee(1-\cut^j)\right]\frac{\partial A(\dcap)}{\partial \dcap},
	\end{aligned}
	\end{equation}
	which means that
	\begin{equation}\label{Proof Equ: val cut compare}
	\textstyle
	\val^i\cut^i+\adfee(1-\cut^i) < \val^j\cut^j+\adfee(1-\cut^j).
	\end{equation}	 
	Multiply both sides of (\ref{Proof Equ: val cut compare}) by $-\frac{\partial A(\dcap')}{\partial \dcap}$, we obtain
	\begin{equation}
	\textstyle
	-\left[\val^i\cut^i+\adfee(1-\cut^i)\right]\frac{\partial A(\dcap')}{\partial \dcap} < -\left[\val^j\cut^j+\adfee(1-\cut^j)\right]\frac{\partial A(\dcap')}{\partial \dcap},
	\end{equation}	
	which implies that $\slope(\type_i,\dcap')<\slope(\type_j,\dcap')$.
\end{proof}

\begin{proof}[\textbf{Proof of Proposition \ref{Proposition: type^1, type^KM}}]
	We prove this proposition based on Lemma \ref{Lemma: independent of Q}.
	
	Recall that we consider a set $\Stheta=\{\val_k:1\le k\le{K}\}$ of $K$ data valuation types and a set  $\Sbeta=\{\cut_m:1\le m\le{M}\}$ of $M$ network substitutability types.
	According to the proof of Lemma \ref{Lemma: independent of Q}, we find that the new user order (based on their willingness-to-pay) only depends on the order of 
	\begin{equation}\label{Proof Equ: slope transform}
	\val\cut+\adfee(1-\cut)=\adfee+(\val-\adfee)\cut.
	\end{equation}
	The expression in (\ref{Proof Equ: slope transform}) monotonically  increases in $\val$.
	However, it increases in $\cut$ if $\val>\adfee$ and decreases in $\cut$ if $\val<\adfee$.
	Therefore, we have three cases depending on the relation between $\adfee$, $\val_1$, and $\val_K$.
	\begin{itemize}
		\item Fig. \ref{fig: ThetaBeta_low}: The case of $\val_{K}<\adfee$ corresponds to $\type_1=\{\cut_{M},\val_1\}$ and  $\type_{KM}=\{\cut_1,\val_{K}\}$.
		\item Fig. \ref{fig: ThetaBeta_cross}: The case of $\val_1<\adfee<\val_{K}$ corresponds to $\type_1=\{\cut_{M},\val_1\}$ and $\type_{KM}=\{\cut_{M},\val_{K}\}$.
		\item Fig. \ref{fig: ThetaBeta_high}: The case of $\adfee<\val_1$ corresponds to $\type_1=\{\cut_1,\val_1\}$ and $\type_{KM}=\{\cut_{M},\val_{K}\}$.
	\end{itemize}
\end{proof}

\begin{proof}[\textbf{Proof of Lemma \ref{Lemma: u independent of Q}}]
	We prove this lemma together with Proposition \ref{Proposition: type^u} based on the user's payoff.
	Recall that the type-$(\cut,\val)$ user's payoff is 
	\begin{equation} \label{Proof Equ: User payoff}
	\payoff(\dcap,\pcap,\cut,\val) 
	=\val\left[\dmean-\cut A(\dcap)  \right] -\adfee(1-\cut) A(\dcap)-\pcap.
	\end{equation}
		
	Take the derivative of (\ref{Proof Equ: User payoff}) with respect to the user's data valuation $\val$, and we obtain
	\begin{equation}
		\frac{\partial \payoff(\dcap,\pcap,\cut,\val)}{\partial \val} = \dmean-\cut A(\dcap) \ge 0, \ \forall\ \dcap,
	\end{equation}
	which means that the user's payoff increases in the data valuation $\val$.
	
	Take the derivative of (\ref{Proof Equ: User payoff}) with respect to the user's network substitutability $\cut$, and we obtain
	\begin{equation}
		\frac{\partial \payoff(\dcap,\pcap,\cut,\val)}{\partial \cut} = \left[ \adfee-\val \right] A(\dcap),
	\end{equation}
	which means that the user's payoff decreases in $\cut$ if $\val<\adfee$ and increases in $\cut$ if $\val>\adfee$.
	
	Therefore, we have three cases depending on the relation between $\adfee$, $\val_1$, and $\val_K$.
	\begin{itemize}
		\item Fig. \ref{fig: ThetaBeta_low}: The case of $\val_1<\val_{K}<\adfee$ corresponds to $\type_{\epsilon}=\{\cut_1,\val_1\}$.
		\item Fig. \ref{fig: ThetaBeta_cross}: The case of $\val_1<\adfee<\val_{K}$ corresponds to $\type_{\epsilon}=\{\cut_1,\val_1\}$.
		\item Fig. \ref{fig: ThetaBeta_high}: The case of $\adfee<\val_1$ corresponds to $\type_{\epsilon}=\{\cut_{M},\val_1\}$.
	\end{itemize}
	
	Now we have proved the independence of $\type_{\epsilon}$ on the contract item $(\dcap,\pcap)$ in Lemma \ref{Lemma: u independent of Q} and the mapping relation in Proposition \ref{Proposition: type^u}.
\end{proof}

%%%%%%%%%%%%%%%%%%%%%%%%%%%%%%%%%%%%%%%%%%%%%%%%%%%%%%%%%%%%%%%%%%%%%%%%%%%
%%%%%%%%%%%%%%%%%%%%%%%%%%%%%%%%%%%%%%%%%%%%%%%%%%%%%%%%%%%%%%%%%%%%%%%%%%%
%%%%%%%%%%%%%%%%%%%%%%%%%%%%%%%%%%%%%%%%%%%%%%%%%%%%%%%%%%%%%%%%%%%%%%%%%%%
\section{\label{Appendix: Necessary}}
\begin{proof}[\textbf{Proof of Lemma \ref{Lemma: Cap-Price}}]
	We prove this lemma based on the IC condition in Definition \ref{Definiation: Incentive Compatibility}.
	
	First, we prove that if $\dcap_i<\dcap_j$, then $\pcap_i<\pcap_j$.
	For any feasible contract, we have the following IC condition for the type-$\type_i$ user
			\begin{equation}
				L(\dcap_i,\type_i)-\pcap_i\ge L(\dcap_j,\type_i)-\pcap_j,
			\end{equation}
			which is equivalent to
			\begin{equation}\label{Proof Lemma: Cap-Price 1}
				\pcap_i-\pcap_j\le L(\dcap_i,\type_i)-L(\dcap_j,\type_i).
			\end{equation}
			Since $\dcap_i<\dcap_j$, we have $L(\dcap_i,\type_i)-L(\dcap_j,\type_i)<0$, which implies that $\pcap_i<\pcap_j$.
	
	Second, we prove that if $\pcap_i<\pcap_j$, then $\dcap_i<\dcap_j$. 
		 For any feasible contract, we have the following IC condition for the type-$\type_j$ user:
			\begin{equation}
				 L(\dcap_j,\type_j)-\pcap_j\ge L(\dcap_i,\type_j)-\pcap_i,
			\end{equation}
			which is equivalent to
			\begin{equation}\label{Proof Lemma: Cap-Price 2}
				 L(\dcap_j,\type_j)-L(\dcap_i,\type_j)\ge\pcap_j-\pcap_i.
			\end{equation}
			Since $\pcap_j>\pcap_i$, we have $L(\dcap_j,\type_j)-L(\dcap_i,\type_j)>0$, which implies that that $\dcap_j>\dcap_i$.
\end{proof}

\begin{proof}[\textbf{Proof of Lemma \ref{Lemma: type-Cap}}]
	We prove the lemma by contradiction. 
	Assume that the lemma is not true and  there exist $\slope(\dcap,\type_i)>\slope(\dcap,\type_j)$ and $\dcap_i<\dcap_j$  in a feasible contract. 
	
	According to the PIC condition in Definition \ref{Defination: PIC}, for the type-$\type_i$ and type-$\type_j$ users, we have
	\begin{subequations}\label{Proof Equ: type-cap ic}
	\begin{align}
	& L(\dcap_i,\type_i)-\pcap_i \ge L(\dcap_j,\type_i)-\pcap_j,\\
	& L(\dcap_j,\type_j)-\pcap_j \ge L(\dcap_i,\type_j)-\pcap_i.
	\end{align}
	\end{subequations}
	
	Combining  the two inequalities in (\ref{Proof Equ: type-cap ic}), we have
	\begin{equation}\label{Proof Equ: type-cap contradiction}
		L(\dcap_i,\type_i)-L(\dcap_j,\type_i)\ge L(\dcap_i,\type_j)-L(\dcap_j,\type_j),
	\end{equation}
	where $\dcap_i<\dcap_j$.
	
	Next we introduce Claim \ref{Lemma: claimIP}.	
	\begin{claim}\label{Lemma: claimIP}
		For any feasible contract, consider two user types $\type_i$ and $\type_j$ with $\slope(\dcap,\type_i)>\slope(\dcap,\type_j)$, we have
		\begin{equation}\label{Equ: Lemma2_claim}
		L(\dcap_i,\type_i)-L(\dcap_j,\type_i)<L(\dcap_i,\type_j)-L(\dcap_j,\type_j), \forall \dcap_i<\dcap_j.
		\end{equation}
	\end{claim}
	\begin{proof}[\textbf{Proof of Claim \ref{Lemma: claimIP}}]
		We compute the difference of the two sides in (\ref{Equ: Lemma2_claim}) as follows,
		\begin{equation}
		\begin{aligned}
		& L(\dcap_i,\type_i)-L(\dcap_j,\type_i)-L(\dcap_i,\type_j)+L(\dcap_j,\type_j) \\
		=& \int_{\dcap_j}^{\dcap_i}{\frac{\partial L}{\partial \dcap}\bigg|_{(q,\type_i)}}{\rm d}q - \int_{\dcap_j}^{\dcap_i} {\frac{\partial L}{\partial \dcap}\bigg|_{(q,\type_j)} }{\rm d}q \\
		=& \int_{\dcap_j}^{\dcap_i}{\slope(q,\type_i)}{\rm d}q - \int_{\dcap_j}^{\dcap_i} {\slope(q,\type_j) }{\rm d}q \\
		=& \int_{\dcap_j}^{\dcap_i} \left[\slope(q,\type_i)-\slope(q,\type_j)\right]{\rm d}q < 0,
		\end{aligned}
		\end{equation}
		where the last line follows $\slope(\dcap,\type_i)>\slope(\dcap,\type_j)$ for any $\dcap$ and $\dcap_i<\dcap_j$.
	\end{proof}
	
	Based on Claim \ref{Lemma: claimIP}, we can find the contradiction between (\ref{Proof Equ: type-cap contradiction}) and (\ref{Equ: Lemma2_claim}), which proves Lemma \ref{Lemma: type-Cap}.
\end{proof}

%%%%%%%%%%%%%%%%%%%%%%%%%%%%%%%%%%%%%%%%%%%%%%%%%%%%%%%%%%%%%%
%%%%%%%%%%%%%%%%%%%%%%%%%%%%%%%%%%%%%%%%%%%%%%%%%%%%%%%%%%%%%%
\section{\label{Appendix: B}}
\begin{proof}[\textbf{Proof of Lemma \ref{Lemma: IC-Transitivity}}]
	We prove this lemma based on the PIC conditions in Definition \ref{Defination: PIC}.
	\begin{itemize}
		\item According to the PIC conditions, for type-$\type_{i_1}$ and type-$\type_{i_2}$ users, we have the following two inequalities,
		\begin{subequations}\label{Proof Equ: IC_tran ic}
			\begin{align}
			& L(\dcap_{i_1},\type_{i_1})-\pcap_{i_1} \ge   L(\dcap_{i_2},\type_{i_1})-\pcap_{i_2},\\
			& L(\dcap_{i_2},\type_{i_2})-\pcap_{i_2} \ge   L(\dcap_{i_3},\type_{i_2})-\pcap_{i_3}.
			\end{align}
		\end{subequations}
		Based on the necessary conditions in Theorem \ref{Theorem: necessary conditions type} and the relation $i_2<i_3$, we have $\dcap_{i_3}\ge\dcap_{i_2}$. 
		Accordingly, Claim \ref{Lemma: claimIP} indicates
		\begin{equation}\label{Proof Equ: IC_tran IP}
		\begin{aligned}
		& L(\dcap_{i_3},\type_{i_2})-L(\dcap_{i_2},\type_{i_2})  \ge  L(\dcap_{i_3},\type_{i_1})-L(\dcap_{i_2},\type_{i_1}).
		\end{aligned}
		\end{equation}
		Combining the three inequalities in (\ref{Proof Equ: IC_tran ic}) and (\ref{Proof Equ: IC_tran IP}), we obtain
		\begin{equation}\label{Equ: Lemma3 - 1}
		L(\dcap_{i_1},\type_{i_1})-\pcap_{i_1} \ge L(\dcap_{i_3},\type_{i_1})-\pcap_{i_3}.
		\end{equation}
		 
		\item Based on the PIC conditions for type-$\type_{i_3}$ and  type-$\type_{i_2}$ users, we have
		 \begin{subequations}\label{Proof Equ: IC_tran ic '}
		 \begin{align}
		 & L(\dcap_{i_3},\type_{i_3})-\pcap_{i_3} \ge  L(\dcap_{i_2},\type_{i_3})-\pcap_{i_2},\\
		 & L(\dcap_{i_2},\type_{i_2})-\pcap_{i_2} \ge  L(\dcap_{i_1},\type_{i_2})-\pcap_{i_1}.
		 \end{align}
		 \end{subequations}
		 
		 Based on the necessary conditions in Theorem \ref{Theorem: necessary conditions type} and the relation $i_1<i_2$, we have $\dcap_{i_2}\ge\dcap_{i_1}$.
		 Similarly, Claim \ref{Lemma: claimIP} implies that
		 \begin{equation}\label{Proof Equ: IC_tran IP '}
		 \begin{aligned}
		 & L(\dcap_{i_2},\type_{i_3})-L(\dcap_{i_1},\type_{i_3})  \ge  L(\dcap_{i_2},\type_{i_2})-L(\dcap_{i_1},\type_{i_2}).
		 \end{aligned}
		 \end{equation}	
		 
		 Combining the three inequalities in (\ref{Proof Equ: IC_tran ic '}) and (\ref{Proof Equ: IC_tran IP '}), we obtain
		 \begin{equation}\label{Equ: Lemma3 - 2}
		 L(\dcap_{i_3},\type_{i_3})-\pcap_{i_3} \ge L(\dcap_{i_1},\type_{i_3})-\pcap_{i_1}.
		 \end{equation}
		 
	\end{itemize}

	%%%%%
	
	Equations (\ref{Equ: Lemma3 - 1}) and (\ref{Equ: Lemma3 - 2}) indicate that $\contract_{i_1}\ic\contract_{i_3}$.
\end{proof}

\begin{proof}[\textbf{Proof of Lemma \ref{Lemma: IR-Transitivity}}]
	We prove this lemma by showing the following inequalities:
	\begin{equation}
	\begin{split}
	\payoff(\contract_i,\type_i) &\ge \payoff(\contract_{\epsilon},\type_i),\ \forall\ i\ne \epsilon,\\
	&\ge \payoff(\contract_{\epsilon},\type_{\epsilon}),\\
	&\ge 0,
	\end{split}
	\end{equation}
	where the first inequality comes from the IC condition for type-$\type_i$ users, the second inequality is due to the definition of the \textit{smallest-payoff} user type $\type_{\epsilon}$ in (\ref{Equ: smallest-payoff user}), and the last inequality is from the condition $\contract_{\epsilon}\ir$ of this lemma. 
	This completes the proof of Lemma \ref{Lemma: IR-Transitivity}.
\end{proof}

%%%%%%%%%%%%%%%%%%%%%%%%%%%%%%%%%%%%%%%%%%%%%%%%%%%%%%%%%%%%%%%%%%%%%%%%%%
%%%%%%%%%%%%%%%%%%%%%%%%%%%%%%%%%%%%%%%%%%%%%%%%%%%%%%%%%%%%%%%%%%%%%%%%%%
\section{\label{Appendix: Optimal Pricing}}

\begin{proof}[\textbf{Proof of Theorem \ref{Theorem: Optimal Pricing type}}] 
	We prove this theorem by showing that the subscription fees $\{\pcap^*_i,1\le i\le {KM}\}$ specified in (\ref{Equ: Optimal Pricing Policy}) is feasible and optimal.
	
	First, the feasibility of $\{\pcap^*_i,1\le i\le {KM}\}$ is obvious, since $\pcap^*_i$ takes the maximal value satisfying the sufficient conditions in Theorem \ref{Theorem: sufficien condition} for all $i\in\{1,2,...,KM\}$.
	
	Second, we show that the subscription fees $\{\pcap^*_i,1\le i\le {KM}\}$ maximize the profit of the MNO by contradiction.	 
	Assume that this is not true and  there exists another feasible subscription fee assignment $\{\tilde{\pcap}_i,1\le i\le KM\}$ such that 
	\begin{equation}
	\begin{aligned}
	&\sum\limits_{i=1}^{KM}\left[ \tilde{\pcap}_i+P(\dcap_i,\type_i)-c\cdot U(\dcap_i,\type_i)-J(\dcap_i) \right] > \\ 
	& \qquad \sum\limits_{i=1}^{KM} \left[ \pcap^*_i+P(\dcap_i,\type_i)-c\cdot U(\dcap_i,\type_i)-J(\dcap_i) \right],	
	\end{aligned}
	\end{equation}
	which is equivalent to
	\begin{equation} \label{Proof Equ: Optimal price p>p}
	\sum\limits_{i=1}^{KM} \tilde{\pcap}_i > \sum\limits_{i=1}^{KM}  \pcap^*_i. 
	\end{equation}
	
	Equation (\ref{Proof Equ: Optimal price p>p}) implies that there exists at least a $t\in\{1,2,...,KM\}$ such that $\tilde{\pcap}_t >  \pcap^*_t$.
	Recall that $\type_{\epsilon}$ is the smallest user type.
	Next we discuss two cases based on the relation between $t$ and $u$.
	\begin{itemize}
		\item Case I ($t\ge u$):
			Based on the PIC condition for type-$\type_t$ and type-$\type_{t-1}$ users, we have
			\begin{equation}\label{Proof Equ: Optimal price t>u 1}
				\tilde{\pcap}_t \le \tilde{\pcap}_{t-1}+L(\dcap_t,\type_t)-L(\dcap_{t-1},\type_t).
			\end{equation}
			
			Furthermore, Theorem \ref{Theorem: Optimal Pricing type} indicates that
			\begin{equation}\label{Proof Equ: Optimal price t>u 2}
				\pcap^*_t=\pcap^*_{t-1}+L(\dcap_t,\type_t)-L(\dcap_{t-1},\type_t).
			\end{equation}
			
			Combining (\ref{Proof Equ: Optimal price t>u 1}) and (\ref{Proof Equ: Optimal price t>u 2}), we have $\pcap^*_{t-1} < \tilde{\pcap}_{t-1}$.
			Continuing the above process, eventually we obtain  
			\begin{equation}
				\tilde{\pcap}_{\epsilon} > \pcap^*_{\epsilon}=L(\dcap_{\epsilon},\type_{\epsilon}),
			\end{equation}
			which means that $\{\tilde{\pcap}_i,1\le i\le KM\}$ violates the IR condition for the type-$\type_{\epsilon}$ users, hence it is not feasible.
			
		\item Case II ($t < u$): 
			Based on the PIC condition for type-$\type_t$ and type-$\type_{t-1}$ users, we have
			\begin{equation}\label{Proof Equ: Optimal price t<u 1}
				\tilde{\pcap}_t \le\tilde{\pcap}_{t+1} + L(\dcap_t,\type_t)-L(\dcap_{t+1},\type_t).
			\end{equation}
			Furthermore, Theorem \ref{Theorem: Optimal Pricing type} indicates that
			\begin{equation}\label{Proof Equ: Optimal price t<u 2}
				\pcap^*_t = \pcap^*_{t+1} + L(\dcap_t,\type_t)-L(\dcap_{t+1},\type_t).
			\end{equation}	
			
			Combining (\ref{Proof Equ: Optimal price t<u 1}) and (\ref{Proof Equ: Optimal price t<u 2}), we have $\pcap^*_{t+1} < \tilde{\pcap}_{t+1}$.			
			Continuing the above process, eventually we obtain 
			\begin{equation}
				\tilde{\pcap}_{\epsilon} > \pcap_{\epsilon}=L(\dcap_{\epsilon},\type_{\epsilon}),
			\end{equation}
			which indicates that $\{\tilde{\pcap}_i,1\le i\le KM\}$ violates the IR constraint for type-$\type_{\epsilon}$ users, hence it is not feasible.
	\end{itemize}
	
	Both cases above  lead to the contradiction, hence the MNO's profit is maximized by $\{\pcap^*_i,1\le i \le KM\}$.	
	%	======================
	%
	%	Then we prove that $\{\hat{\pcap}(n,k)\}$ is the unique best price assignment by contradiction, and we assume that there exists a price assignment $\{\tilde{\pcap}(n,k)\}\ne\{\hat{\pcap}(n,k)\}$ such that 
	%	\begin{equation}
	%	\sum_{n=1}^{N}  \sum_{k=1}^{K}  \left[ \tilde{\pcap}(n,k)+P(\dcap(n,k),\dmax_n) \right] = \sum_{n=1}^{N}  \sum_{k=1}^{K}  \left[ \hat{\pcap}(n,k)+P(\dcap(n,k),\dmax_n) \right]
	%	\end{equation}
	%	which is equivalent to
	%	\begin{equation}
	%	\sum_{n=1}^{N}  \sum_{k=1}^{K}   \tilde{\pcap}(n,k) = \sum_{n=1}^{N}  \sum_{k=1}^{K}  \hat{\pcap}(n,k)
	%	\end{equation}
	%	
	%	Obviously there exists at least one price $\tilde{\pcap}(n,k)\ne\hat{\pcap}(n,k)$. Without loss of generality, we suppose $\tilde{\pcap}(n,k)<\hat{\pcap}(n,k)$. Then there must be another price $\tilde{\pcap}(t,l)>\hat{\pcap}(t,l)$. Using the same method, we can finally find a $m$ ($1\le m\le t$) such that 
	%	\begin{equation}
	%	\tilde{\pcap}(m,1) > \hat{\pcap}(m,1)=L(\dcap(m,1),\dmax_m,\theta_1)
	%	\end{equation}
	%	which violates the IR constraint for type-($\dmax_m,\theta_1$). Therefore, there is not such a price assignment $\{\tilde{\pcap}(n,k)\}$, which implies that the price assignment $\{\hat{\pcap}(n,k)\}$ given by (\ref{Equ: Optimal Pricing Policy}) is the unique optimal prices.
\end{proof}

\section{Monotonicity Relaxation\label{Appendix: Monotonicity Relaxation}}
To take advantage of the separable structure in the objective function, we can first relax the monotonicity constraints and maximize each $G_i(\cdot)$ over  $\dcap_i$ separately as follows:  
\begin{equation}\label{Equ: Solution saparate}
\tilde{\dcap}^i=\arg\max\limits_{\dcap\in\mathbb{N} } G_i(\dcap) ,\ \forall\ i\in\{1,2,...,KM\}. \\
\end{equation}

If the solution $\{\tilde{\dcap}^i,1\le i\le {KM}\}$ obtained from (\ref{Equ: Solution saparate}) is feasible, i.e., satisfying the monotonicity constraints $\tilde{\dcap}^1\le\tilde{\dcap}^2\le...\le\tilde{\dcap}^{KM}$, then we obtain the optimal solution of Problem \ref{Problem: Optimal cap type}. 
If not, however, we will use the Dynamic Algorithm first proposed in \cite{gao2011spectrum} to adjust the solution $\{\tilde{\dcap}^i,1\le i\le {KM}\}$ to make it feasible and generate a new \emph{adjusted solution} $\{\bar{\dcap}^i,1\le i\le {KM}\}$. 
The intuition behind the Dynamic Algorithm is to first 
\begin{enumerate}
	\item  Find a consecutive infeasible subsequence, e.g., $\tilde{\dcap}^n\ge\tilde{\dcap}^{n+1}\ge...\ge\tilde{\dcap}^{m}$ where $n<m$ and $\tilde{\dcap}^n>\tilde{\dcap}^m$, then generate the \textit{adjusted solution} $\{\bar{\dcap}^i,1\le i\le {KM}\}$ as follows:
	\begin{equation} \label{Equ: adjust}
	\begin{aligned}
	&\bar{\dcap}^i =\left\{
	\begin{aligned}
	& \arg\max_{\dcap\in\mathbb{N}}\sum\limits_{j=n}^{m} G_j(\dcap) &\text{ if}\ n\le i\le m, \\
	& \tilde{\dcap}^i			& \text{otherwise}.
	\end{aligned}
	\right.
	\end{aligned}
	\end{equation}
	\item If the adjusted solution $\{\bar{\dcap}^i,1\le i\le {KM}\}$ is not feasible, then return to Step 1).
	If the adjusted solution $\{\bar{\dcap}^i,1\le i\le {KM}\}$ is feasible, then terminate.
\end{enumerate}

The above two steps run iteratively until there is no infeasible subsequence.
Such an adjusted solution $\{\bar{\dcap}^i,1\le i\le {KM}\}$ must be  feasible.
When Problem \ref{Problem: Optimal cap type} is convex, the adjusted solution produced by the algorithm must be globally optimal as it  satisfies the KKT condition \cite{gao2011spectrum}. 
When Problem \ref{Problem: Optimal cap type} is non-convex, then the adjusted solution is a locally optimal solution  (but may not be globally optimal).

%The above monotonicity relaxation approach is very efficient, since it only needs to deal with some single-variable optimization problems. 
%Therefore, it is widely used in the related contract problems.
%However, such an approach cannot guarantee the global optimality when the problem is not convex.
%Moreover, it is difficult to justify the gap between the obtained locally optimal value and the globally optimal value.
%To obtain the globally optimal solution of Problem \ref{Problem: Optimal cap type} efficiently, we propose the following approach based on dynamic programming.

%%%%%%%%%%%%%%%%%%%%%%%%%%%%%%%%%%%%%%%%%%%%%%%%%%%%%%%%%%%%%%%%%%%%%%%%
%%%%%%%%%%%%%%%%%%%%%%%%%%%%%%%%%%%%%%%%%%%%%%%%%%%%%%%%%%%%%%%%%%%%%%%%
%%%%%%%%%%%%%%%%%%%%%%%%%%%%%%%%%%%%%%%%%%%%%%%%%%%%%%%%%%%%%%%%%%%%%%%%
\section{\label{Appendix: H(n,q)}}

\begin{proof}[\textbf{Proof of Proposition \ref{Proposition: function H Recursiveness}}]
	We prove this proposition based on the definition of Problem \ref{Problem: sub}.
	Recall that $H(n,q)$ is the optimal value of the level-($n,q$) sub-problem, defined as follows:
	\begin{subequations} \label{Proof Equ: H(n,q) 0}
	\begin{align}
	H(n,q)=\max\   & \sum_{i=1}^{n}G_i( \dcap_i ) 					\\
	\textit{s.t. }\ & \dcap_1\le\dcap_2\le...\le\dcap_{n-1}\le\dcap_n\le q, 	\\
	& \dcap_i\in\mathbb{N},\ \forall\ i\in\{1,2,...,n\},\\
	\textit{var. }\ & \dcap_i,\ 1\le i\le n.
	\end{align}
	\end{subequations}
	Here the decision variables are $\dcap_i$ for all $i\in\{1,2,...,n\}$.
	
	For notation clarity, we express $H(n,q)$ as 
	\begin{subequations} \label{Proof Equ: H(n,q) 1}
	\begin{align}
	H(n,q)=\max\   & \sum_{i=1}^{n-1}G_i( \dcap_i ) +G_n( x )		\label{Proof Equ: H(n,q) 1 objective}				\\
	\textit{s.t. }\ & \dcap_1\le\dcap_2\le...\le\dcap_{n-1}\le x, 	\label{Proof Equ: H(n,q) 1 constraint 1}	\\
	& \dcap_i\in\mathbb{N},\ \forall\ i\in\{1,2,...,n-1\}, 			\label{Proof Equ: H(n,q) 1 constraint 2} \\
	& x\le q,\\
	& x\in\mathbb{N},\\
	\textit{var. }\ & x \textit{ and } \dcap_i,\ \forall\ i\in\{1,2,...,n-1\}.
	\end{align}
	\end{subequations}
	
	Furthermore, the optimal value $H(n-1,x)$ of the level-($n-1,x$) sub-problem is 
	\begin{subequations} \label{Proof Equ: H(n-1,x)}
		\begin{align}
		H(n-1,x)=\max\   & \sum_{i=1}^{n-1}G_i( \dcap_i ) 				\label{Proof Equ: H(n-1,x) object}	\\
		\textit{s.t. }\ & \dcap_1\le\dcap_2\le...\le\dcap_{n-1}\le x, 	\label{Proof Equ: H(n-1,x) constraint 1} \\
		& \dcap_i\in\mathbb{N},\ \forall\ i\in\{1,2,...,n-1\},			\label{Proof Equ: H(n-1,x) constraint 2} \\
		\textit{var. }\ & \dcap_i,\forall \ i\in\{1,2,...,n-1\}.
		\end{align}
	\end{subequations}
	
	We note from (\ref{Proof Equ: H(n,q) 1}) and (\ref{Proof Equ: H(n-1,x)}) that
	\begin{itemize}
		\item The first term $\sum_{i=1}^{n-1}G_i( \dcap_i )$ in (\ref{Proof Equ: H(n,q) 1 objective}) is the objective function in (\ref{Proof Equ: H(n-1,x) object}).
		\item The constraint (\ref{Proof Equ: H(n,q) 1 constraint 1}) is the same as (\ref{Proof Equ: H(n-1,x) constraint 1}).
		\item The constraint (\ref{Proof Equ: H(n,q) 1 constraint 2}) is the same as (\ref{Proof Equ: H(n-1,x) constraint 2}).
	\end{itemize}
	
	Therefore, combining (\ref{Proof Equ: H(n,q) 1}) and (\ref{Proof Equ: H(n-1,x)}), we obtain	
	\vspace{20pt}
	\begin{equation} \label{Proof Equ: H(n,q) recussiveness}
		\begin{aligned}
		H(n,q)=\max\   & H(n-1,x) +G_n( x )					\\
		\textit{s.t. }\ & x\le q, 	\\
		& x\in\mathbb{N}, \\
		\end{aligned}
	\end{equation}
	which completes the proof.
%
%For any $2\le n\le KM$ and $0\le q\le \dmax$, $H(n,q)$ has the following recursive relation
%\begin{subequations}
%	\begin{align}
%	H(n,q) = \max\limits_{x\in\mathbb{N}}\ & H(n-1, x ) + G_n( x ) \\
%	\text{s.t. } & x\le q.
%	\end{align}
%\end{subequations}
\end{proof}

\begin{proof}[\textbf{Proof of Proposition \ref{Proposition: function H Monotonicity}}]
	We prove the two properties of $H(n,q)$ based on the definition in Problem \ref{Problem: sub}.
	
	First, it is easy to see that  $H(n,q)$ is non-decreasing in $q$, since the parameter $q$ in Problem \ref{Problem: sub} represents the upper bound of the feasible domain.
	
	Next we prove the existence of the critical point $\hat{q}_n$.
	Recall that $H(n,\dmax)$ and $\bm{\dcap}^\star(n,\dmax)=\{\dcap^\star_i(n,\dmax),1\le i\le n\}$ represent the optimal value and the optimal solution of the level-($n,\dmax$) sub-problem, respectively.
	Then the $n$-th element of the optimal solution $\bm{\dcap}^\star(n,\dmax)$ is the critical point $\hat{q}_n$, i.e., 
	\begin{equation}\label{Proof Equ: critical point}
	\hat{q}_n \triangleq \dcap_n^\star(n,\dmax).
	\end{equation}
	In this case, we have
	\begin{equation}
	H(n,q)=H(n,\dmax),\ \forall q\ge \hat{q}_n.
	\end{equation}	
\end{proof}

\begin{proof}[\textbf{Proof of Theorem \ref{Theorem: solution KM,D}}]
	We prove this theorem based on the definition of the level-($n,q$) sub-problem (in Problem \ref{Problem: sub}) and the critical point  (\ref{Proof Equ: critical point}).
	
	According to Problem \ref{Problem: sub}, we have the following level-($KM,\dmax$) sub-problem, which is equivalent to Problem \ref{Problem: Optimal cap type}.
	\begin{subequations} \label{Proof Equ: sub KM}
		\begin{align}
		\max\   & \sum_{i=1}^{KM}G_i( \dcap_i ) 						\\
		\textit{s.t. }\ & \dcap_1\le\dcap_2\le...\le\dcap_{KM}\le \dmax 	\\
		& \dcap_i\in\mathbb{N},\ \forall\ i\in\{1,2,...,KM\} \\
		\textit{var: }	& \dcap_i,1\le i\le KM.
		\end{align}
	\end{subequations}	
	Next we explain the optimal solution $\{\dcap^\star_i(KM,\dmax),1\le i \le KM\}$ of the above level-($KM,\dmax$) sub-problem (\ref{Proof Equ: sub KM}).
	
	First, for the $KM$-th element $\dcap^\star_{KM}(KM,\dmax)$, according to the definition of the critical point (\ref{Proof Equ: critical point}), we have
	\begin{equation}
		\dcap^\star_{KM}(KM,\dmax) = \hat{q}_{KM}.
	\end{equation}
	
%	Before we explain the $i$-th ($1\le i \le KM-1$) element $\dcap^\star_{i}(KM,\dmax)$, we first present the following lemma
%	\begin{lemma}\label{Lemma: H and Problem solution}
%	Denote $\{\dcap^\star_i(KM,\dmax),1\le i \le KM\}$ the optimal solution of the level-($KM,\dmax$) sub-problem.
%	We have 
%	\begin{equation}
%	\begin{aligned}
%	H(KM,\dmax) = H\left(n-1,{\dcap}^\star_{n}\right)+\sum_{i=n}^{KM}G_i\left({\dcap}^\star_i\right),\ \forall n\ge2.
%	\end{aligned}
%	\end{equation}
%	\end{lemma}	
%	\begin{proof}[\textbf{Proof of Lemma \ref{Lemma: H and Problem solution}}]
%		We prove Lemma \ref{Lemma: H and Problem solution} by applying Proposition \ref{Proposition: function H Recursiveness}.
%		Recall that $H(n,q)$ has the recursiveness in (\ref{Equ: H Recursiveness}), 
%	\end{proof}

	Substituting $\{\dcap^\star_i(KM,\dmax),1\le i \le KM\}$ into (\ref{Equ: H Recursiveness}), we obtain
	\begin{equation}
	\begin{aligned}
	H(KM,\dcap_{KM}^\star) = H\left(KM-1,\dcap_{KM}^\star\right)+G_{KM}\left( \dcap_{KM}^\star \right),
	\end{aligned}
	\end{equation}		
	where $H\left(KM-1,\dcap_{KM}^\star\right)$ is the optimal value of the level-($KM-1,\dcap_{KM}^\star$) sub-problem.
	Based on the non-decreasing property and the critical point in Proposition \ref{Proposition: function H Monotonicity}, we know that
	\begin{equation}
		\dcap_{KM-1}^\star(KM,\dmax) = \min\{ \hat{q}_{KM-1},\dcap_{KM}^\star(KM,\dmax) \}.
	\end{equation}
	Continuing the above process, we can show that for all $i\in\{1,2,...,KM-1\}$, we have
	\begin{equation}
	\dcap_i^\star(KM,\dmax) = 
	\min\{ \hat{q}_i,\dcap_{i+1}^\star(KM,\dmax)\},
	\end{equation}	
	which completes the proof.
\end{proof}

%%%%%%%%%%%%%%%%%%%%%%%%%%%%%%%%%%%%%%%%%%%%%%%%%%%%%%%%%%%%%%%%%%%%%%%%
%%%%%%%%%%%%%%%%%%%%%%%%%%%%%%%%%%%%%%%%%%%%%%%%%%%%%%%%%%%%%%%%%%%%%%%%
%%%%%%%%%%%%%%%%%%%%%%%%%%%%%%%%%%%%%%%%%%%%%%%%%%%%%%%%%%%%%%%%%%%%%%%%
\section{A Numerical Example for the DQA Algorithm\label{Appendix: Numerical example}}
Next we provide  a numerical example to demonstrate the computation of  the optimal data caps based on $H(n,q)$.

For illustration simplicity, we let $KM=4$ and $\dmax=9$.
Accordingly, we denote $\dcap^*_1$, $\dcap^*_2$, $\dcap^*_3$, and $\dcap^*_4$ as the optimal data caps in Problem \ref{Problem: Optimal cap type}.

As shown in Fig. \ref{fig: DP_illustration}, the horizontal axis represents $q\in\{0,1,...,9\}$,
the vertical axis represents $n\in\{1,2,3,4\}$,
and the value in each box represents $H(n,q)$.
In this numerical example, the optimal value of Problem \ref{Problem: Optimal cap type} is $H(4,9)=90$.
Furthermore, we will show  that the optimal data caps of Problem \ref{Problem: Optimal cap type} is $\{\dcap_1^*,\dcap_2^*,\dcap_3^*,\dcap_4^*\}=\{3,5,5,7\}$, as follows:
\begin{itemize}
	\item For $\dcap_4^*=7$:  
	In Fig. \ref{fig: DP_illustration}, $H(4,8)=H(4,9)$ indicates that the optimal data cap $\dcap_4^*$ is no larger than the domain upper bound $8$, i.e., $\dcap_4^*\le8$.
	Otherwise, $H(4,8)$ must be smaller than $H(4,9)$.
	Similarly, 	$H(4,7)=H(4,8)$ implies $\dcap_4^*\le7$ as well.
	However, $H(4,6)<H(4,7)$ reveals that $\dcap_4^*=7$.
	Otherwise, if $\dcap_4^*<7$, then we would have $H(4,6)=H(4,7)$.
	Furthermore, according to Proposition \ref{Proposition: function H Recursiveness}, we have
	\begin{equation}
	\underbrace{H(4,9)}_{90} = \underbrace{H(3,\dcap_4^*)}_{87}  +G_4(\dcap_4^*),
	\end{equation}
	which shows that $G_4(\dcap_4^*)=3$ according to Fig. \ref{fig: DP_illustration}.
	
	\item For $\dcap_3^*=5$: 
	Given $\dcap_4^*=7$, we know that $\dcap_3^*\le\dcap_4^*=7$ considering the monotonic constraints (\ref{SubEqu: Optimal cap type, monotonicity}) in Problem \ref{Problem: Optimal cap type}.
	Similarly, the equality $H(3,6)=H(3,7)$ implies $\dcap_3^*\le6$ and $H(3,5)=H(3,6)$ implies $\dcap_3^*\le5$.
	Then the inequality $H(3,4)<H(3,5)$ implies $\dcap_3^*=5$.
	Furthermore, according to Proposition \ref{Proposition: function H Recursiveness}, we have 
	\begin{equation}
	\underbrace{H(4,9)}_{90} = \underbrace{H(2,\dcap_3^*)}_{40} + G_3(\dcap_3^*) + \underbrace{G_4(\dcap_4^*)}_{3},
	\end{equation}
	which indicates that $G_3(\dcap_3^*)=47$ according to Fig. \ref{fig: DP_illustration}.
	
	\item For $\dcap_2^*=5$:
	Given $\dcap_3^*=5$, we know $\dcap_2^*\le\dcap_3^*=5$ considering the monotonic constraints.
	Similar to the above argument, the inequality $H(2,4)<H(2,5)$ implies $\dcap_2^*=5$.
	Moreover, we have the following equality
	\begin{equation}
	\underbrace{H(4,9)}_{90} = \underbrace{H(1,\dcap_2^*)}_{35} + G_2(\dcap_2^*) + \underbrace{G_3(\dcap_3^*)}_{47} + \underbrace{G_4(\dcap_4^*)}_{3},
	\end{equation}	
	which leads to $G_2(\dcap_2^*)=5$.
	
	\item For $\dcap_1^*=3$:
	Given $\dcap_2^*=5$, we know $\dcap_1^*\le\dcap_2^*=5$.
	The inequality and equalities $H(1,2)<H(1,3)=H(1,4)=H(1,5)$ implies $\dcap_1^*=3$ and $G_1(\dcap_1^*)=35$.
	%	Therefore, we do not need to consider the two gray boxes marked number $87$.
	
\end{itemize}

\begin{figure} 
	\centering
	\includegraphics[width=0.9\linewidth]{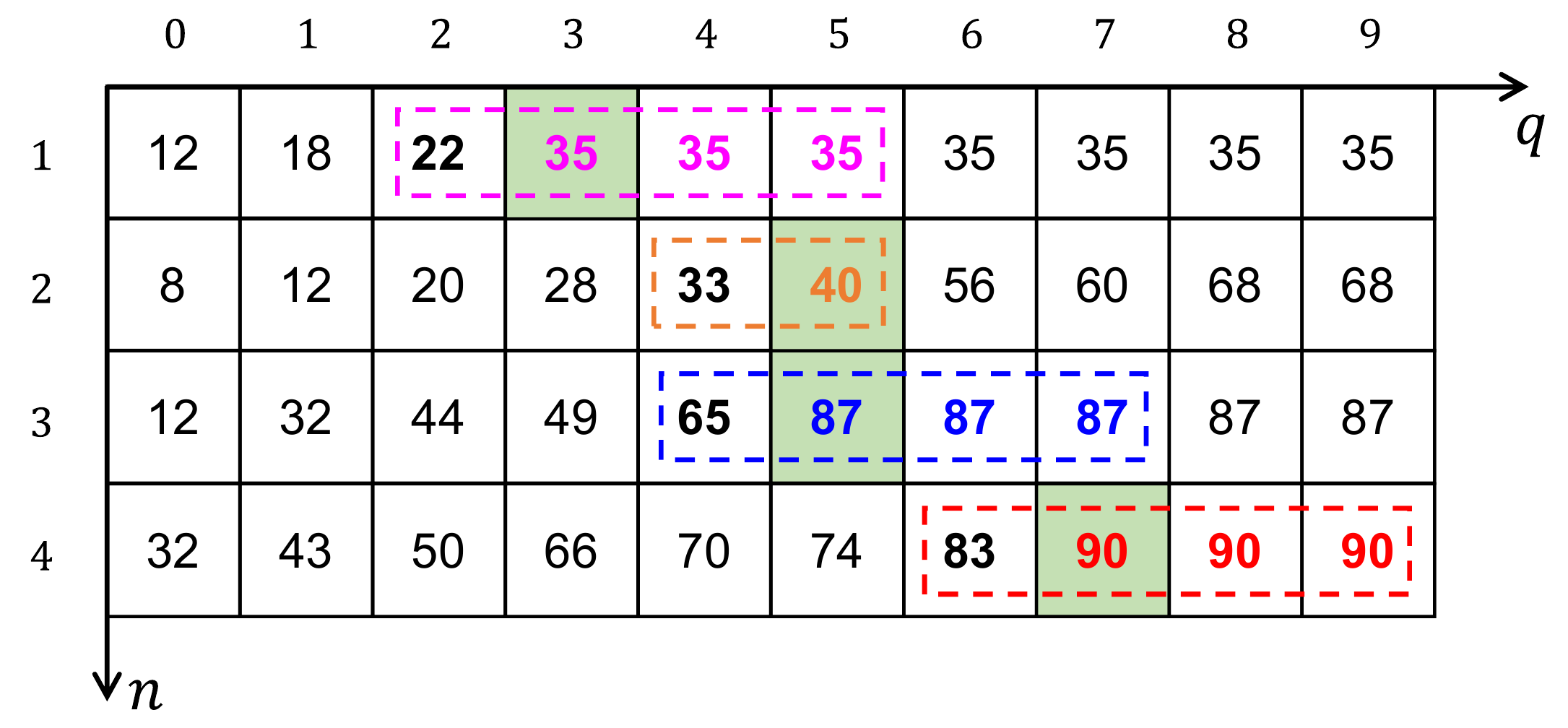}
	\caption{An example of computing $\{\dcap_1^*,\dcap_2^*,\dcap_3^*,\dcap_4^*\}$ based on the table of $H(n,q)$ for all $1\le n\le 4$ and $0\le q \le 9$.}
	\label{fig: DP_illustration}
\end{figure}

%To summarize up, based on the table of $H(n,q)$ in Fig. \ref{fig: DP_illustration}, we can conclude that
%the optimal solution of Problem \ref{Problem: Optimal cap type} is $\{\dcap_1^*,\dcap_2^*,\dcap_3^*,\dcap_4^*\}=\{3,5,5,7\}$.

%%%%%%%%%%%%%%%%%%%%%%%%%%%%%%%%%%%%%%%%%%%%%%%%%%%%%%%%%%%%%
%%%%%%%%%%%%%%%%%%%%%%%%%%%%%%%%%%%%%%%%%%%%%%%%%%%%%%%%%%%%%
%%%%%%%%%%%%%%%%%%%%%%%%%%%%%%%%%%%%%%%%%%%%%%%%%%%%%%%%%%%%%
\section{\label{Appendix: simulation}}
Similar as Section \ref{Subsubsection: Impact of MNO's Costs}, we take the data mechanism $\mechanism=1$ as an example to investigate how the MNO's capacity cost $z$ affects the optimal contract items.

Fig. \ref{fig: impact of z} shows the impact of the MNO's capacity cost $z$.
Specifically, there are a total of seven different contract items in the optimal contract.
The seven curves in Fig. \ref{fig: z_Cap} represent the corresponding  different data caps.
We note that the optimal data caps (except the zero cap) decrease in the MNO's capacity cost.
Fig. \ref{fig: z_Subscription} plots the corresponding subscription fees in the optimal contract.
We find that
\begin{itemize}
	\item The subscription fee of the zero-cap contract item (i.e., the bottom blue circle curve in Fig. \ref{fig: z_Subscription}) does not change in the capacity cost $z$.
	This results from the individual rationality condition in (\ref{Equ: Optimal Pricing Policy 1}).
	
	\item The subscription fees of small-cap contract items (e.g., the cross and triangle curves in Fig. \ref{fig: z_Subscription}) decrease in the MNO's costs.
	While the subscription fees of the large-cap contract item (e.g., the diamond curves in Fig. \ref{fig: z_Subscription}) increases in the MNO's costs.
	Therefore, the large-cap contract item becomes less economical to the users (in terms of the average price $\pcap/\dcap$) as the MNO's capacity cost increases.
	That is, the profit-maximizing MNO tends to compensate its capacity cost by charging those users who are willing to pay for the large-cap contract item.
\end{itemize}

\begin{figure}
	\centering
	\subfigure[Optimal data caps.]{\label{fig: z_Cap}{\includegraphics[height=0.4\linewidth]{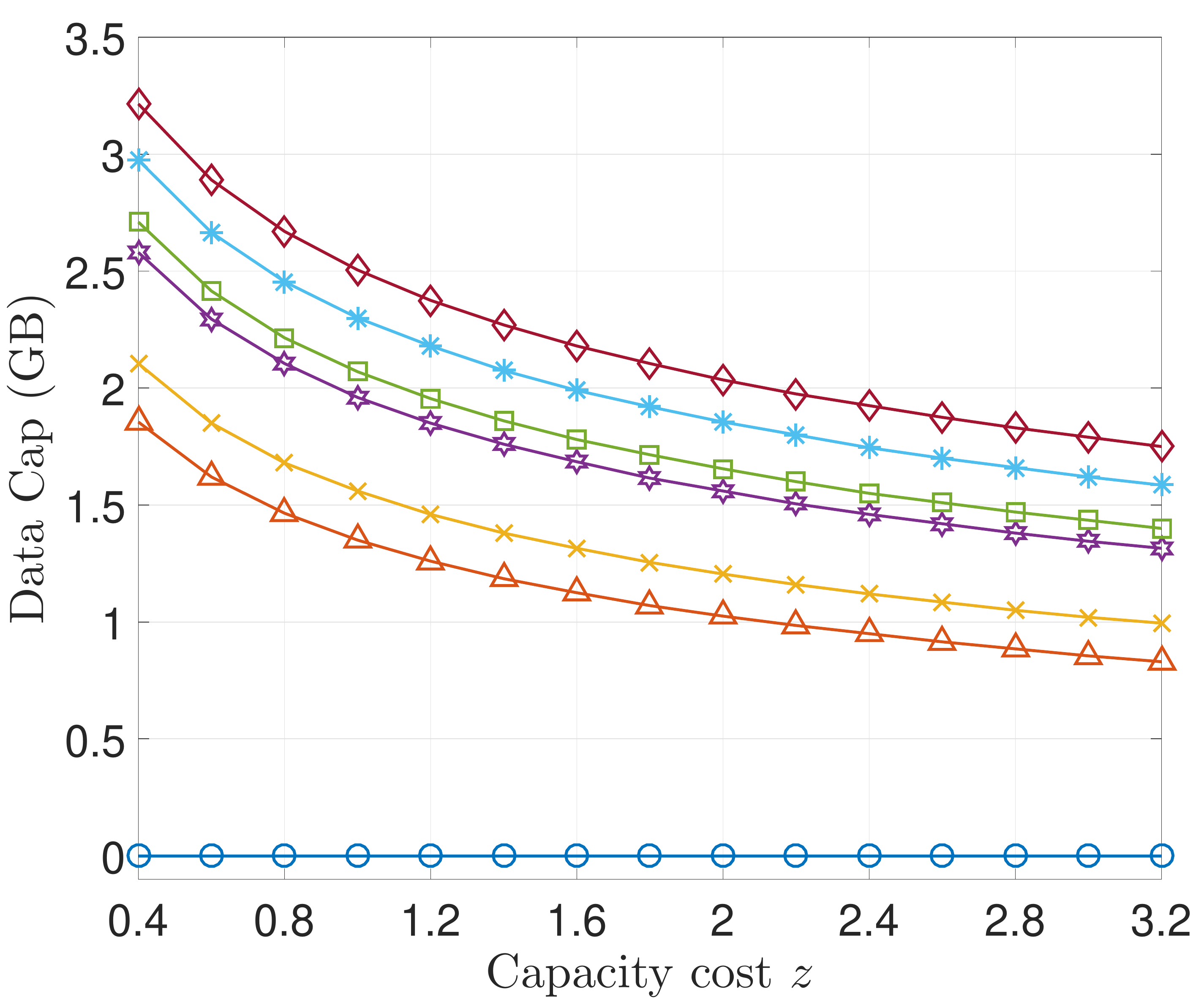}}}\
	\subfigure[Optimal subscription fees.]{\label{fig: z_Subscription}{\includegraphics[height=0.4\linewidth]{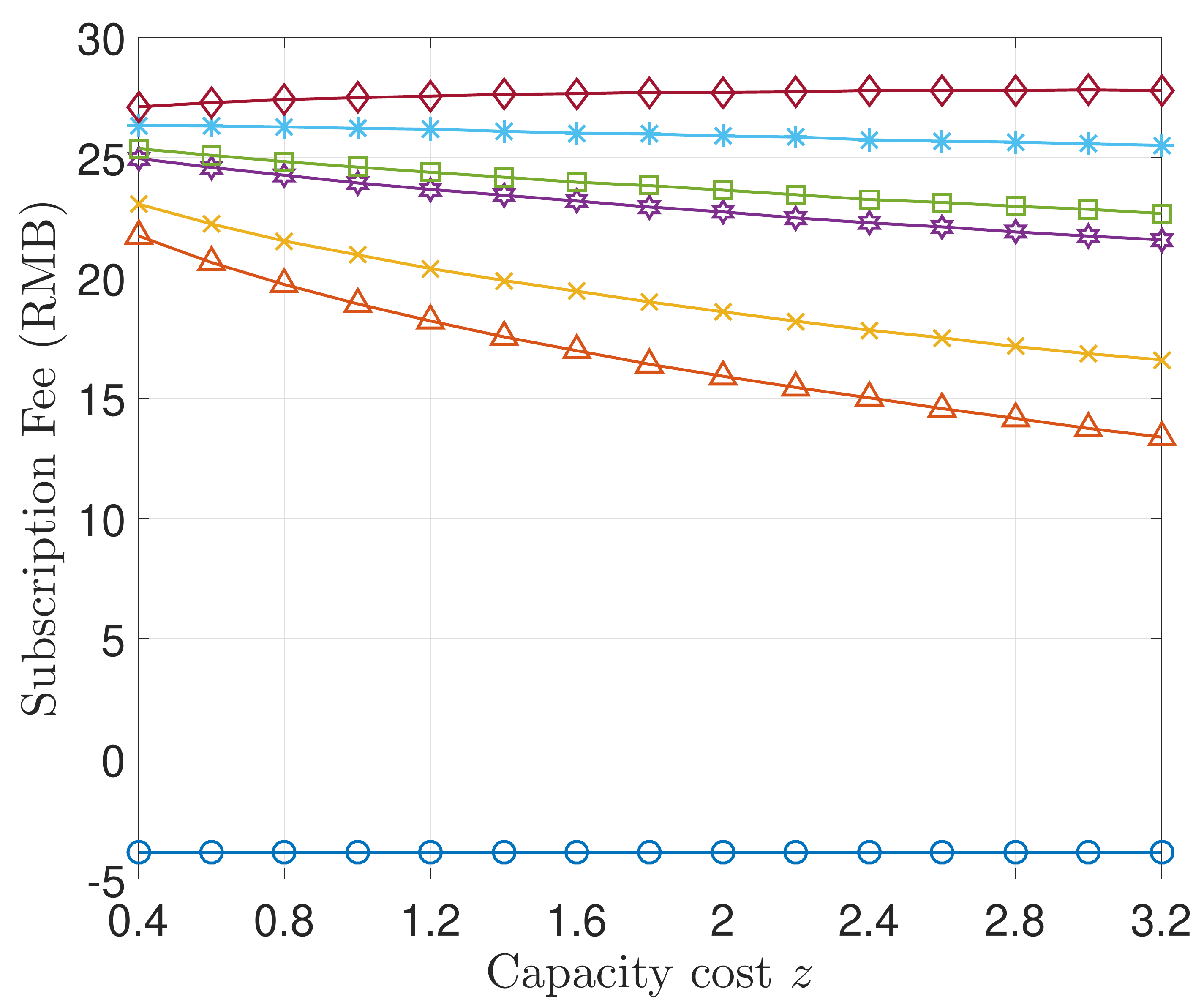}}}
	\caption{Impact of MNO's capacity cost $z$.}
	\label{fig: impact of z}
\end{figure}
\begin{figure}
	\centering
	\subfigure[MNO's profit.]{\label{fig: z_Profit}{\includegraphics[width=0.49\linewidth]{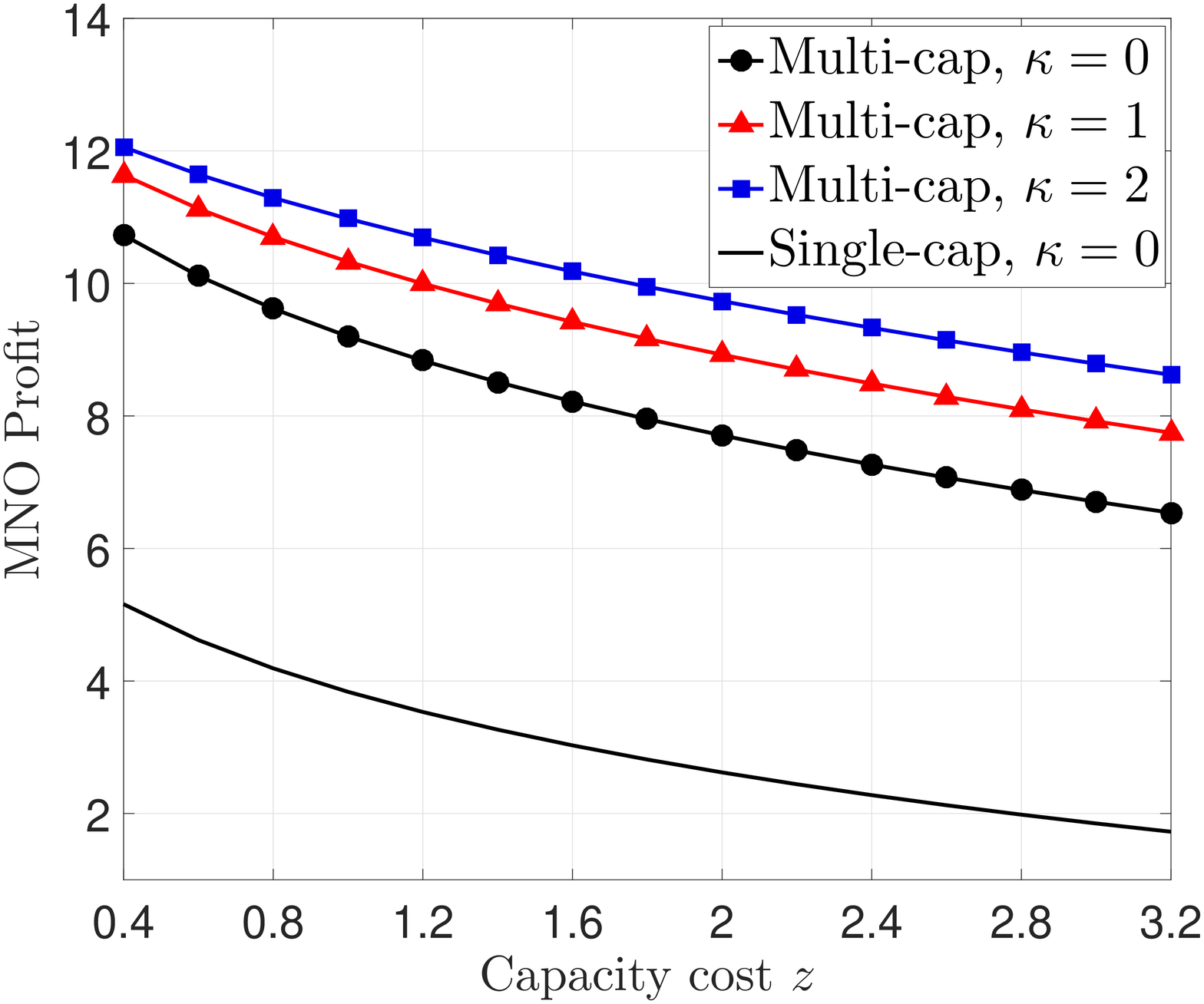}}}
	\subfigure[All users' payoff.]{\label{fig: z_Payoff}{\includegraphics[width=0.49\linewidth]{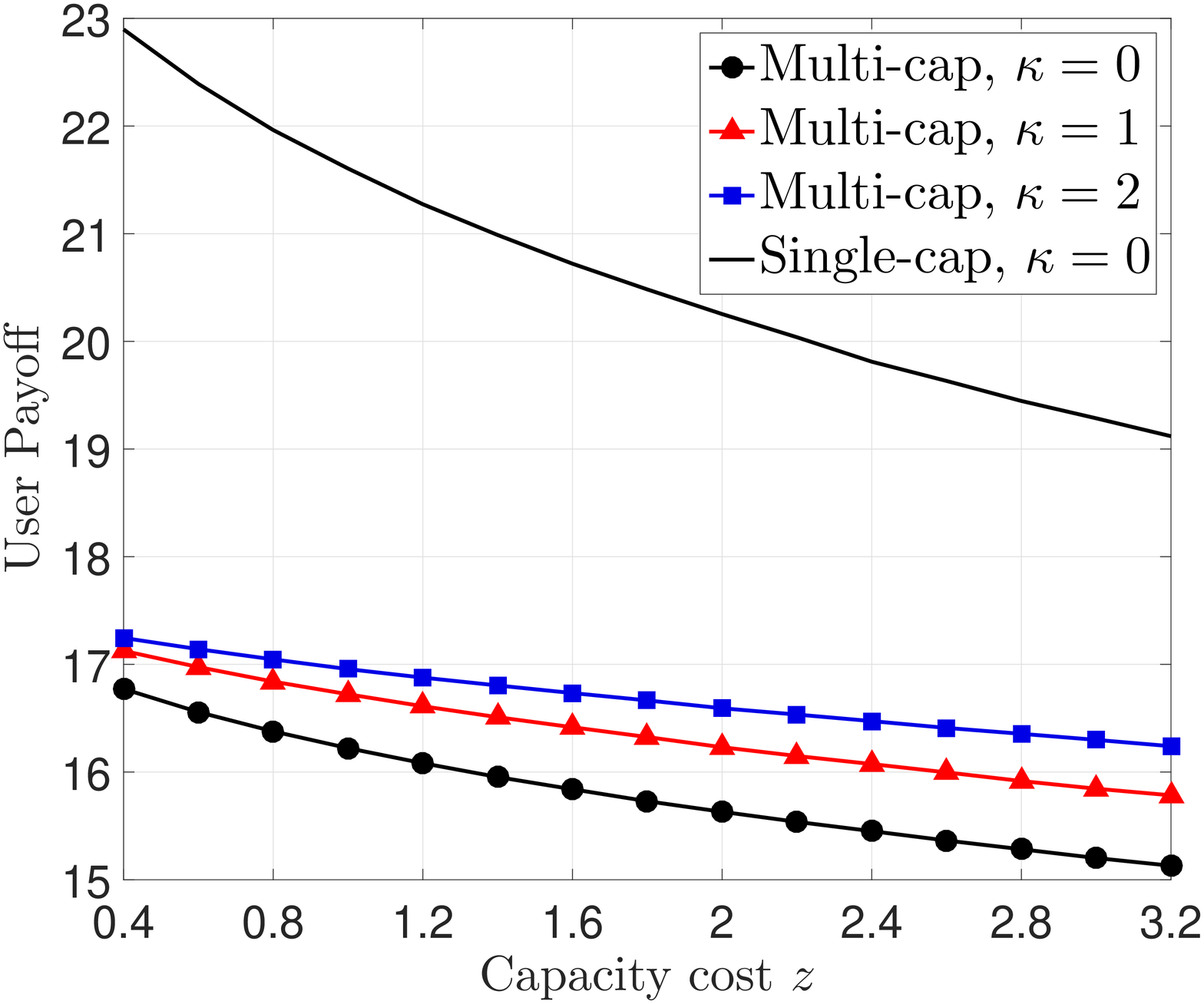}}}
	\caption{Impact of the MNO's capacity cost $z$.\vspace{-5pt}}
	\label{fig: z profit payoff}
\end{figure}

Next we evaluate the MNO's profit and user's payoffs in the four scenarios of Table \ref{table: four cases} under different capacity cost.
In Fig. \ref{fig: z profit payoff}, the horizontal axises in the two sub-figures represent the MNO's marginal capacity cost.
\begin{itemize}
	\item Fig. \ref{fig: z_Profit} plots MNO's profits in the four scenarios.
	Overall, the MNO's profits decrease in its capacity cost $z$.
	By comparing single-cap traditional pricing benchmark and the multi-cap traditional pricing scheme, we note that the price discrimination under our optimal contract can significantly increase the MNO's profit (176\% on average).
	By comparing the three multi-cap curves, we find that the MNO obtains a higher profit under a more time-flexible data mechanism.
	Specifically, compared with Scenario (ii) (i.e., the circle curve), MNO's profits increases by 12\% on average in Scenario (iii) (i.e., the triangle curve) and 23\% on average in Scenario (iv) (i.e., square curve).
	This implies that under the multi-cap scheme, offering a better time flexibility can further improve the MNO's profit.
	
	\item Fig. \ref{fig: z_Payoff} plots the users' total expected payoff in four scenarios.
	First, we observe that users' payoff decreases in the MNO's capacity cost.
	By comparing the single-cap traditional pricing benchmark and the multi-cap traditional pricing scheme, we notice that the price discrimination under our optimal contract reduces users' expected payoff (20\% on average), which means that the MNO captures more consumer surplus through the price discrimination.
	Comparing the three multi-cap schemes, we find that the time-flexible data mechanisms can improve the users' payoff. 
	Specifically, compared with Scenario (ii) (i.e., the circle curve), users' payoff increases by 3.1\% on average in Scenario (iii) (i.e., the triangle curve) and 5.2\% on average in Scenario (iv) (i.e., square curve).
%	Hence a better time flexibility leads to a higher expected payoff for the user market.
\end{itemize}

%%%%%%%%%%%%%%%%%%%%%%%%%%%%%%%%%%%%%%%%%%%%%%%%%%%%%%%%%%%%%
%%%%%%%%%%%%%%%%%%%%%%%%%%%%%%%%%%%%%%%%%%%%%%%%%%%%%%%%%%%%%
%%%%%%%%%%%%%%%%%%%%%%%%%%%%%%%%%%%%%%%%%%%%%%%%%%%%%%%%%%%%%
\section{\label{Appendix: type}}
In Section \ref{Section: Numerical Results}, we cluster the empirical data valuation and network substitutability into four groups, respectively.
We then proceed the contract design based on the method in Section \ref{Section: Contract Optimality} for a total of sixteen user types.
Next we investigate the impact of the number of clustered user types on the performance of the optimal contract.

\begin{table}
	\setlength{\abovecaptionskip}{3pt}
	\setlength{\belowcaptionskip}{0pt}
	\renewcommand{\arraystretch}{1}
	\caption{Four cases of user type cluster. }
	\label{table: Four cases of user type cluster}
	\centering
	\begin{tabular}{p{0.55cm}p{3.4cm}p{3.3cm}}
		\toprule
		Case			& $\ $Data valuation					&$\ $ Network substitutability \\
		\midrule
		$2\times2$		&$\{24.6,63.4\}$						& $\{0.60, 0.87\}$ \\
		$3\times3$		&$\{18.3, 44.4, 80.7\}$					& $\{0.48, 0.72, 0.90\}$ \\
		$4\times4$		&$\{16.2, 36.1, 61.9, 96.3\}$			& $\{0.51, 0.71, 0.84, 0.95\}$ \\
		$5\times5$		&$\{14.4, 32.2, 48.9, 76.1, 103.9\}$	& $\{0.40,0.58,0.71,0.82,0.94\}$ \\
		\bottomrule
	\end{tabular}
\end{table}

\begin{figure}
	\centering
	\subfigure[MNO's profit.]{\label{fig: Type_Profit}{\includegraphics[width=0.49\linewidth]{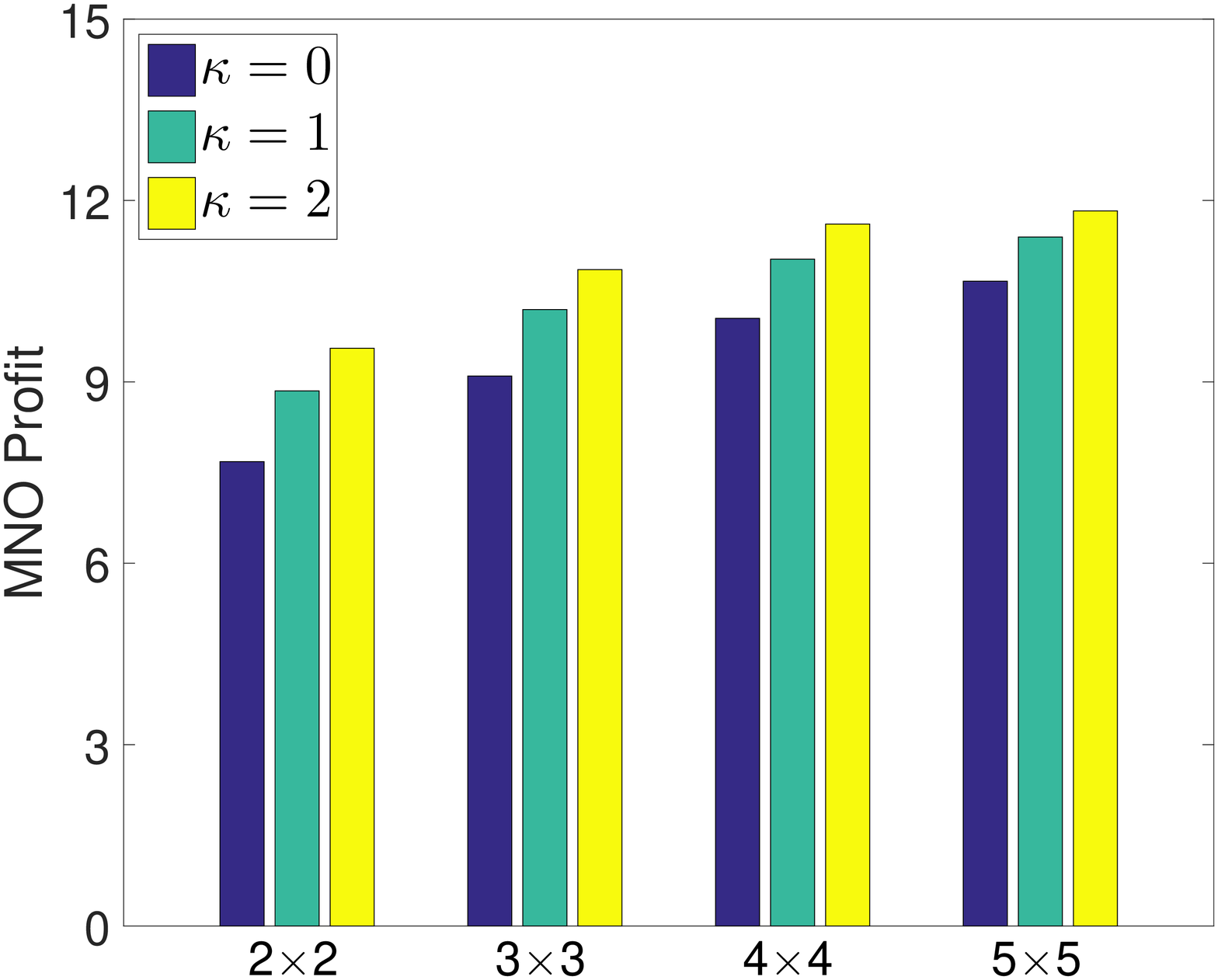}}}
	\subfigure[All users' expected payoff.]{\label{fig: Type_Payoff}{\includegraphics[width=0.49\linewidth]{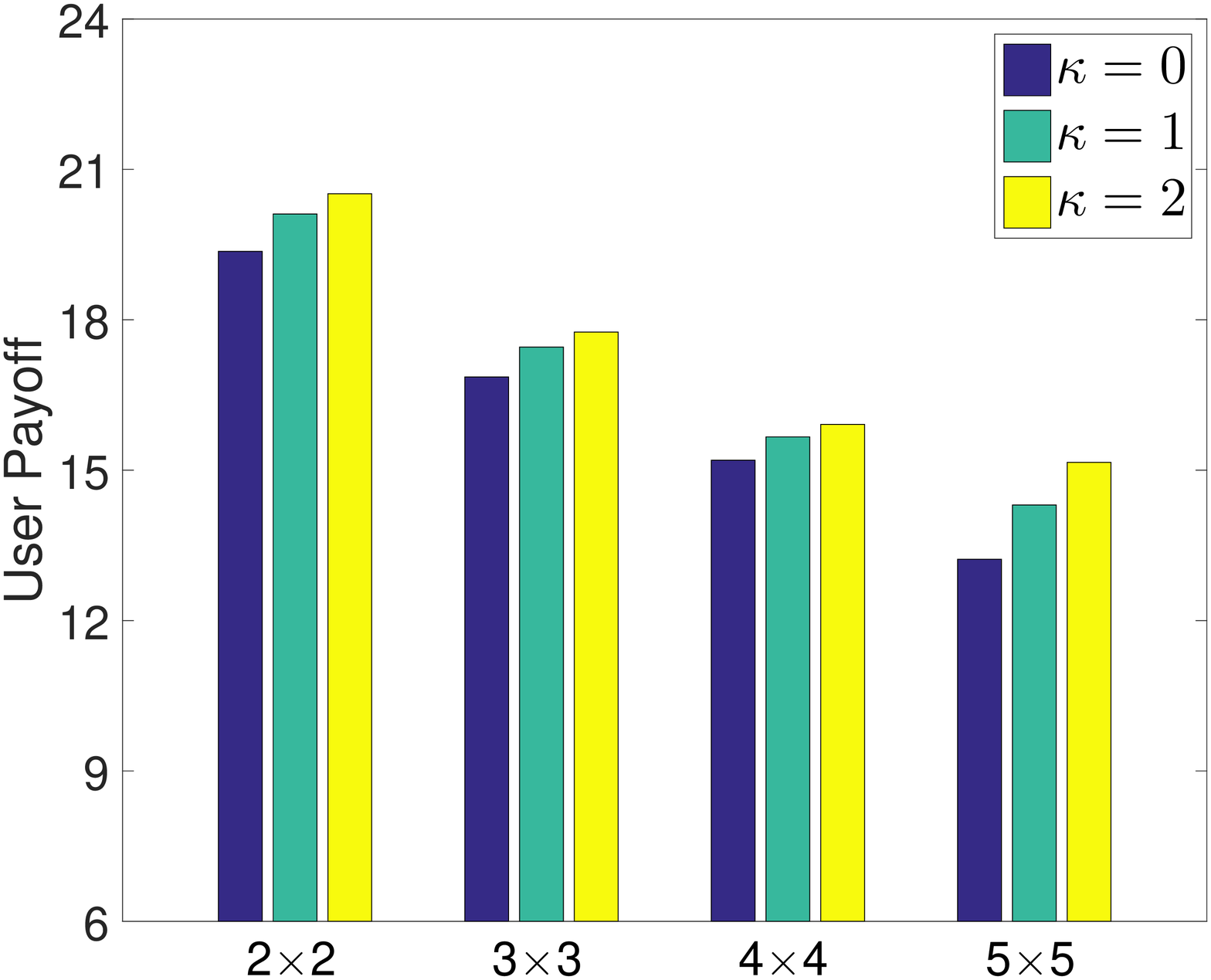}}}
	\caption{Impact of the number of clustered user types.}
	\label{fig: type}
\end{figure}

We will compare four scenarios, where the users' each characteristic is clustered into two groups, three groups, four groups, and five groups, respectively. 
Table \ref{table: Four cases of user type cluster} shows the corresponding mean values of the users types in the four cases.
Fig. \ref{fig: type} shows the performance of the optimal contract in the four cases under three data mechanisms.
Here we let the operational cost be $c=6$RMB/GB and the capacity cost be $z=1.3$RMB/GB.
\begin{itemize}
	\item Fig. \ref{fig: Type_Profit} plots the MNO's profits in the four cases under three different data mechanisms.
	Overall, the MNO's profit increases in the number of clustered user types.
	Moreover, the time-flexible data mechanism can further increase the MNO's profit given the number of user types.
	
	\item Fig. \ref{fig: Type_Payoff} plots all users' average payoff in the four cases under three data mechanisms.
	Overall, all users' average payoff decreases in the number of user types considered by the MNO.
	But a more time-flexible data mechanism increases the users' payoff given the number of user types.
\end{itemize}

Based on the above discussion, we conclude that when the MNO divides the users into more types, the MNO's profit increases but the users' average payoff decreases.
It means that a finer granularity  price discrimination reduces the consumer surplus.
In all cases,  the time-flexible data mechanisms can increase both the MNO's profit and users' payoff, leading to a win-win situation.

\end{document}